\documentclass[11pt, letterpaper]{article}
\title{Optimal Testing of Generalized Reed-Muller Codes in Fewer Queries}
\author{
Dor Minzer\thanks{Department of Mathematics, Massachusetts Institute of Technology, Cambridge, USA. Supported by a Sloan Research
Fellowship, NSF CCF award 2227876 and 
NSF CAREER award 2239160.}
\and
Kai Zheng\thanks{Department of Mathematics, Massachusetts Institute of Technology, Cambridge, USA. Supported by the NSF GRFP.}
}
\date{\vspace{-5ex}}
\usepackage{fullpage}
\usepackage{amsthm}
\usepackage{amsmath,amssymb,amsfonts,nicefrac}
\usepackage{xspace}
\usepackage{color}
\usepackage{url}
\usepackage{hyperref}
\usepackage{enumitem}
\usepackage{bm}
\usepackage{bbm}
\usepackage{times}
\usepackage{mathtools}
\usepackage{enumitem}

\newcommand{\AffBilin}{{\sf AffBilin}}
\newcommand{\AffShort}{{\sf AffBilin}}
\newcommand{\AffGras}{{\sf AffGras}}

\newcommand{\Lift}{{\sf Lift}}

\DeclareMathOperator{\rej}{rej}

\DeclareMathOperator{\poly}{poly}
\DeclareMathOperator{\spa}{span}

\DeclareMathOperator{\RM}{RM}

\DeclareMathOperator{\im}{Im}

\DeclareMathOperator{\Res}{Res}
\DeclareMathOperator{\lin}{lin}
\DeclareMathOperator{\fl}{fl}

\newcommand{\E}{\mathop{\mathbb{E}}}
\newcommand{\Ff}{\mathbb{F}}
\newcommand{\D}{\mathcal{D}}

\newcommand{\Ac}{\mathcal{A}}
\newcommand{\Sc}{\mathcal{S}}

\newcommand{\C}{\mathcal{C}}
\newcommand{\V}{\mathcal{V}}

\newcommand{\T}{\mathcal{T}}

\newcommand{\U}{\mathcal{U}}

\newcommand{\eps}{\varepsilon}
\renewcommand{\epsilon}{\eps}

\newcommand\inner[2]{\langle{#1},{#2}\rangle}
\newcommand\skipi{{\vskip 10pt}}
\renewcommand\leq{\leqslant}
\renewcommand\geq{\geqslant}

\theoremstyle{plain}
\newtheorem{theorem}{Theorem}[section]
   \newtheorem{thm}{Theorem}[section]
   
   \newtheorem{lemma}[thm]{Lemma}

   \newtheorem{remark}[thm]{Remark}
   \newtheorem{definition}{Definition}

    \newtheorem{proposition}[thm]{Proposition}

\DeclareMathOperator\supp{supp}
\DeclarePairedDelimiter{\ceil}{\lceil}{\rceil}
\begin{document}
\maketitle
\begin{abstract}
A local tester for an error correcting code $C\subseteq \Sigma^{n}$ is a tester that makes $Q$ oracle queries to a given word $w\in \Sigma^n$ and decides to
accept or reject the word $w$. An optimal local tester is a local tester that has the additional properties of completeness and optimal soundness. By completeness, we mean that 
the tester must accept with probability $1$ if $w\in C$.
By optimal soundness, we mean that if the tester accepts with probability 
at least $1-\eps$ (where $\eps$ is small), then it must be the case that $w$ is $O(\eps/Q)$-close to some codeword $c\in C$ in Hamming distance.

We show that Generalized Reed-Muller codes admit optimal testers with $Q = (3q)^{\ceil{\frac{d+1}{q-1}}+O(1)}$ queries.
Here, for a prime power $q = p^{k}$, the Generalized Reed-Muller code, $\RM[n,q,d]$, consists of the evaluations of all $n$-variate degree $d$ polynomials over $\Ff_q$.
Previously, no tester achieving this query complexity was known, and the best known testers due to
Haramaty, Shpilka and Sudan~\cite{HSS}(which is optimal) and due to Ron-Zewi and Sudan~\cite{RZS}(which was not known to be optimal) both required $q^{\ceil{\frac{d+1}{q-q/p}}}$
queries. Our tester achieves query complexity which is polynomially better than by a power of $p/(p-1)$, which is 
nearly the best query complexity possible 
for generalized Reed-Muller codes.

The tester we analyze is due to Ron-Zewi and Sudan, and we show that their basic tester is in fact optimal. Our 
methods are more general and also allow us 
to prove that a wide class of 
testers, which follow the form of the Ron-Zewi and Sudan tester, are optimal. This result applies to testers for all affine-invariant codes (which are not necessarily generalized Reed-Muller codes).
\end{abstract}

\section{Introduction}
\subsection{Local Testing of Reed Muller Codes}
The Reed-Muller Code is a widely used code with many applications in complexity theory, and more broadly in theoretical computer science.
One reason for this is that the Reed-Muller code enjoys very good local testability properties which are crucial in many applications (for example, in the construction
of probabilistically checkable proofs). The primary goal of this paper is to present local testers for Reed-Muller codes over extension fields with improved query 
complexity, which additionally satisfy a stronger notion of soundness known as \emph{optimal testing}.

Throughout this paper, $p\in\mathbb{N}$ is a prime and $q = p^k$ is a prime power, where $k$ should be thought of as moderately large (so that $q$ is significantly larger than $p$). For a degree parameter $d\in\mathbb{N}$
and a number of variables parameter $n\in\mathbb{N}$, the Reed-Muller code consists of all evaluation vectors of degree $d$ polynomials $f\colon\mathbb{F}_q^n\to\mathbb{F}_q$. That
is,
\[
\RM[n,q,d] = \left\{(f(\vec{x}))_{\vec{x}\in\mathbb{F}_q^n}~|~f\colon\mathbb{F}_q^n\to\mathbb{F}_q\text{ is a polynomial of degree at most $d$}\right\}.
\]
When $k > 1$, $\RM[n,q,d]$ is sometimes called a generalized Reed-Muller code, to distinguish from the prime field case, and as the title suggests, our results are most relevant to this case. However, henceforth, we will refer to them as Reed-Muller codes for simplicity.

Abusing notation, we will not distinguish functions and their evaluations when referring to codewords. That is, for $f: \Ff_q^n \xrightarrow[]{} \Ff_q$, we will simply say $f \in \RM[n,q,d]$ if $\deg(f) \leq d$ and we will view the codewords of $\RM[n,q,d]$ themselves as functions. When talking about the degree of a function $f$, we mean the total degree when $f$ is written as a polynomial.

Given a proximity parameter $\delta>0$, the goal in the problem of local testing of Reed-Muller codes is to design a randomized tester $\mathcal{T}$ that makes $Q$ oracle queries (for $Q$ which is as small as possible)
to a given function $f\colon\mathbb{F}_q^n\to\mathbb{F}_q$ such that:
\begin{enumerate}
  \item {\bf Completeness:} if $f$ is a polynomial of degree at most $d$, then $\mathcal{T}$ accepts with probability $1$.
  \item {\bf Soundness:} if $f$ is $\delta$-far from all degree $d$ polynomials, then the tester $\mathcal{T}$ rejects with probability at least $2/3$.
\end{enumerate}

\paragraph{Local Testing.}
Reed-Muller codes have a very natural and well studied local test \cite{AKKLR, KR, JPRZ, HSS} called the $t$-flat test. 
This test has its origins in the study of probabilistically checkable proofs~\cite{FGLSS,AroraSafra,ALMSS,RubinfeldSudan,AroraSudan,RS} (as it is related to the well known
plane versus plane, plane versus line and line versus point tests), as well as in the theory of Gowers' uniformity norms~\cite{Gowers}. To check if a given function $f$ is indeed low
degree, the tester samples a random $t$-dimensional affine subspace $U\subseteq \mathbb{F}_q^n$, queries $f(\vec{x})$ for all $\vec{x}\in U$ and checks whether the resulting
$t$-variate function $f|_U$ has degree at most $d$. Clearly the number of queries made is $q^t$, and it is also clear that the test is complete: if $f$ has degree at most $d$, then
the tester always accepts.
As for the soundness, it is known that one can take $t = \ceil{\frac{d+1}{q-q/p}}$ and get that the tester is somewhat sound~\cite{AKKLR}, meaning that the rejection probability is bounded 
away from $0$ (as opposed to at least  $2/3$). Indeed, a typical local-testing result shows that if a function $f$ is $\delta$-far from being a degree $d$ function, then
the tester rejects it with probability at least $\eps = \eps(q,d,\delta)>0$, which is a quantity that typically vanishes with the various parameters. To amplify the soundness, one 
repeats the tester $\theta(1/\eps)$ time sequentially to get $2/3$ rejection probability, thereby giving a tester whose query complexity is $O(q^t/\eps)$ and whose soundness is at least $2/3$.

Such testers for the Reed-Muller have been known for a long time. Indeed, in~\cite{AKKLR} the $t$-flat tester is analyzed and it is shown that the rejection probability 
of the basic tester is at least $\eps\geq \Omega(\delta/q^{t})$ leading to a tester with query complexity $O(q^{2t}/\delta)$. This soundness analysis turns out to not 
be optimal, and indeed, as we explain below, follow-up works have shown a better analysis of the $t$-flat tester.
In particular, it was shown that the $t$-flat tester is an \emph{optimal tester}.

\subsection{Optimal Testing of Reed Muller Codes}
In this paper, we focus on a much stronger notion of testing known as \emph{optimal testing}. Clearly, if a tester makes $Q$ queries and the proximity parameter is $\delta$, then
the rejection probability can be at most $\min(1, Q\delta)$; indeed, this is a bound on the probability to distinguish between a Reed-Muller codeword and
a Reed-Muller codeword that has been perturbed on a randomly chosen $\delta$-fraction of inputs. 
A tester is called optimal if the query-to-rejection probability tradeoff it achieves
is roughly that. Oftentimes, one settles for rejection probability which is a bit worse and has the form $c(q)\min(1,Q\delta)$ for some function $c(q)>0$ depending only on the field size
$q$. We will refer to such results also as optimal testing results. Thus, one would ideally like a tester which both achieves as small as possible query complexity, while
simultaneously being optimal.

Optimal testers are known for Reed-Muller codes. Such results were first proved over $\Ff_2$ by~\cite{BKSSZ}. Later on, 
optimal testing results were 
established for Reed-Muller codes 
over general fields \cite{HSS} as 
well as more broadly for the class of 
affine lifted codes~\cite{HRZS}. In all of these results, the $t$-flat test is analyzed (for $t$ chosen as above), and is shown to be an optimal tester. 
We remark that additionally, 
the analyses
of~\cite{BKSSZ,HSS,HRZS} led to improved query complexity for testing Reed-Muller codes. These results have important applications, most notably to the construction of small-set
expander graphs with many eigenvalues close to $1$~\cite{Shortcode}, which have later been used for improved PCP constructions and constructions of integrality
gaps~\cite{GHHSV,DHSV,DinurGuruswami,KaneMeka}.

Quantitatively, these results have two drawbacks. First, due to their application of the density Hales-Jewett theorem, the dependency of $c(q)$ on the field size $q$ is tower type, hence making
the result primarily effective for small fields. Secondly, while the query complexity achieved by their tester is the best possible for prime fields (as a lower bound on the query complexity
needed is given by the distance of the dual code of $\RM[n,q,d]$, which is $q^t$ if $q$ is prime),
it is not known to be optimal for prime power size 
fields. This raises two natural questions: 
does the flat tester actually perform worse on large fields (in comparison to small fields)? Is there a tester with smaller query complexity over non-prime fields 
(i.e.\ extension fields)?

In~\cite{RZS} a new local tester for the Reed-Muller code was designed. 
The query complexity of the tester is $Q = {\sf poly}(q) (3q)^{\ceil{\frac{d+1}{q}}}$, which is polynomially smaller than $q^{t}$ above (by a power of
$\frac{p}{p-1}$). This tester, which will be referred to as the \emph{sparse $t$-flat tester}, plays a crucial role in the current work and will be presented in subsequent sections.

Unfortunately, the rejection probability proved in~\cite{RZS} for the sparse $t$-flat tester is an $\eps$ which is sub-constant, and after repeating the tester $\Theta(1/\eps)$ times its query complexity
turns out to be roughly the same as that of the $t$-flat tester above. This leaves the Reed Muller code over extension fields in a rather precarious situation: a local characterization for the code ---  namely a basic tester that rejects far from Reed-Muller codewords with some non-negligible probability ---  is known, but amplifying 
the soundness to be constant sets us back to the same query complexity as of the $t$-flat tester.

\subsection{Main Result: Optimal, Query Efficient Tester for Generalized Reed Muller Codes}
The main result of this paper is a new and improved analysis of the tester of~\cite{RZS}. We show that the soundness of the tester is much better than the guarantee given by~\cite{RZS},
and that in fact this tester is also an optimal tester:
\begin{thm}\label{thm:main}
For all primes $p\in\mathbb{N}$
and prime powers $q = p^k$ there exists a tester $\mathcal{T}$ with query complexity $Q \leq 3q^{p+O(1)} (3q)^{\ceil{\frac{d+1}{q-1}}}$
such that given an oracle access to a function $f\colon\mathbb{F}_q^n\to\mathbb{F}_q$,
\begin{enumerate}
  \item {\bf Completeness:} if $f$ has degree at most $d$, then $\mathcal{T}$ accepts with probability $1$.
  \item {\bf Soundness:} if $f$ is $\delta$-far from degree $d$ , then $\mathcal{T}$ rejects with probability at least $c(q)\min(1, Q\delta)$, where $c(q) = \frac{1}{{\sf poly}(q)}$.
\end{enumerate}
\end{thm}
The test uses a parameter $t$ (where $\lceil \frac{d+1}{q-1}\rceil + t$ is analogous to the ``dimension'' parameter in the flat test), 
and the $t$ that we use will be $t = \max(p + 2, 10)$. 
We remark however that the analysis we give applies to all $t \geq \max(p + 2, 10)$, 
and we choose this specific $t$ so as to 
minimize the query complexity. We defer to Section~\ref{sec: sparse flat test construction} for more details on this parameter.

A lower bound on the query complexity needed is $q^{\ceil{\frac{d+1}{q-1}}}$, which once again follows by considering the dual code of the generalized Reed-Muller code
and arguing that this number is its distance. Hence, Theorem~\ref{thm:main} is tight up to a factor of ${\sf poly}(q) 3^{\ceil{\frac{d+1}{q-1}}}$; for large $q$, this factor
is very small compared to $q^{\ceil{\frac{d+1}{q-1}}}$, hence the query complexity achieved by our tester is essentially optimal for large field size $q$.

\subsection{Optimal Testing of Other Linear Lifted Affine Invariant Codes} \label{sec: lifted codes}
Our techniques are in fact more general, and also apply to testers of a wider family of codes, called (linear) lifted affine invariant codes. These codes were introduced by \cite{GKSS} and shown to be optimally testable in \cite{HRZS, KM}. In words, we show that any tester for such codes is optimal if it can be expressed as the product of polynomials, such that each polynomial is defined on a constant number of variables and such that the variables for each polynomial are disjoint. We describe this result in more detail below, but defer the full discussion to Section~\ref{sec:optimal_testing_from_other}. 

A code $\mathcal{C} \subseteq \{f: \Ff_q^{n} \xrightarrow[]{} \Ff_q \}$ is linear if its codewords form a linear subspace and is affine invariant if for any affine transformation $T: \Ff_q^{n} \xrightarrow[]{} \Ff_q^{n}$ and $f \in \mathcal{G}$, we have that $f \circ T \in \mathcal{G}$, where $f\circ T$ is defined as the function $f': \Ff_q^{n} \xrightarrow[]{} \Ff_q^n$ such that $f'(x) = f(T(x))$ for all $x \in \Ff_q^n$. Since $\C$ is linear, it has a dual $\C^{\perp}$ also consisting of functions from $\Ff_q^n \xrightarrow[]{} \Ff_q$, and $f \in \C$ if and only if $\langle f, h \rangle = 0$ for all $h \in \C^{\perp}$, where this inner product is the standard dot product of the evaluation vectors of $f$ and $h$ (over $\mathbb{F}_q^n$). Notice that, one can also compose $f$ with affine transformations $T: \Ff_q^{k} \xrightarrow[]{} \Ff_q^n$ for $k < n$. In this case, $f \circ T$ is a function from $\Ff_q^k \xrightarrow[]{} \Ff_q$, and we can consider the code consisting of all $f: \Ff_q^n \xrightarrow[]{} \Ff_q$ such that $f \circ T$ is in some affine invariant base code $\mathcal{G} \subseteq \{\Ff_q^{k} \xrightarrow[]{} \Ff_q \}$. This code is called the $n$-dimensional lift of $\mathcal{G}$ and is defined as:

\[
\Lift_n(\mathcal{G}) = \{f: \Ff_q^n \xrightarrow[]{} \Ff_q \; | \; f \circ T \in \mathcal{G} \text{ for all affine transformations } T: \Ff_q^{k} \xrightarrow{} \Ff_q^n \}.
\]
It is not hard to see that $\Lift_n(\mathcal{G})$ is also affine invariant and linear. Suppose we want to design a local tester for $\mathcal{C}$ and we know $\mathcal{C} = \Lift_n(\mathcal{G})$ for some affine invariant $\mathcal{G}$ defined as above with $k \geq 10$. A natural test is the following:
\begin{enumerate}
    \item Take $\mathcal{H} \subseteq \mathcal{G}^{\perp}$.
    \item Choose an affine transformation $T: \Ff_q^k \xrightarrow[]{} \Ff_q^n$ uniformly at random.
    \item Accept if and only if $\langle f \circ T, H \rangle = 0$ for all $h \in \mathcal{H}$.
\end{enumerate}
We remark that not every choice of $\mathcal{H}$ results in a local tester, and it is indeed possible to choose $\mathcal{H}$ so that there exist $f \notin \C$ that still pass the above test with probability $1$. Our main result shows that when such a test \textit{is} a local test and $\mathcal{H}$ consists of functions of a specified form, then the tester is automatically an optimal tester. In order to obtain explicit optimal testers one still has to find such a $\mathcal{H}$ that is a local tester, but this is not the focus of the current paper. 

The form for $\mathcal{H}$ that we require is as follows. Let 
\[
H(x_1, \ldots, x_{k'}) = \left(\prod_{i = 1}^{s} P_i(x_{m(i)}, \ldots, x_{m(i+1)-1}))\right),
\]
where $\poly(q) \geq  k-k'  \geq 0$, and $m(i+1) - m(i) \leq t'$ for some constant $t'$ and all $1 \leq i \leq s-1$. In words, $H$ is a polynomial in at most $k'$-variables, for some $k'$ that is within some constant power of $q$ from $k$, that can be factored into the product of polynomials in constant number of variables, and such that the variables of each of these polynomials is disjoint. Next let $\mathcal{M} \subsetneq \{\Ff_q^{k-k'} \xrightarrow[]{} \Ff_q \}$ be an affine invariant code and let $\overline{\mathcal{M}}$ be a basis for $\mathcal{M}^{\perp}$. Finally, suppose 
\[
\mathcal{H} = \{H(x_1, \ldots, x_{k'}) M(x_{k'+1}, \ldots x_k) \; | \;  M \in \mathcal{M}\}.
\]
Our theorem states:
\begin{thm}\label{thm:main_2}
    Suppose $\mathcal{H}$ is of the previously described form and suppose that the previously described test using $\mathcal{H}$ is a local tester for $\C = \Lift_n(\mathcal{G})$. Then the local tester is also optimal. That is,
    \begin{enumerate}
  \item {\bf Completeness:} if $f \in \C$ then the test accepts with probability $1$.
  \item {\bf Soundness:} if $f$ is $\delta$-far from degree $\C$ , then the test rejects with probability at least $c(q)\min(1, Q\delta)$, where $c(q) = \frac{1}{{\sf poly}(q)}$.
\end{enumerate}
\end{thm}
Although our result is stated for lifted affine-invariant codes, it also applies equally to affine-invariant codes by taking $\C = \mathcal{G} = \Lift_{k}(\mathcal{G})$. In contrast, the optimal testing result for lifted affine-invariant codes in \cite{KM} applies only to the $k$-flat test for $\Lift_n(\mathcal{G})$, which is no longer ``local'' in the case of $\C = \mathcal{G}$ as it looks at the entire domain. On the other hand, for the $\C = \mathcal{G}$ case, one could still obtain locality using our result by designing a set $\mathcal{H}$ of the specified form that has sparse support. 
\skipi
Theorem~\ref{thm:main_2} 
gives a general recipe to
construct optimal testers, 
thus making progress
on the task of establishing
optimal testing results
for general affine invariant codes. 
We would 
like to highlight two 
interesting avenues that 
we leave for future works.
First, it would be  
interesting to 
investigate what other 
affine invariant codes  
can be analyzed via 
Theorem~\ref{thm:main_2}. 
This could lead to new optimal testing
result for other codes, or otherwise 
to a more general class of testers that 
one can then try to analyze using the 
tools presented herein (and their 
extensions). Second, it would be 
interesting to extend our techniques to remove the requirements on the form of $\mathcal{H}$ (or perhaps weaken it somehow), and show that \empty{any} local test for affine-invariant codes is optimal. 

\subsection{Optimal Testing via Global Hypercontractivity}
The proof of Theorem~\ref{thm:main} is very different from the proofs of~\cite{BKSSZ,HSS,HRZS} (which proceed by induction on $n$) as well as from the proof of~\cite{RZS}  (which
proceeds by presenting a local characterization of the generalized Reed-Muller code and appealing to a general and powerful 
result due to~\cite{KS}, that converts local characterizations to local tester). Instead, our proof is inspired by a new approach
for establishing such results via global hypercontractivity~\cite{KM}.

Global hypercontractivity is a recently introduced tool that is often useful when working 
with non small set expander graphs. One corollary of global hypercontractivity (which is
morally equivalent) is a useful characterization of all small sets in a graph that have edge expansion bounded away from $1$. The study of global hypercontractivity has its roots
in the proof of the 2-to-2 Games Theorem \cite{KMS, DKKMS2, DKKMS, KMS2}, however by now it is known to be useful in the study a host of different problems (see for example~\cite{KLLM,KLLMcodes,KLM,KM,BBKSS,BHKL,BHKL2,GLL}).

Below, we explain the global hypercontractivity approach to proving local testability results. 
In Section~\ref{sec:techniques} we explain how we extend this
approach to the realm of generalized Reed-Muller codes in order to analyze the sparse $t$-flat tester and prove Theorems~\ref{thm:main},~\ref{thm:main_2}.

\subsubsection{Optimal Testing of Reed-Muller Codes via Global Hypercontractivity}
In~\cite{KM}, the authors relate the analysis of the $t$-flat tester of the Reed-Muller code to expansion properties of the affine Grassmann graph.
Here, the affine Grassmann graph is the graph whose vertex set is the set of all $t$-flats in $\mathbb{F}_q^n$, and two vertices are adjacent if they intersect in a $t-1$ dimensional
affine subspace. In short, the approach of~\cite{KM} starts with the assumption that the $t$-flat tester accepts $f$ with probability at least $1-\eps$ (for $\eps$ thought of as small)
and proceeds to prove a structural result on the set of $t$-flats on which the tester rejects:
\[
S = \left\{T\subseteq\mathbb{F}_q^n~|~{\sf deg}(f|_{T}) > d, T\text{ is a $t$-flat}\right\}.
\]
In particular, using a lemma from~\cite{HSS} they prove that the set $S$ has poor expansion properties in the affine Grassmann graph, and use it to prove that there exists a point
$x^{\star}\in\mathbb{F}_q^n$ such that the tester almost always rejects when it selects a subspace $T$ containing $x^{\star}$. This suggests that $f$ is erroneous at $x^{\star}$
and should be corrected, and indeed using that a local correction procedure is devised in~\cite{KM} to show that the value of $f$ at $x^{\star}$ could be changed and lead to 
an increase in the acceptance probability of additive factor $q^{t-n-O(1)}$. Iterating this argument, one eventually changes $f$ in at most $C(q)\cdot q^{-t}\eps$ fraction 
of the inputs and gets a function $f'$ on which the tester accepts with probability $1$. Such $f'$ must be of degree at most $d$, hence one gets that $f$ is 
close to a degree $d$ polynomial.

\subsubsection{Optimal Testing of Generalized Reed-Muller Codes via Global Hypercontractivity}
While the approach of~\cite{KM} seems to be more robust and thus potentially applicable towards analyzing a richer class of codes, it is not completely obvious how to do so. For the $t$-flat tester one may associate a very natural graph with the tester, but this is much less clear in the context of the sparse $t$-flat tester (which is the tester 
that we analyze). The pattern of inputs 
which are queried by the tester is no longer a nice-looking subspace (but this seems inherent in order to improve upon the query complexity of the flat tester).

At a high level, we show that another way of approaching this problem is by utilizing the symmetries of the code and constructing graphs on them. 
For that, we have to think of the tester as the composition of a ``basic tester'' and a random element from
the group of symmetries of the code. In our case,
the group of symmetries is the group of affine linear transformations over $\mathbb{F}_q^n$, and the graph that turns out to be related to the analysis of the sparse 
$t$-flat tester is the so-called Affine Bi-linear Scheme Graph.

At first sight, affine linear transformations seem to be morally equivalent to flats (identifying the image of
an affine linear transformation with a subspace), and indeed there are many connections between results on the former graph and result on the latter graph. 
However, the distinction between affine linear transformations and flats will be crucial for us. Indeed, while two affine linear transformations $A_1$ and $A_2$ may have
the same image, the performance of the tester when choosing $A_1$ and when choosing $A_2$ will not be the same whatsoever, and therefore we must capitalize on this distinction
in our soundness analysis.

\subsection{Our Techniques}\label{sec:techniques}
In this section, we give a brief overview of the techniques underlying the proof of Theorem~\ref{thm:main}.
We start by presenting the sparse $(s+t)$-flat tester of~\cite{RZS} and then take a somewhat different perspective on it in the form of groups of symmetries. After that, we explain the high level strategy of the proof of Theorem~\ref{thm:main}, and explain some of the unique challenges that arise compared to the
analysis of the standard $t$-flat tester.
For the sake of presentation, we focus on the case that $p=2$ for the remainder of this section, and switch back to general $p$ in Section~\ref{sec: preliminaries}.
\subsubsection{The Sparse Flat Tester} \label{sec: sparse flat tester}
The construction of the sparse flat tester of~\cite{RZS} begins with the following observation. In the $p=2$ case, define the bivariate polynomial
$P\colon \mathbb{F}_q^2\to\mathbb{F}_q$ by 

\[
P(x_1, x_2) = \frac{-x_2^{q-1} + (x_1 + x_2)^{q-1}}{x_1}= \sum_{i = 0}^{q-2} x_1^ix_2^{q-2-i}.
\]
 In~\cite{RZS}, the authors observe the following two crucial properties of $P$ that make it useful towards getting a local testing result:
\begin{enumerate}
  \item {\bf Sparse Support:} the support of $P$ has size at most $3q$; indeed, if $x_1+x_2\neq 0$, $x_1\neq 0$ and $x_2\neq 0$, then by Fermat's little theorem
  $(x_1+x_2)^{q-1} = x_2^{q-1} = 1$ and $x_1\neq 0$, so $P(x_1,x_2) = 0$.
  \item {\bf Detects the Monomials $x_1^{q-i}x_2^{i}$:} looking at the inner product of $P$ with a monomial $M$, defined as
  $\inner{P}{M} = \sum\limits_{x_1,x_2} P(x_1,x_2)M(x_1,x_2)$, we get that if $M = x_1^{q-i}x_2^{i}$ for
  $i \in \{1,\ldots,q-1\}$ then $\inner{P}{M}\neq 0$. Indeed,
  \[
  \inner{P}{M}
  = \sum\limits_{j=0}^{q-2}\sum\limits_{x_1}x_1^{q-i+j}\sum\limits_{x_2}x_2^{q-2-j+i}
  = \sum\limits_{j=0}^{q-2}(q-1)1_{q-i+j = q-1} (q-1) 1_{q-2-j+i = q-1},
  \]
  so we only have contribution from $j = i-1$ and it is non-zero. For any other monomial $M$, a similar computation shows that
  $\inner{P}{M} = 0$, hence taking an inner product of a function $f$ with $P$ may be thought of detecting whether in $f$
  there is some monomial of the form $x_1^{q-i}x_2^{i}$.
\end{enumerate}

With this in mind, the sparse tester follows. In~\cite{RZS} it is argued that to design a local tester for the generalized Reed-Muller code it suffices to design a tester that detects whether certain canonical monomials exist. Writing $d+1 = s \cdot \frac{q}{2} + r$, where $s$ is even and $r < q$, it is sufficient to detect whether any monomials of the form $\prod_{i=1}^{s/2}x_{2i-1}^{q/2}x_{2i}^{q/2}  \cdot\prod\limits_{i=s+1}^{s+t} x_{i}^{e_i}$ where $\sum_{i=1}^{t}e_i \geq r$ and $t$ is a small constant (say, $t = 10$). Using $P$ from above, a detector
for such monomials is given by
\[
H_{e'}(x_1,\ldots,x_{s+t}) = P(x_1,x_2)\cdots P(x_{s-1}, x_{s})\cdot M_{e'}(x_{s+1},\ldots,x_{s+t}),
\]
where $e' + e = (q-1, \ldots, q-1)$ and $M_{e'}(x_{s+1},\ldots,x_{s+t}) = \prod\limits_{i=s+1}^{s+t} x_{i}^{q-1-e_i}$. 
We note that most of the degree of $H_{e'}$ comes from the product of the $P$'s, while at most $t(q-1)-r$ of the degree comes from the rest. As the support of $P$ is rather sparse, it follows that the support of $H_{e'}$ is also rather sparse. More precisely, the support of $H_{e'}$ has size at most $(3q)^{\frac{s}{2}+t}$, and as $t$ should be thought of as small and $s$ is roughly $2d/q$, 
the support of $H_{e'}$ has size roughly $(3q)^{d/q}$. This yields a query complexity of $(3q)^{d/q}$, which is quadratically better than $q^{\ceil{\frac{d}{q-q/p}}} \approx q^{2d/q}$ given by the flat tester.

As mentioned earlier, the soundness analysis in~\cite{RZS} relies on a black box result from~\cite{KS}. They manage to show that the detector they construct
implies a tester with the same query complexity that rejects functions that are $\delta$-far from Reed-Muller codewords with probability $\Omega((3q)^{-2(s/2+t)} \delta )$. 
To get constant rejection probability one has to repeat the tester $(3q)^{2(s/2+t)}$ times and drastically increasing the query complexity; we defer a more detailed discussion to Section~\ref{sec:local_testers}.

\subsubsection{The Group of Symmetries Perspective}
Our first task in proving Theorem~\ref{thm:main} is to design a tester based on the ideas from~\cite{RZS}. The tester is very natural:
\begin{enumerate}
  \item Sample a random affine map $T\colon \mathbb{F}_q^{s+t}\to\mathbb{F}_q^{n}$; here, $s$ and $t$ should be thought of as in the previous section.
  We are going to look at the function $f\circ T$, but not query all of its values (indeed, querying all of its values would amount to the $(s+t)$-flat tester).
  \item For any sequence of degrees $e = (e_{s+1},\ldots,e_{s+t})$ such that $\sum_{i=1}^{t} q-1-e_i \geq r$, take $H_e$ and compute $\inner{f\circ T}{H_e}$. Reject if this inner product is non-zero for any such degree sequence. Otherwise, accept.
\end{enumerate}
In words, we first take the restriction of $f$ to a random $(s+t)$-flat, and then apply the detector polynomials of~\cite{RZS}. Although we test for multiple degree sequences (up to $q^t$), we will see in Section~{\ref{sec:local_testers}} that the support of $H_e$ is the same in each case. Therefore, we can perform the inner product for all of the degree sequences using the same $q^{s+t}$ queries. Another way to think about this tester is that we have the group of symmetries of the Reed-Muller code (affine linear transformations), and our tester proceeds by first taking a random symmetry from this group, taking a restriction, and then using the detector of~\cite{RZS}.

\subsubsection{The Graph on Affine Linear Transformations and Its Analysis}
With the above tester in mind, the next question is how to analyze it.
In the case of the flat tester we had a very natural graph associated with the tester (the affine Grassmann graph). The above tester seems related to flats as well, since the image of an affine transformation is a flat; however, there is a key, crucial distinction between the two. In the above tester, we are only going to look at the value of $f$ at a few locations
in $\im(T)$, hence the tester may perform differently on $T$ and $T'$ even if they have identical images. This lack of symmetry
is crucial for the sparseness of the test, but makes the task of associating a graph with the tester trickier.

To gain some intuition as to what this graph is supposed to be, recall that in the flat testers, one could look at the $t$-flat tester for
all $t$ (not necessarily the smallest one). The relations between these testers for different $t$'s play a crucial role in all of the works establishing
optimal testing results, and in particular it is known that the rejection probability of the $(t+1)$-flat tester is at most $q$ times the rejection probability
of the $t$-flat tester. We will elaborate on this property a bit later (referred to as the ``shadow lemma'' below). 
Another benefit of looking at various testers is that it gives a natural way of arriving
at the affine Grassmann graph, by doing an up-down walk between these testers. To obtain the edges of the affine Grassmann graph, one can start with a $t$-flat, go up to a $(t+1)$-flat that contains it, and
then back down to a $t$-flat contained in the $(t+1)$-flat. What is the right analog of this operation in the context of the sparse flat tester?

\paragraph{Going up.}
The above discussion suggests looking for analogs of the tester above for higher dimensions, and there is a clear natural analog to the up step:
for any $r\geq 0$, we can look at the $(s+t+r)$ sparse flat tester, in which one chooses an affine map
$T\colon \mathbb{F}_q^{s+t+r}\to\mathbb{F}_q^{n}$ randomly, and proceeds with the tester as above for all viable degree sequences (but over 
more variables). Thus, starting with the $(s+t)$ sparse flat tester and with an affine map
$T$ thought of as $(n\times (s+t))$ matrix over $\mathbb{F}_q$ and an affine shift $c \in \mathbb{F}_q^n$, we can go ``up'' to the $(s+t+1)$ sparse flat tester by choosing a random
vector $w \in \mathbb{F}_q^n$ and looking at affine transformation $A$ corresponding to the matrix whose columns are the same as of $T$, 
except that we append $w$ as the last column. Just like in the flat tester, it can be observed without much difficulty that if the $(s+t)$ sparse flat tester rejects when choosing $T$,
then the $(s+t+1)$ sparse flat tester rejects when choosing $A$.

\paragraph{Going down.}
The ``going down'' step is also simple, but a bit harder to motivate. Taking inspiration from the flat tester, one may want to apply some linear shuffling on the columns of
$A$ and then ``drop'' one of the columns. This doesn't seem to work though, as doing so would lead to a $T'$ in which the performance of the sparse flat tester seems ``completely
independent'' to its performance on $T$ (in the sense that the set of points it looks at in $T'$ would typically be disjoint from the set of points it looks at in $T$). 

Thus, when going down we wish to do so in a way that keeps $T'$ and $T$ equal on many points. A natural thing to try is to apply an affine transformation from $R: \Ff_q^{s+t+1} \xrightarrow[]{} \Ff_q^{s+t}$ to $A$ that fixes a co-dimension $1$ space. In this case, $T' = A \circ R$ is outputted and $T'$ is equal to $A$ on the co-dimension $1$ space that $R$ fixes. On the other hand, by construction $T$ is equal to $A$ on a co-dimension $1$ space as well - namely the subspace with last coordinate equal to $0$. Therefore, after the down step we get $T'$ which is equal to $T$ on a subspace of dimension $s+t$ - which is essentially as similar to $T$ as possible while still being distinct from it.

Put a different way, to
go down from the $(s+t+1)$ sparse flat tester to $(s+t)$ sparse flat tester we proceed by choosing $b_1,\ldots,b_{s+t},b_{s+t+1}\in\mathbb{F}_q$ uniformly and independently
and taking the affine transformation $T'$ corresponding to the matrix whose $i$th column is ${\sf col}_{i}(T) + b_i {\sf col}_{s+t+1}(T)$, and 
whose shift is $c+b_{s+t+1}u$. In words, we drop the final column but add a random multiple of it to
each one of the other columns of $T$. 

\paragraph{Going up and then down.}
Stitching these two operations, one gets a graph whose set of vertices is the set of affine linear transformations $T\colon \mathbb{F}_q^{s+t}\to\mathbb{F}_q^{n}$
and whose edges are very similar in spirit to the affine Grassmann graph; this graph is known as the bi-linear scheme graph.
The core of our analysis relies on $3$ components:
\begin{enumerate}
  \item Relating the performance of the tester and expansion on this graph (the shadow lemma), and proving that the set of $T$'s on which the tester rejects
has edge expansion at most $1-1/q$.
  \item Studying the structure of sets in this graph whose expansion is at most $1-1/q$ and proving they must have some strong local structure.
  \item Using said local structure towards a correction argument, proving that if the sparse flat tester rejects with small probability, then $f$ is close to
  a Reed-Muller codeword.
\end{enumerate}

\subsubsection{The Shadow Lemma}
The shadow lemma is a result asserting a relation between the rejection probability of a $(s+t+1)$ tester and the $(s+t)$ tester.
\paragraph{The Shadow Lemma for Flat Testers.}
In the context of flat testers, the lemma asserts that the fraction of flats of dimension $(s+t+1)$ on which the tester rejects is
at most $q$ times the fraction of $(s+t)$ flats the tester rejects. The name of the result stems from the fact that letting $S$ be
the set of $(s+t)$-flats on which the tester rejects, the set of $(s+t+1)$ flats on which the tester rejects is exactly the upper
shadow of $S$:
\[
S^{\uparrow} = \{A\subseteq\mathbb{F}_q^n~|~{\sf dim}(A) = s+t+1, \exists B\in S, B\subseteq A\}.
\]
This means that on average, an element $A\in S^{\uparrow}$ has $1/q$ of its subsets $B\subseteq A$ in $S$, and thinking of an
edge in the affine Grassmann graph an up-down step we get that the probability that a random step from $S$ goes to a vertex outside
$S$ is at most $1-1/q$. That is, the edge expansion of $S$, defined as
\[
\Phi(S) = \frac{|\{ e = (A,A')\in E~|~A, A'\in S\}|}{|\{ e = (A,A')\in E~|~A\in S\}|},
\]
is at most $1-\frac{1}{q}$.

\paragraph{The Shadow Lemma for Sparse Flat Testers.}
In the context of the sparse flat tester, we wish to argue that something along the same lines holds.
Towards this end, fixing the function $f\colon\mathbb{F}_q^n\to\mathbb{F}_q$, we define
\[
S = \{T\colon \mathbb{F}_q^{s+t}\to\mathbb{F}_q^n~|~\text{the sparse flat tester rejects when choosing $T$}\}.
\]
Due to the asymmetry between the up and down steps, there is no clear analog of the upper shadow of $S$. However, it is
still true that if $T\in S$ and we append to $T$ some vector $u\in\mathbb{F}_q^n$ to form an affine map
$R\colon\mathbb{F}_q^{s+t+1} \to \mathbb{F}_q^n$, then the sparse $(s+t+1)$ tester rejects $R$ and it makes sense
to define
\[
S^{\uparrow} = \{R\colon \mathbb{F}_q^{s+t+1}\to\mathbb{F}_q^n~|~\text{the sparse flat tester rejects when choosing $R$}\}.
\]
We prove that starting from $R\in S^{\uparrow}$ and performing a down step to a $T'\colon \mathbb{F}_q^{s+t+1}\to\mathbb{F}_q^n$,
with probability at least $1/q$ the sparse flat tester still rejects and hence $T'\in S$. In particular, we conclude that
$\mu(S^{\uparrow}) \leq q\mu(S)$ (where $\mu(S)$ denotes the ratio between the size of $S$ and the total number of affine maps
from $\mathbb{F}_q^{s+t}$ to $\mathbb{F}_q^n$, and $\mu(S^{\uparrow})$ is defined similarly). Using the same logic as before
we conclude that $\Phi(S)\leq 1-\frac{1}{q}$, where here we measure edge expansion with respect to the bi-linear scheme graph.

\subsubsection{Expansion on the Bi-linear Scheme Graph}
Equipped with the understanding that $S$ is a small set (as we assume the sparse flat tester rejects with small probability)
and $\Phi(S)\leq 1-\frac{1}{q}$, the next question is what sort of structure this implies. In the bi-linear scheme graph there are
natural examples of such sets, which are analogs of the zoom-in/ zoom-out sets in the context of subspaces. Roughly speaking,
there is one type of examples which intuitively should be relevant for us, which is zoom-in sets:
\[
\mathcal{C}_{x^{\star}, y^{\star}} = \{T\colon \mathbb{F}_q^{s+t}\to\mathbb{F}_q^n~|~T(x^{\star}) = y^{\star}\}.
\]
There are other examples of small sets which have poor expansion in the bi-linear scheme graph, however these seem ``irrelevant'' in the present context. 
Indeed, after showing that our small set $S$ cannot be correlated with any of these other examples (a notion which is often referred to as pseudo-randomness
with respect to zoom-outs and zoom-ins on the linear part), we use the theory of global hypercontractivity to prove that there must be 
$x^{\star}$ and $y^{\star}$ such that $\mu(S\cap \mathcal{C}_{x^{\star}, y^{\star}})\geq (1-o(1))\mu(\mathcal{C}_{x^{\star}, y^{\star}})$.
In words, the sparse $(s+t)$-flat tester almost always rejects inside $\mathcal{C}_{x^{\star}, y^{\star}}$. 

We remark that the proof that $S$ has no ``correlation'' with any other non-expanding set in the bi-linear scheme is 
rather tricky, and much of the effort 
in the current paper is devoted to that.
To do so, we have to build upon ideas from~\cite{KM} as well as use a new 
relation between the sparse $(s+t)$-flat 
tester applied on a function $f$ and 
the $t$-flat tester applied on a related 
function $\tilde{f}$. As the 
construction of
this related function is somewhat 
technical, we defer a detailed 
discussion of it to 
Section~\ref{sec: relation to full flat tester}.

\subsubsection{Finishing the Proof via Local Correction}
Intuitively, the only way that $S$ could be dense inside some $\mathcal{C}_{x^{\star}, y^{\star}}$ is if we started with a Reed-Muller codeword $g$, and
perturbed it on a small number of inputs, of which $y^{\star}$ is one. Indeed, in this case we would have that $\inner{g\circ T}{H_e} = 0$ for all
$T\in \mathcal{C}_{x^{\star}, y^{\star}}$ and exponent vectors $e$ checked by the tester prior to perturbing, and after changing the value of $g$ at $y^{\star}$, we would also change the value of $g\circ T$ at $x^{\star}$, breaking our previous equality. This intuition suggests that $y^{\star}$ is a point in which we should fix the value of $f$ and get closer to a Reed-Muller codeword,
and we show that this is indeed the case.

There are several difficulties that arise when inspecting this argument more deeply. If $H_e(x^{\star}) = 0$, then the above reasoning breaks (as the value of $f$
at $y^{\star}$ is multiplied by $0$); however, in this case it does not make sense that $S$ could be very dense inside
$\mathcal{C}_{x^{\star}, y^{\star}}$ (as essentially the only input of $f$ included in the inner product is $f(y^{\star})$, but in any case it is multiplied by $H_e(x^{\star}) = 0$).
Indeed, in our argument we show that if $S$ is very dense in $\mathcal{C}_{x^{\star}, y^{\star}}$, then it must be the case that $H_e(x^{\star})\neq 0$. Moreover, we show that
the density of $S$ inside $\mathcal{C}_{x^{\star}, y^{\star}}$ and inside $\mathcal{C}_{{x^{\star}}', y^{\star}}$ is roughly the same for all $x^{\star}$ and ${x^{\star}}'$ 
in the support of $H_e$. This last step is crucial for the analysis to go through and requires us to adapt and strengthen techniques from \cite{KM}. At the end of this step, roughly speaking, we conclude that the tester rejects with probability close to $1$ whenever it queries the value of $f$ at $y^{\star}$.

The last step in the argument is to show that we can change the value of $f$ at $y^{\star}$ and decrease the rejection probability of the tester by additive factor of $\Theta(q^{s+t-n})$ (which is proportional to the probability that the tester looks at $y^{\star}$). 
We do so by a reduction to the same problem over the standard flat tester (which was solved in~\cite{KM}). The idea is to look at somewhat larger
tester of dimension $s+t+100$, fix the first $s$ coordinates and let the rest vary, so that the tester becomes a local version of the standard flat tester.

Given that, and iterating the argument, we eventually reach a function $f'$ that differs from $f$ on at most $\Theta(\eps/q^{s+t-n})$ of the inputs 
(where $\eps$ is the original rejection probability) and passes the test with probability $1$. Hence, $f'$ is a Reed-Muller codeword, and so $f$ is $O(\eps/Q)$ close to a Reed-Muller codeword.

\section{Preliminaries} \label{sec: preliminaries}
\paragraph{Notations.}
For an integer $n$ we denote $[n] = \{1, \ldots, n\}$. For a prime power $q$ we let 
$\mathbb{F}_q$ be the field of size $q$, and 
we denote by $\mathbb{F}_q^{*}\subseteq \mathbb{F}_q$ 
the set of non-zero elements in it.
\subsection{Basic Facts about Fields}
Throughout, abusing notations we define the Reed-Muller code $\RM[n,q,d]$ as the set of functions over $\Ff_q^n$ that can be written as polynomials of degree at most $d$. 
Henceforth fix $d$ and write 
\[
d + 1 = \ell \cdot p(q-q/p) + r = s(q-q/p) + r,
\]
where we set $s = \ell \cdot p$, and $0 < r \leq p(q-q/p)$.  

We will need a few 
basic facts. First, it is 
a well known fact that $\mathbb{F}_q^{*}$ has a multiplicative generator which we often
denote by $\gamma$. 
The next lemma uses the existence of a multiplicative generator to estimate 
sums over $\Ff_q$.
\begin{lemma} \label{lm: power sums}
For any prime power $q$ and integer $i \in  \{0,\ldots, q-1\}$,
\begin{equation*}
    \sum_{\alpha \in \Ff_q}\alpha^i =
    \begin{cases}
      -1, & \text{if}\ i = q-1, \\
      0, & \text{otherwise}.
    \end{cases}
  \end{equation*}
\end{lemma}

\begin{proof}
If $i = q-1$, then $\alpha^i = 1$ for all $\alpha \neq 0$, while $0^i = 0$. Therefore, the sum is one summed up $q-1$ times which is $-1$ in $\Ff_q$. For $i \in \{1, \ldots, q-2\}$, recall that $\Ff_q^*$ has a generator $\gamma$. That is, $\Ff_q^* = \{1, \gamma, \ldots, \gamma^{q-2}\}$. Since $\gamma^i \neq 1$, we may write
\[
\sum_{\alpha \in \Ff_q} \alpha^i = \sum_{j = 0}^{q-2} (\gamma^i)^j = 
\frac{(\gamma^i)^{q-1}-1}{\gamma^i - 1} = 
\frac{1-1}{\gamma^i-1} = 0.
\]

On the other hand, if $i = 0$ then the sum on the left hand side of the lemma is equal to $\sum_{\alpha \in \Ff_q} \alpha^0 = \sum_{\alpha \in \Ff_q} 1 = q = 0$.
\end{proof}

\subsection{Functions over Fields}
 Given two functions $f,g\colon\mathbb{F}_q^n\to\mathbb{F}_q$, we measure the distance between them in terms of the normalized Hamming distance, that is,
 \[
 \delta(f,g) = \Pr_{x\in\mathbb{F}_q^n}[f(x)\neq g(x)].
 \]
 The distance of a function $f\colon\mathbb{F}_q^n\to\mathbb{F}_q$ from the Reed-Muller code $\RM[n,q,d]$, denoted by $\delta_d(f)$, is defined as the minimal
 distance between $f$ and some function $g\in \RM[n,q,d]$. That is,
\[
\delta_d(f) = \min_{g \in \RM[n,q,d]} \delta(f, g) = \min_{g \in \RM[n,q,d]} \Pr_{x \in \Ff_q^n} [f(x) \neq g(x)].
\]
For two functions $f, g: \Ff_q^{n} \xrightarrow[]{} \Ff_q$, define their inner product as
$$\langle f, g \rangle = \sum_{v \in \Ff_q^n} f(v)g(v).$$
It is clear that this inner product is bi-linear. 
Monomials are the basic building blocks of all polynomials over $\mathbb{F}_q$, and the following notation will be convenient for us to use:
\begin{definition}
 For an exponent vector $e \in  \{0,\ldots, q-1\}^n$, we define the monomial $x^e = \prod_{i=1}^n x_i^{e_i}$  
\end{definition}

The next lemma allows us to compute inner product between monomials:
\begin{lemma} \label{lm: monomial inner product}
    For $e,e'\in  \{0,\ldots, q-1\}^n$, we have that
    \begin{equation*}
    \langle x^{e}, x^{e'} \rangle =
    \begin{cases}
      (-1)^n, & \text{if}\ e + e' = (q-1, \ldots, q-1), \\
      0, & \text{otherwise}.
    \end{cases}
  \end{equation*}
\end{lemma}
\begin{proof}
    Let $e'' = e + e'$. By definition
    \[
     \langle x^{e}, x^{e'} \rangle = \sum_{(\alpha_1, \ldots, \alpha_n) \in \Ff_q^n} \prod_{i=1}^n \alpha_i^{e_i} = \prod_{i=1}^n \sum_{\alpha \in \Ff_q} \alpha^{e''_i}.
    \]
    The result follows from Lemma~\ref{lm: power sums}.
\end{proof}

\subsection{Affine Linear Transformations and the Affine Bi-linear Scheme}
In this section we present the Affine Bi-linear scheme, which plays a vital role in our arguments.
\begin{definition}
  We denote by $\T_{n, \ell}$ the set of affine transformations $T: \Ff_q^{\ell}\xrightarrow[]{} \Ff_q^{n}$.
\end{definition}
Each affine transformation $T \in \T_{n, \ell}$ consists of a linear part $M \in \Ff_q^{n \times \ell}$ and a translation $c \in \Ff_q^n$, and we will use the writing convention $T = (M, c)$ to refer to the fact that $T$ is the affine transformation such that $Tx = Mx + c$ for $x \in \Ff_q^{\ell}$. In words, $M$ is the linear part of $T$ and $c$ is the affine shift. We stress that an affine transformation $T \in \T_{n,\ell}$ is not necessarily full rank.
 
 \begin{definition}
The density of a set $S \subseteq \T_{n, \ell}$ is defined as $\mu_{\ell}(S) = \frac{|S|}{| \T_{n, \ell}|}$.
\end{definition}
Often times, we drop the subscript $\ell$ if it is clear from the context. 
Our analysis will require us to view affine transformations as vertices of some suitable test graph. For example, 
for the (standard) $t$-flat test, this graph is the \emph{affine Grassmann graph}, $\AffGras(n, t)$. The vertices of the graph are all of the $t$-flats in $\Ff_q^n$, and two vertices $U_1$ and $U_2$ are adjacent if they intersect in a $(t-1)$-flat, that is, $\dim(U_1 \cap U_2) = t-1$. Below we define an analogous graph structure on affine transformations, which we refer to as the affine bi-linear scheme graph.

\begin{definition}
    The affine bi-linear scheme graph, $\AffBilin(n, \ell)$, has vertex set $\T_{n, \ell}$. Two vertices $T_1, T_2 \in \T_{n, \ell}$ are adjacent adjacent if and only if they are equal on an $\ell-1$ dimensional affine subspace. This condition can also be written as $\dim(\ker(T_1 - T_2)) \geq \ell-1$.
\end{definition}
Write $T_1 \sim T_2$ to denote an adjacency. We remark that the affine Grassmann graph can be obtained from the affine bi-linear scheme by identifying each $T_1 \in \T_{n, \ell}$ with its image (that is, closing the set $\T_{n, \ell}$ under some group operation), 
however as explained in the introduction this distinction will be crucial for us.

\subsection{Expansion and pseudo-randomness in the Affine Bi-linear Scheme}
We will use the standard notion of edge expansion, defined as follows.
\begin{definition}
For a regular graph $G = (V,E)$ and a set of vertices $S\subseteq V$, we define the edge expansion of $S$ as 
\[
\Phi(S) = \Pr_{\substack{u\in S\\ (u,v)\in E}}[v\not\in S].
\]
In words, the edge expansion of $S$ is the probability to escape it in a single step of the random walk on $G$.
\end{definition}

We will use $\Phi_{n, \ell}$ to denote edge expansion on $\AffBilin(n, \ell)$. When $n$ and $\ell$ are clear from context we drop the subscripts. 
Later on in the paper, we will also consider edge expansion over the affine Grassmann graph $\AffGras(n, \ell)$; abusing notations, we will denote edge expansion there also using the notation $\Phi_{n, \ell}$, and it will be clear from context which graph we are considering.

\subsubsection{The Up-Down view of the Ranom Walk}
It will be helpful for us to consider the following equivalent, two-step process for sampling a neighbor of $T_1 = (M_1, c_1) \in \T_{n, \ell}$ in the affine bi-linear scheme graph:
\begin{enumerate}
    \item 
    \textbf{Go Up:} Choose a random $w \neq 0$. Let $M'$ be the matrix obtained by appending the $w$ as a new column to $M_1$, and $T' = (M', c) \in \T_{n, \ell+1}.$
    \item \textbf{Go Down:} 
    Choose a random $R = (A, b) \in \T_{\ell+1, \ell}$, where the first $\ell$ rows of $A$ are the identity matrix and the first $\ell$ entries of $b$ are $0$, and at least one entry out of the last row in $A$ and the last entry in $b$ is nonzero.
    
    \item Output $T_2 = T' \circ R$.
\end{enumerate}
It is easy to see that for $T_2 = (M_2, c_2)$, each column of $M_2$ is equal to the corresponding column in $M_1$ plus some multiple of $w$ and likewise for $c_2$ and $c_1$.

\subsubsection{Non-expanding Sets in the Affine Bi-linear Scheme}
As explained in the introduction, 
our proof considers the set $\Sc$ of affine-transformations on which the test rejects 
as a set of vertices in the affine bi-linear scheme graph. We prove that $\Sc$ has small edge expansion in that graph, and then use 
global hypercontractivity in order to derive some conclusion about the structure of $\Sc$, 
which is in turn helpful towards a local correction argument.

To facilitate this argument, we must first understand the structure of a few canonical examples of non-expanding sets in ${\sf AffBilin}(n,\ell)$. 
In the affine Grassmann graph one has the following examples:
\begin{itemize}
    \item zoom-ins: $\D_{a} = \{U \in \AffGras(n, \ell)\; | \; a \in U\}$, for $a \in \Ff_q^{n}$.
    \item zoom-outs: $\D_{W} = \{U  \in \AffGras(n, \ell)\; | \; U \subseteq W\}$, for a hyperplane $W \subseteq \Ff_q^{n}$.
    \item zoom-ins on the linear part: $\D_{a, \lin} = \{U = z + V \in \AffGras(n, \ell)\; | \; a \in V\}$, for $a \in \Ff_q^{n}$.
    \item zoom-outs on the linear part: $\D_{W, \lin} = \{U = z+V  \in \AffGras(n, \ell)\; | \; V \subseteq W\}$, for a hyperplane $W \subseteq \Ff_q^{n}$.
\end{itemize}

It is not hard to see that each example has expansion at most $1-1/q$; also, it is an easy calculation to show that the 
density of each one of these sets is small (and is vanishing provided that $\ell$ is significantly smaller than $n$ and both go to infinity). Each of these examples also has a counterpart in ${\sf AffBilin}(n,\ell)$:
\begin{itemize}
    \item zoom-ins: $\C_{a,b} = \{T \in \T_{n, \ell} \; : \; T(a) = b\}$, for $a\in \Ff_q^{\ell}$ and $b \in \Ff_q^n$.
    \item zoom-outs: $\C_{a^T, b, \beta}  = \{(M, c) \in \T_{n, \ell} \; : \; a^T \cdot M = b, \; a^T \cdot c = \beta \}$, for $a \in \Ff_q^{n}$, $b \in \Ff_q^{\ell}$, and $\beta \in \Ff_q$.
    \item zoom-ins on the linear part: $\C_{a,b, \lin} = \{(M, c) \in \T_{n, \ell} \; : \; M\cdot a = b\}$, for $a \in \Ff_q^{\ell}$ and $b \in \Ff_q^n$.
    \item zoom-outs on the linear part:  $\C_{a^T,b, \lin}  = \{(M, c) \in \T_{n, \ell} \; : \; a^T \cdot M = b\}$, for $a, \in \Ff_q^{n}$, and $b \in \Ff_q^{\ell}$.
\end{itemize}
Likewise, one can check that each example here also has expansion at most $1-1/q$ in ${\sf AffBilin}(n,\ell)$. 
Our argument will require us to show that, in a sense, these examples exhaust all small sets of vertices in ${\sf AffBilin}(n,\ell)$ whose edge expansion 
is at most $1-1/q$.
To formalize this, we first define the notion of pseudo-randomness. Intuitively, we say that a set $S$ is pseudo-random with respect to some example -- say zoom-ins for concreteness -- if $S$ may only have little correlation with $\mathcal{C}_{a,b}$'s, in the sense that 
the density of $S$ inside such sets is still small. More precisely:

\begin{definition}
    Let $\Sc \subseteq \T_{n, \ell}$ and let $\xi \in [0,1]$.
    \begin{enumerate}
        \item We say that $\Sc$ is $\xi$-pseudo-random with respect to zoom-ins if for each $a \in \Ff_q^{\ell}$ and $b \in \Ff_q^n$ we have that 
        \[
        \mu(\Sc_{a,b}) := \Pr_{T \in \C_{a,b}}[T \in \Sc] \leq \xi.
         \]
         \item We say that $\Sc$ is $\xi$-pseudo-random with respect to zoom-outs if for each $a \in \Ff_q^{n}, b \in \Ff_q^{\ell},$ and $\beta \in \Ff_q$ we have that 
        \[
        \mu(\Sc_{a^T,b, \beta}) := \Pr_{T \in \C_{a^T,b, \beta}}[T \in \Sc] \leq \xi.
         \]
         \item We say that $\Sc$ is $\xi$-pseudo-random with respect to zoom-ins on the linear part if for each $a \in \Ff_q^{\ell}$ and $b \in \Ff_q^n$ we have that
        \[
        \mu(\Sc_{a,b, \lin}) := \Pr_{T \in \C_{a,b, \lin}}[T \in \Sc] \leq \xi.
         \]
         \item We say that $\Sc$ is $\xi$-pseudo-random with respect to zoom-outs on the linear part if for each $a \in \Ff_q^{n}, b \in \Ff_q^{\ell}$  we have that 
        \[
        \mu(\Sc_{a^T,b, \lin}) := \Pr_{T \in \C_{a^T,b, \lin}}[T \in \Sc] \leq \xi.
         \] 
    \end{enumerate}
\end{definition}

Typically, global hypercontractivity results say that a small set with expansion bounded away from $1$ cannot be pseudon-random (where the notion of pseudo-randomness is in the same spirit but a bit more complicated; see~\cite{DKKMS, KMS} for example). 
For results in coding theory however it seems that a somewhat different (yet very much related) type of result is needed~\cite{KM}. 
Intuitively, in these type of results 
the assumption on the expansion of the set $\Sc$ is much stronger, and roughly speaking asserts that the expansion of $\Sc$ is almost ``as small as it could be''. In exchange for such a strong assumption, one often requires a stronger conclusion regarding the non-pseudo-randomness of the set $\Sc$. For instance, in~\cite{KM} it is shown that 
a set of vertices $\Sc$ with expansion at most $1-1/q$ which is highly pseudo-random with respect to $3$ of the above type of examples, must be highly non-pseudo-random with respect to the fourth type, in the sense that it almost contains a copy of such set.
For our purposes we require an analogous statement for ${\sf AffBilin}(n,\ell)$, which is the following statement:
\begin{theorem} \label{th: pseudorandom}
If $\Sc \subseteq \T_{n, \ell}$ satisfies 
\begin{enumerate}
    \item $\mu(\Sc) \leq \xi$,
    \item $\Sc$ is $\xi$-pseudorandom with respect to zoom-outs, zoom-outs on the linear part, and zoom-ins on the linear part,
    \item $1 - \Phi(\Ac) \geq \frac{1}{q}$.
\end{enumerate}
Then there exist $a \in \Ff_q^{\ell}$ and $b \in \Ff_q^n$ such that $\mu(\Sc_{a,b}) \geq 1 - \frac{1}{(q-1)^2} - \frac{q^3}{q-1} \left(4q^{-\ell} + 867\xi^{1/4} \right)$, where $\mu(\Sc_{a,b})$ is the density of $\Sc$ in $\C_{a,b}$.
\end{theorem}
\begin{proof}
The proof uses a reduction to an analogous result over the affine Grassmann graph using ideas from~\cite{BKS}, and is deferred to Section~\ref{sec:pf_expansion}.
\end{proof}
\section{Local Testers for the Reed Muller Code}\label{sec:local_testers}

We now formally describe both the sparse flat test and the full flat test. Although we focus on the sparse test, at times it will be convenient to reduce to the full test to aid our analysis. 

\subsection{The $t$-flat Tester and the Inner Product View}
At a high level, both the $t$-flat test and the sparse $(s+t)$-flat test can be described in the following way:
\begin{enumerate}
    \item \textbf{Restriction:} Choose a random $T \in \T_{n, t}$ for some suitable dimension $t$.
    \item \textbf{Local test on restriction:} Check if the $t$-variate function $f \circ T$ is indeed degree $d$. If so, accept, and otherwise reject.
\end{enumerate}

The difference between the two tests lies in how the ``check'' in step 2 is done. The straightforward way to perform this check is by querying $f \circ T$ on all points in $\Ff_q^{t}$, interpolating $f \circ T$ and checking its degree. Indeed, 
this is precisely how the $t$-flat test is defined. In that case, it is easy to see that the result of the test depends only on the image of $T$; this is because $f \circ T \circ A$ and $f \circ T$ have the same degree for any full rank $A \in \T_{t, t}$. Therefore, the $t$-flat test can be rephrased in its familiar form as follows:
\begin{enumerate}
    \item Choose a random $t$-flat $U\subseteq \Ff_q^n$.
    \item Accept if and only if $\deg(f|_{U}) \leq d$. 
\end{enumerate}
However, there are other ways of trying to test whether $f\circ T$ has degree at most $d$ or not. One way to do so is by taking inner products that check whether certain high-degree monomials exist in $f$ using Lemma~\ref{lm: monomial inner product}. A simple way to do this is as follows:
\begin{center}
     Accept if and only if $\langle f \circ T, x^e \rangle = 0$ for all $e \in  \{0,\ldots, q-1\}^t$ such that $\sum_{i = 1}^{t} q-1-e_i > d$.
\end{center}
By Lemma~\ref{lm: monomial inner product} it is not hard to see that this condition is equivalent to $\deg(f \circ T) \leq d$. Furthermore it is clear that calculating all of the inner products requires querying every point in the support of some $x^e$ that is used -- which is $\Ff_q^t$ in this case. Hence, taking inner product with all of these monomials does not lead to any savings in terms of the query complexity.

It turns out that there are more clever choices of ``test polynomials'' (which are not just monomials) that allow one 
to design an inner-product test above which is more query efficient. This idea was used by~\cite{RZS} who 
introduced the sparse $(s+t)$-flat test, 
which we present formally in the next subsection.

\subsection{The sparse $(s+t)$-flat test} \label{sec: sparse flat test construction}
Our presentation of the sparse $(s+t)$-flat is somewhat different than in~\cite{RZS}, and this view will be necessary for our analysis. Recall that we write $d+1 = s\cdot (q-q/p) + r$ where $s$ is divisible by $p$ and $0 < r \leq p(q-q/p)$.
For a vector $x = (x_1,\ldots,x_p)$ and a set $I\subseteq \{1,\ldots,p\}$ denote $x_I = \sum\limits_{i\in I} x_i$. Define the $p$-variate polynomial
\begin{align*}
    &P(x_1, \ldots, x_p) = \frac{\sum_{I \subseteq [p-1]} (-1)^{|I|+1}(x_I + x_p)^{q-1}}{x_1 \cdots x_{p-1}}.
\end{align*}
For any degree sequence $e = (e_1, \ldots, e_t) \in [q]^t$ define
\[
    M_{e}(x_1, \ldots, x_t) = \prod_{i=1}^{t}x_i^{e_i},
    \qquad \qquad
    H_{e}(x) = \left(\prod_{i=1}^{s/p} P(x_{p(i-1)+1}, \ldots, x_{pi}) \right) M_e(x_{s+1}, \ldots, x_{s+t}).
\]
The sparse $(s+t)$-flat tester works as follows:
\begin{enumerate}
    \item Choose a random affine transformation $T: \Ff_q^n \xrightarrow[]{} \Ff_q^{s+t}$.
    \item Accept if and only if, $\langle f \circ T, H_{e} \rangle = 0$ for all $e \in  \{0,\ldots, q-1\}^t$ such that $\sum_{i=1}^{t} e_i \leq t(q-1) - r$. 
\end{enumerate}
Recall, we set the parameter $t$ to be $t = \max(p+2, 10)$. Essentially, $t$ just needs to be large enough so that the degree sequences in step $2$ can account for monomials of degree up to $r+q-1$. 

We refer to a degree sequence $e$ satisfying the inequality in step $2$ as \emph{valid} throughout. As $d$ is fixed throughout, the notion of valid degree sequences will also not change throughout the paper. Hence, we say that $T$ rejects $f$ or $f \circ T$ is rejected, if $\langle f \circ T, H_{e} \rangle \neq 0$ for some valid $e$. Otherwise we say that $T$ accepts $f$ or $f \circ T$ is accepted. We define $\Sc_t$ to be 
the set of $T$'s on which the tester rejects:
\begin{equation}\label{eq:rej_set}
\Sc_t = \left\{T\in{\sf AffBilin}(n,s+t)~|~\inner{f\circ T}{H_e}\neq 0
\text{ for some $e\in \{0,\ldots, q-1\}^t$ such 
that $\sum\limits_{i=1}^{t}e_i \leq t(q-1)-r$}\right\},
\end{equation}
and let $\rej_{s+t}(f)$ be the probability that the test rejects. Clearly, we have that $\rej_{s+t}(f) = \mu_{s+t}(\Sc_t)$

\subsection{Assuming $n$ is Sufficiently Large}
In this section we argue that without loss of generality we may assume that $n \geq s + t + 100$. Indeed, we may view an $n$-variate function $f$ as a function in some $N$-variables, for some sufficiently large $N = O(n)$. Call this function $g: \Ff_q^{N} \xrightarrow[]{} \Ff_q$ and define it by $g(a,b) = f(a)$ for any $a \in \Ff_q^{n}$, $b \in \Ff_q^{N-n}$. We show that $\delta(f, \RM[n, q,d]) = \delta(g, \RM[N, q, d])$, which implies that we can view the test as over $g$ but still have the $\delta$ in Theorem~\ref{thm:main} be $\delta(f, \RM[n,q,d])$. 

\begin{lemma}
    Let $f$ and $g$ be as defined above. Then $\delta(f, \RM[n, q,d]) = \delta(g, \RM[N, q, d])$.
\end{lemma}
\begin{proof}
    We first show that $\delta(f, \RM[n, q,d]) \geq \delta(g, \RM[N, q, d])$. Suppose $h: \Ff_q^n \xrightarrow[]{} \Ff_q$ is the degree $d$ function such that $\delta(f, h) = \delta(f, \RM[n, q,d])$. Define $h'$ as the extension of $h$ to $N$ variables (in the same way as $g$ is defined). Let $h'(\cdot, b)$ denote the $N$-variate function where the last $N-n$ variables are set to $b \in \Ff_q^{N-n}$. Define $g(\cdot, b)$ similarly. Note that for any $b$, $h'(\cdot, b) = h$, $g(\cdot, b) = f$ . We then have
    \[
    \delta(h', g) = \E_{b \in \Ff_q^{N-n}}\left[\delta\left(h'(\cdot, b), g(\cdot, b)\right)\right] = \E_{b\in \Ff_q^{N-n}}\left[\delta(h, f)\right] = \delta(f, \RM[n, q,d]).
    \]
    Since $h$ is degree $d$, so is $h'$, thus the above implies that $\delta(f, \RM[n, q,d]) \geq \delta(g, \RM[2n, q, d])$.

    For the other direction, suppose $h'$ is the degree $d$, $N$-variate function such that $\delta(g, \RM[N, q,d]) = \delta(g, h')$. Keeping the same notation as above, we have that
    \[
    \E_{b \in \Ff_q^{N-n}}[\delta\left(g(\cdot, b), h'(\cdot, b)\right)] = \delta(g, h') = \delta(g, \RM[N, q,d]).
    \]
    Hence, there is a choice of $b$ such that 
    \[
    \delta\left(f, h'(\cdot, b)\right) = \delta\left(g(\cdot, b), h'(\cdot, b)\right) = \delta(g, \RM[N, q,d]).
    \]
    Since $h'$ is degree $d$, $h'(\cdot, b)$ must be degree $d$ as well, so the above inequality implies that $\delta(f, \RM[n,q,d]) \leq \delta(g, \RM[N, q,d])$, completing the proof.
\end{proof}

Using this lemma, we can always view the test as over $g$ instead of $f$, and show that the sparse $s+t$-flat test rejects with probability at least $c(q) \min(1, Q\delta(g, \RM[N,q,d]))$. Where $Q$ is the number of queries and $c(q)$ is the some $\frac{1}{\poly(q)}$. Since none of these parameters depend on $N$ we can use the lemma above and get that this rejection probability is the same as $c(q) \min(1, Q\delta(f, \RM[n,q,d])$.

Henceforth we assume that $n \geq s + t + 100$. This assumption is helpful as it allows us to bound the number of non-full rank affine transformations from $\Ff_q^n \xrightarrow[]{} \Ff_q^{s+t}$ as 
we argue in the following remark.
\begin{remark} \label{rmk: non full rank}
    The fraction of $\T_{n, s+t}$ that is not full rank is at most $\frac{1}{q^{100}(q-1)}$. The same holds for the fraction of $\T_{n, s+t}$ that are not full rank conditioned on $Ta = b$ for arbitrary $a \in \Ff_q^{s+t}$, $b \in \Ff_q^n$. To see that, 
    note that we can upper bound the probability that the linear part of $M$ is not full rank by
    \[
    \sum_{i=0}^{s+t} \frac{q^{i-1}}{q^n} \leq \frac{q^{s+t}-1}{(q-1) q^n} \leq \frac{1}{q^{100}(q-1)}.
    \]
    For the second part, note that conditioning on $Ta = b$ does not change this estimate, as we can still choose the linear part uniformly at random and set the affine shift so that $Ta = b$.
\end{remark}

\subsection{Some Basic Facts of the Sparse Flat Tester} 
We now collect some basic facts regarding the sparse flat tester that will be necessary for our analysis.
\subsubsection{Basic Soundness Properties}
We begin by describing why this tester works. Our presentation herein will be 
partial, focusing on the most essential 
properties necessary for our analysis, and we refer the reader to Appendix~\ref{sec:prob_1_accept} or \cite{RZS} for more details. 

Consider the monomials $x^{e'}$, for $e' \in \{0,\ldots, q-1\}^{s+t}$ such that $\langle x^{e'}, H_{e}(x) \rangle \neq 0$ for some $e \in  \{0,\ldots, q-1\}^t$ such that $\sum_{i=1}^{t} e_i \leq t(q-1) - r$, and $t \geq p+2$. By Lemma~\ref{lm: power sums}, these monomials, $x^{e'}$, must satisfy:
\begin{itemize}
    \item $e'_{p(i-1)+1} + \cdots + e'_{pi} = p\cdot (q-q/p)$ for $1 \leq i \leq \frac{s}{p}$
    \item $e'_{s+1} + \cdots + e'_{s+t} \geq r$.
\end{itemize}

More generally, we can explicitly express any inner product with $H_e$ as follows.
\begin{lemma} \label{lm: inner product}
For $f(x) = \sum_{a \in  \{0,\ldots, q-1\}^n} C_a x^a$, and $e \in  \{0,\ldots, q-1\}^t$ we have that
\begin{equation*}
    \langle f, H_e \rangle = \sum_{e'} D_{e'} C_{e'},
\end{equation*}
where $D_{e'}$ is a constant depending on the coefficients of $H_e$, and the sum is over $e'$ satisfying both the first condition above and $e'_{s+i} + e_i = q-1$ for $1 \leq i \leq t$. 
\end{lemma}
\begin{proof}
We have
\[
\langle x^{e'}, H_e \rangle =  \left(\prod_{i=1}^{s/p} \sum_{\alpha_1, \ldots, \alpha_p \in \Ff_q} \alpha_1^{e'_{p(i-1)+1}} \cdots \alpha_p^{e'_{pi}}P(\alpha_1, \ldots, \alpha_p)\right) \sum_{\beta_1, \ldots, \beta_t \in \Ff_q^t} \beta_{1}^{e'_{s+1} + e_1} \cdots \beta_{t}^{e'_{s+t} + e_t}.
\]
By Lemma~\ref{lm: monomial inner product}, for each $i$ we have
$$
 \sum_{\alpha_1, \ldots, \alpha_p \in \Ff_q} \alpha_1^{e'_{p(i-1)+1}} \cdots \alpha_p^{e'_{pi}}P(\alpha_1, \ldots, \alpha_p) \neq 0,
$$
only if the monomial $x_1^{q-1} \cdots x_p^{q-1}$ appears in $x_1^{e'_{p(i-1)+1}} \cdots x_p^{e'_{pi}}P(x_1, \ldots, x_p)$. Since the degree of $P$ is $q-p$, this can only be the case if $e'_{p(i-1)+1} + \cdots + e'_{pi} = p\cdot (q-q/p)$. Likewise, 
\[
 \sum_{\beta_1, \ldots, \beta_t \in \Ff_q^t} \beta_{1}^{e'_1 + e_1} \cdots \beta_{t}^{e'_t + e_t} \neq 0,
\]
only if $e'_i + e_i = q-1$ for each $1 \leq i \leq t$.
\end{proof}
From Lemma~\ref{lm: inner product} it is clear that $T$ rejects $f$ only if $\deg(f \circ T) > d$ as in this case, $f$ does not have any monomials satisfying the above. Therefore if $f$ is indeed degree $d$ then it is accepted with probability $1$. But is it
necessarily the case that the sparse flat 
tester rejects a function $f$ with positive probability if its degree exceeds $d$? This was shown to be true 
in~\cite{RZS} using canonical monomials. As we are going to need this fact and so as to be self contained, we include a proof in Appendix~\ref{sec:prob_1_accept}.
\begin{thm}\label{thm:perfect_accept}
 A function $f$ passes the sparse $(s+t)$-flat test with probability $1$ if and only if $f$ has degree at most $d$.
\end{thm}
\begin{proof}
    Deferred to Appendix~\ref{sec:prob_1_accept}.
\end{proof}
\subsubsection{Sparsity}
The next important feature of the sparse 
flat tester, as the name suggests, is that
it has a small query complexity. Note that the 
number of queries the test makes is proportional to the size of the support of the test polynomials $H_e$, and the next lemma shows that for any of the $e$ in $\Ff_q^t$ of step 2, $H_e(x)$ has a sparse support. Moreover, this support is the same no matter which $e$ is chosen. As a result, calculating $\langle f \circ T, H_{e} \rangle$ does not require querying all $q^{s+t}$ points in the domain of $f \circ T$.

\begin{lemma} \label{lm: support of P}
    The support of $P$ has size at most $(2^{p-1}+p-1)q^{p-1}$ and is contained in the set
    \[
    \left( \bigcup_{i=1}^{p-1}\{x_i = 0\} \right)  \cup \left( \bigcup_{I \subset [p-1]} \{x_I + x_p = 0\} \right).
    \]
    where $\{x_I = x_p \}$ denotes the hyperplane given by $x_I = x_p$.
\end{lemma}
\begin{proof}
    Suppose $x$ is not in the set described. Then in the expression for $P$, each term $(x_I + x_p)^{q-1} = 1$ by Fermat's Little Theorem. Moreover, the denominator is nonzero, so a direct calculation yields $P(x) = 0$.
\end{proof}
\begin{lemma} \label{lm: support}
    For any $e \in  \{0,\ldots, q-1\}^{t}$, we have 
    \begin{equation*}
        \supp(H_e) \subseteq \prod_{i=1}^{s/p} \supp(P) \times \Ff_q^{t},
    \end{equation*}
    
\end{lemma}
\begin{proof}
    Since $H_e(x) = \left(\prod_{i=1}^{s/p} P(x_{p(i-1)+1}, \ldots, x_{pi})\right) \prod_{j=1}^{t}x_{s+j}^{e_{j}}$, the result is evident.
\end{proof}

Henceforth, we let $\supp(H) = \bigcup_{e \in  \{0,\ldots, q-1\}^t} \supp(H_e)$. For any $T$, the sparse $s+t$-flat test can be done by only querying $f\circ T$ on points in $\supp(H)$, which gives the following upper bound on query complexity.

\begin{lemma} \label{lm: query complexity}
The sparse $s+t$-flat test has query complexity at most $(3q)^{\frac{d+1}{q}}q^t$. Since we can choose any $t \geq p + 2$, the query complexity can be as low as $(3q)^{\frac{d+1}{q}}q^{p+2}$
\end{lemma}
\begin{proof}
By Lemma~\ref{lm: support}, we can bound the size of $\prod_{i=1}^{s/p} \supp(P) \times \Ff_q^{t}$. By Lemma~\ref{lm: support of P}, 
\[
|\supp(P)| \leq (2^{p-1} + p-1)q^{p-1},
\]
so 
\[
\left|\prod_{i=1}^{s/p} \supp(P) \times \Ff_q^{t} \right| \leq ((2^{p-1} + p-1)q^{p-1})^{s/p}q^t \leq (3q)^{s(p-1)/p}q^t,
\]
where we use the fact that $(2^{p-1} + p-1)^{1/p} \leq 3$. Moreover, $s \leq \frac{d+1}{q-q/p}$, so
\[
(3q)^{s}q^t \leq (3q)^{\frac{d+1}{q}}q^t.
\qedhere
\]
\end{proof}

\subsection{Relating the Sparse Flat Tester and the Flat Tester} \label{sec: relation to full flat tester}
In this section we describe a way to interpret the sparse $(s+t)$-flat test as a $t$-flat test. This relation will be useful later on as it allows us to borrow techniques from the
analysis of the $t$-flat test from~\cite{KM}. 

Fix an affine transformation $A \in \T_{n, s+\ell}$ for some $\ell > t$ and let $\Res_{s, \ell, t}$ denote the set of affine transformations $(R, b)$ of the following form,
\begin{equation} \label{eq: B2}
 R =  \begin{bmatrix}
   I_s & 0 \\
   0 & R'
   \end{bmatrix}, b = \begin{bmatrix} 0 \\ b' \\ \end{bmatrix},
\end{equation}
for $R' \in \Ff_q^{\ell \times t}, b' \in \Ff_q^{\ell}$. We call the affine transformations in $\Res_{s, \ell, t}$ restrictions. In words, composing $A$ with a random  $(R,b)$ in $\Res_{s, \ell, t}$ corresponds 
to the operation of preserving the first 
$s$ columns, and randomizing the rest of them. As we explain below, this allows 
to view the $(s+t)$-sparse test as having 
the $t$-flat tester embedded within it 
on some restricted-type function of $f$.

More precisely, note that we can sample $T \in \T_{n, s+t}$ by choosing 
$A\in \T_{n, s+\ell}$ as well as 
a restriction $B = (R, b) \in \Res_{s, \ell, t}$ uniformly at random and outputting $T = A \circ B$. In Lemmas~\ref{lm: generic test} and \ref{lm: invariances} we show that the result of the test is ``entirely dependent'' on $b' + \im(R') = \im((R', b'))$. To this end, after fixing $A$ we define $\tilde{f}: \Ff_q^{\ell} \xrightarrow[]{} \Ff_q$ by
\[
\tilde{f}(\beta_1, \ldots, \beta_{\ell}) = \sum_{\alpha \in \Ff_q^s} f \circ A(\alpha, \beta_1, \ldots, \beta_{\ell}) \prod_{i=1}^{s/p} P(\alpha_{p(i-1)+1},\ldots, \alpha_{pi}).
\]
Also for $B = (R, b) \in \Res_{s, \ell, t}$ let $\fl(B) \subseteq \Ff_q^{\ell}$ be the $t$-flat given by $b' + \im(R')$. 
The following lemma gives a relation 
between the sparse flat tester rejecting $f \circ T$, and the (standard) flat tester
rejecting $\tilde{f}$ on $b'+\im(R')$.
\begin{lemma} \label{lm: generic test}
For any $B \in \Res_{s, \ell, t}$, $A \circ B$ rejects $f$ if and only if $\deg(\tilde{f}|_{\fl(B)}) \geq r$. 
\end{lemma}
\begin{proof}
Write $B = (R,b)$ in the form of~\eqref{eq: B2} and let $B': \Ff_q^{t} \xrightarrow[]{} \Ff_q^{\ell}$ be the affine transformation given by $(R', b')$. Recall that $A \circ B$ rejects if and only if 
\[
\langle f \circ A \circ  B, H_e \rangle = \langle \tilde{f} \circ B, H_e \rangle \neq 0.
\]
for some $e \in \{0,\ldots, q-1\}^t$ such that $\sum_{i=1}^{t} e_i \leq t(q-1) - r$. For any $e$, we can rewrite this inner product on the right hand side as follows.
\begin{align*}
    \langle f \circ A \circ B, H_e \rangle &= \sum_{\alpha \in \Ff_q^s, \beta \in \Ff_q^{t}} \left(f \circ A \circ B(\alpha, \beta) \right) \cdot \prod_{i=1}^{s/p} P(\alpha_{p(i-1)+1},\ldots, \alpha_{pi}) \beta^e \\
    &=   \sum_{\beta \in \Ff_q^t} \left(\sum_{\alpha \in \Ff_q^{s}} f(\alpha, B'(\beta))\cdot \prod_{i=1}^{s/p} P(\alpha_{p(i-1)+1},\ldots, \alpha_{pi}) \right) \cdot \beta^e \\
    &= \sum_{\beta \in \Ff_q^t} \Tilde{f} \circ B'(\beta) \cdot \beta^e \\
    &=   \langle \tilde{f} \circ B',  x^e \rangle.
\end{align*}
 However, $\langle \tilde{f} \circ B',  x^e \rangle \neq 0$ if and only if $\tilde{f} \circ B'$ contains the monomial $x_1^{q-1-e_1} \cdots x_{t}^{q-1-e_t}$. It follows that $A \circ B$ rejects if and only if $\deg(\tilde{f}\circ B') \geq r$. Finally, since the degree of a polynomial is invariant under affine transformations, it follows that this is equivalent to $\deg(\tilde{f}|_{\fl(R)}) \geq r$. 
\end{proof}
We remark that Lemma~\ref{lm: generic test} in particular implies that the sparse $s+t$-flat test is invariant under any affine transformation that only affects the last $t$-coordinates. In particular, 
we get:
\begin{lemma} \label{lm: invariances}
    For any $B \in \Res_{s,t,t}$, $T \circ B$ rejects $f$ if and only if $T$ rejects $f$.
\end{lemma}
\begin{proof}
    We apply Lemma~\ref{lm: generic test} with $\ell = t$ and $A = T$. Then, applying $B$ does not change the degree of $\tilde{f}$ and the result follows.
\end{proof}

After fixing $A \in \T_{n, s+ \ell}$ this lemma allows us to think of the remainder of the sparse $(s+t)$-flat test, i.e.\ choosing a restriction $B \in \Res_{s, \ell, t}$, as a $t$-flat test on $\tilde{f} = f \circ A$. Setting $\ell = t+100$, we can view the sparse $s+t$-flat test as follows.

\begin{enumerate}
    \item Choose a random $A \in \T_{n, s+t+100}$.
    \item Perform the standard $t$-flat test on the $t+100$-variate function $\Tilde{f}$ defined according to the above.
\end{enumerate}

This view of the test will allow us to borrow some concepts and facts from the $t$-flat test which we now introduce. First, we formally define the upper shadow for flats.

\begin{definition}
    For a set of $\ell$-flats, $S$, define the upper shadow as follows
    \[
    S^\uparrow = \{V \; |\; \dim(V) = \ell+1, \exists U \in S, U \subseteq V \}.
    \]
\end{definition}

The following shadow lemma is shown in \cite{BKSSZ} and is used  extensively in the analysis of \cite{KM}. This lemma will play a role in our analysis as well, so we record it below.

\begin{lemma} \label{lm: uppershadow}
    For a function $f$, let $S$ be the set of $\ell$-flats, $U$, such that $\deg(f|_{U})> d$. Then,
    \[
    \mu(S^{\uparrow}) \leq q \mu(S),
    \]
    where the measures are in the sets of $(\ell+1)$-flats and $\ell$-flats respectively.
\end{lemma}
We can also define an analogous notion of upper shadow for affine transformations that fits with the view of the sparse flat tester 
just introduced.

\begin{definition}
    For a set of affine transformations $\Sc \subseteq \T_{n, s+\ell}$ define 
    \[
    \Sc^\uparrow = \{T \in \T_{n, s+\ell+1} \;|\; 
    \exists R \in \Res_{s+\ell+1, s+\ell}, T \circ R \in \Sc
    \}
    \]
\end{definition}
We will need the following simple result about the upper shadow of sets of affine transformations.
\begin{lemma}
    For a set of affine transformations $\Sc \subseteq \T_{n, s+\ell}$, 
    \[
    \mu(\Sc^{\uparrow}) \geq \mu(\Sc).
    \]
\end{lemma}
\begin{proof}
We can sample $T$ by first sampling $T' \in \T_{n, s+\ell+1}$ and then choosing a restriction $R \in \Res_{s+\ell+1, s+\ell}$ and outputting $T = T' \circ R$. We can only have $T \in \Sc$ if $T' \in \Sc^{\uparrow}$, so the result follows.
\end{proof}

\section{Locating a Potential Error}
We now begin our proof of Theorem~\ref{thm:main}. Recall $\Sc_t = \{T \in \T_{n, s+t} \; | \; T \text{ rejects } f \}$ defined in~\eqref{eq:rej_set} and that $\rej_{s+t}(f) = \mu_{s+t}(\Sc_t)$; we 
denote $\eps = \mu_{s+t}(\Sc_t)$ for simplicity.
In this section we show that if $\eps \leq q^{-M}$ (where $M$ is a large absolute constant to be determined), then we may find a potential 
erroneous input $x^{\star}$ for $f$, in the sense that the sparse flat tester almost always rejects if the chosen 
test polynomial has $x^{\star}$ in 
its support.

We begin by checking that the set $\Sc_t$ 
satisfies the conditions of Theorem~\ref{th: pseudorandom}.
\subsection{The Edge Expansion of $\Sc_t$ in the Affine Bi-linear Scheme}
We start by proving an upper bound on the expansion of $\Sc_t$. For a matrix $M$ and column $v$ vector $v$, let $[M, v]$ denote the matrix with $v$ appended as an additional column. Consider the following procedure for sampling an edge in $\AffShort(n, s+t)$. 
\begin{enumerate}
    \item Choose $T_1 = (M, c) \in \T_{n, s+t}$ uniformly at random.
    \item Choose $v \in \Ff_q^n$ uniformly at random conditioned on $v \neq 0$ and let $T' = ([M, v], c)$.
    \item Choose a uniformly random matrix $R \in \Ff_q^{(s+t+1) \times (s+t)}$ of the form $R = [I_{s+t}, w]^T$, and a uniformly random $b \in \Ff_q^{s+t+1}$ such that the first $s+t$ entries in $b$ are zero. 
    \item Set $T_2 = T' \circ (R,b)$.
\end{enumerate}

The following lemma will allow us to show $\Sc_t$ is poorly expanding. This takes the place of the ``shadow'' lemma from \cite{KM}.
\begin{lemma} \label{lm: expansion general}
Let $G: \Ff_q^{s+t} \xrightarrow[]{} \Ff_q$ be an arbitrary polynomial and let $T = (M, c)$ be an $s+t$-affine transformation such that $\langle f \circ T, G \rangle \neq 0$. Fix $v \in \Ff_q^n$ chosen in step $2$, and let $T' = ([M, v], c)$. Then,
\begin{equation*}
    \Pr_{(R,b)}[\langle f \circ T' \circ (R,b), G \rangle \neq 0] \geq \frac{1}{q},
\end{equation*}
where $R$ is sampled according to step 3.
\end{lemma}
\begin{proof}
Let $w = (w_1, \ldots, w_{s+t})$ be the last row of $R$ and let $b_{s+t+1}$ be the last entry of $b$. Choosing $(R,b)$ uniformly at random amounts to choosing $\beta$ and each $\alpha_i$ uniformly at random in $\Ff_q$. To obtain the result we will view $\langle f \circ T' \circ R, G \rangle$ as a function in these random values and apply the Schwartz-Zippel lemma.

The function $ f\circ T' \circ (R, b)$ can be obtained by composing the $(s+t+1)$-variable function, $\tilde{f} = f\circ T'$, with $(R, b)$. Then $ f\circ T' \circ (R, b)$ is just the $(s+t)$-variable function, 
\begin{equation*}
    \left(f\circ T' \circ (R, b)\right)(x_1, \ldots, x_{s+t}) = \tilde{f}\left(x_1, \ldots, x_{s+t}, \sum_{i = 1}^{s+t}w_i x_i + b_{s+t+1} \right).
\end{equation*}

By Lemma~\ref{lm: inner product},  $\langle f \circ T' \circ (R,b), G \rangle$ is some linear combination of the coefficients of $\tilde{f}$, but these coefficients are polynomials in $w_1, \ldots, w_{s+t}, b_{s+t+1}$ of total degree at most $q-1$, because $\tilde{f}$ can be written as a polynomial with individual degrees at most $q-1$. Thus, $\langle f \circ T' \circ R, G \rangle$ is a polynomial in $w_1, \ldots, w_{s+t}, b_{s+t+1}$ of total degree at most $q-1$. Moreover, this polynomial is nonzero because when setting $w_1 = \cdots = w_{s+t} = b_{s+t+1}= 0$, it evaluates to
\begin{equation*}
    \langle f \circ T, G \rangle \neq 0.
\end{equation*}
By the Schwartz-Zippel lemma it follows that
    $\Pr_{(R,b)}[\langle f \circ T' \circ R, G \rangle \neq 0] \geq \frac{1}{q}$.
\end{proof}
As an immediate corollary, we get that the set of affine transformations that reject is poorly expanding.
\begin{lemma} \label{lm: expansion specific}
The set of $s+t$-affine transformations that reject $f$, $\Sc_t$, satisfies
\begin{equation*}
   1 - \Phi(\Sc_t) \geq \frac{1}{q}.
\end{equation*}
where the expansion is in $\AffShort(n, s+t)$.
\end{lemma}
\begin{proof}
    Fix $T_1 \in \Sc_t$ and choose $T_2$ adjacent to $T_1$ in $\AffBilin(n, s+t)$. By assumption there is some valid $e$ and $H_e$ such that $\langle f \circ T_1, H_e \rangle \neq 0$. Since $T_2$ can be chosen according to the process described at the start of the section, we can apply Lemma~\ref{lm: expansion general} with $G = H_e$, 
    \[
    \Pr_{T_2}[\langle f \circ T_2, H_e \rangle \neq 0] \geq \frac{1}{q}.
    \]
    It follows that $ 1 - \Phi(\Sc_t) \geq \frac{1}{q}$.
\end{proof}

\subsection{Pseudorandom with respect to zoom-outs and zoom-outs on the linear part}
Next, we show that the set $\Sc_t$ is 
pseudorandom with respect to zoom-outs 
and zoom-out on the linear part.
\begin{lemma} \label{lm: zoom-out}
The set $\Sc_t$ is $q\mu(\Sc_t)$-pseudorandom with repsect to zoom-outs and zoom-outs on the linear part.
\end{lemma}
\begin{proof}
We show the proof for zoom-outs. The argument for zoom-outs on the linear
part is very similar.

Fix a zoom-out $\C_{a^T, b, \beta}$ denote by $\mu_{a^T, b, \beta}(\Sc_t)$ the measure of $\Sc_t$ in $\C_{a^T, b, \beta}$, 
that is, 
 $\frac{|\Sc_t \cap \C_{a^T, b, \beta}|}{|\C_{a^T, b, \beta}|}$. Fix an arbitrary $v$ such that $a^T \cdot v \neq 0$. Sample an $(s+t)$-affine transformation by first choosing $(M, c) \in \C_{a^T, b, \beta}$ uniformly at random, and then $(M', c')$ according to steps 3 and 4 with $v$. That is, if the $i$th $M$ column is $M_i$, then the $i$th column of $M'$ is $M_i + \alpha_i v$, and the affine shift is $c' = c + \alpha_0 v$ for uniformly random $\alpha_i's$ and $\alpha_0$ in $\Ff_q$. By Lemma~\ref{lm: expansion general}, if $T_1 \in \Sc_t \cap \C_{a^T, b, \beta}$, then $T_2 \in \Sc_t$ with probability at least $1/q$, so
\begin{equation*}
    \Pr[T_2 \in \Sc_t] \geq \frac{\mu_{a^T,b,\beta}(\Sc_t)}{q}.
\end{equation*}

On the other hand, notice that the distribution of $T_2$ is uniform over $\T_{n, s+t}$. Indeed, fix a $T = (M', c') \in \T_{n, s+t}$ and let $M'_i$ denote the $i$th column of $M'$. Then the only way to get $T_2 = T$ is to choose, 
\[
\alpha_i = \frac{a^T M'_i - b_i}{a^T v} \text{ for } 1 \leq i \leq n, \quad \text{ and }\quad \alpha_0 = \frac{a^Tc' - \beta}{a^Tv},
\]
and $(M, c) \in \C_{a^T, b, \beta}$ such that
\[
M_i = M'_i - \alpha_i v \text{ for } 1 \leq i \leq n, \quad \text{ and }\quad c = c' - \alpha_0v.
\]
Since $T_2$ is uniform over $\T_{n, s+t}$, we have
\[
  \mu(\Sc_t)  \geq \Pr[T_2 \in \Sc_t] \geq \frac{\mu_{a^T,b,\beta}(\Sc_t)}{q},
\]
and hence $\mu_{a^T,b,\beta}(\Sc_t) \leq q \mu(\Sc_t)$. Since this applies for all zoom-outs, the result follows. A similar argument works for zoom-outs on the linear part.

\end{proof}

\subsection{Pseudorandomness with respect to zoom-ins on the linear part} \label{sec: pseudorandomness zoom-in linear}
We next show pseudorandomness with respect to zoom-ins on the linear part. The argument is more involved. In particular, we use the relation to the $t$-flat test to reduce the proof of this statement to analogous 
statements on affine Grassmann graphs. 
We then prove this statement using ideas similar to~\cite{KM}.
Formally, in this section we show:
\begin{lemma} \label{lm: lin part pseudo overall}
    The set $\Sc_t$ is $tq^{162}\eps$ pseudorandom with respect to zoom-ins on the linear part $C_{a,b,\lin}$.
\end{lemma}
We start with a few definitions.
For a vector $a$ (of arbitrary length greater than $s$), let $a_{[s]}$ be the vector that is equal to $a$ in its first $s$ coordinates and $0$ in all of its other coordinates. Let $a|_{[s]}$ be the vector $a$ restricted to its first $s$ coordinates. For a matrix $M \in \Ff_q^{n \times (s + \ell)}$, denote 
\[
\im_a(M) = \{Ma' \; | \; a' \in \Ff_q^{\ell}, a'|_{[s]} = a|_{[s]}, a' \neq a'_{[s]}  \}.
\]
In words, $\im_a(M)$ is the image of elements $a'$ that agree with $a$ on 
the first $s$ coordinates and are not identically $0$ on the rest of the coordinates.

\skipi
Due to the asymmetry induced by the test, 
we have a different argument depending on 
the $0$ pattern of $a$. We will first show that $\Sc_t$ is pseudorandom with respect to zoom-ins on the linear part, $C_{a,b,\lin}$, where $a \neq a_{[s]}$,
by a direct argument. We will then prove
that $\Sc_t$ is pseudorandom with respect
to any zoom-in on the linear part by a 
reduction to this case.

\begin{lemma} \label{lm: lin part pseudo case}
The set $\mathcal{S}_t$ is $q^{161}\eps$ pseudorandom with respect to zoom-ins on the linear part $C_{a,b,\lin}$ with $a \neq a_{[s]}$.
\end{lemma}

The proof of Lemma~\ref{lm: lin part pseudo case} requires some set-up and 
auxiliary statements, which we present next. Fix $a,b$ as in the lemma and suppose that $\Sc_t$ has $\alpha$ fractional size in $\C_{a,b, \lin}$, that is,
\begin{equation*}
    \frac{\mu(\mathcal{S}_t \cap \C_{a,b, \lin})}{\mu(\C_{a,b, \lin})} = \alpha.
\end{equation*}
 Consider another $a' \in \Ff_q^{s+t}$ such that $a'_{[s]} = a_{[s]}$, and $a'\neq a'_{[s]}$. In words $a'$ is equal to $a$ in its first $s$ coordinates, and nonzero in its last $t$ coordinates. There is an invertible matrix $R$ whose first $s$ rows and columns form the identity matrix, $I_s$, such that $R a' = a$. Composition with $R$ gives a bijection between the sets $\C_{a,b, \lin}$ and $\C_{a',b, \lin}$. Moreover, by Lemma~\ref{lm: invariances}, $(M, c) \in \mathcal{S}_t$ if and only if $(M\cdot R, c) \in \mathcal{S}_t$. Therefore, for any $a'$ satisfying  $a'_{[s]} = a_{[s]}$, and $a'\neq a'_{[s]}$ it holds that
\begin{equation*}
    \frac{\mu(\Sc_t \cap \mathcal{C}_{a',b,t})}{\mu(\mathcal{C}_{a',b,t})} = \frac{\mu(\Sc_t \cap \mathcal{C}_{a,b,t})}{\mu(\mathcal{C}_{a,b,t})} = \alpha.
\end{equation*}
Put another way, $\mathcal{S}_t$ has fractional size $\alpha$ in the set of $(M, c)$ such that $b \in \im_a(M)$, and thus we define 
\[
\D_{a,b,\ell} = \{(M, c) \in \T_{n, s+\ell} \; | \; b \in \im_a(M) \}.
\]

We can sample an $(M', c') \in \D_{a,b,t}$ by first choosing $(M, c) \in \D_{a, b,t+100}$, then choosing $(R, b) \in \Res_{s,t+100, t}$ uniformly at random, and outputting $(M', c') = (M, c) \circ (R, v) = (M\cdot R, c + Mv)$ conditioned on $M \cdot R \in \D_{a, b, t}$. Since the probability that $(M', c') \in \mathcal{S}_t$ is $\alpha$ and this can only happen if $(M, c) \in \mathcal{S}_t^{\uparrow^{100}}$, we get that 
\begin{equation*}
    \frac{\mu(\mathcal{S}_{b}')}{\mu(\D_{a,b,s+t+100})} \geq \alpha,
\end{equation*}
where $\mathcal{S}'_{b} =\mathcal{S}_t^{\uparrow^{100}} \cap \D_{a,b,t+100}$. Recall $\mathcal{S}_t^\uparrow$ is the set of $(s+t+1)$-affine transformations, $T$, for which there exists a restriction $R$ such that $T \circ R \in \Sc_t$, and $\mathcal{S}_t^{\uparrow^{100}}$ is the same operation applied $100$ times.

The following lemma shows that there is a way to fix the first $s$ columns of the test so that the rejection probability 
remains small:
\begin{lemma} \label{lm: zoom-in linear 1}
    There exists $(M, c) \in \mathcal{S}'_b$ such that choosing $(R,v) \in \Res_{t+100, t}$ uniformly at random and setting $M' = M \cdot R$, $c' = c + Mb$, we have
     \[
     \Pr_{B = (R, b) \in \Res_{s, t+100, t}}[(M, c) \circ B \in \Sc_t \; | \; b \notin \im_a(M')] \leq \frac{2}{\alpha}\epsilon.
     \]
\end{lemma}

\begin{proof}
    Choose a full rank $(s+t+100)$-affine transformation $(M, c)$ uniformly conditioned on $(M, c) \in \mathcal{S}'_b$, then choose an $(R,v) \in \Res_{s, t+100, t}$ uniformly and set $(M', c') = (M, c) \circ R$. For the remainder of this proof all probabilities are according to this distribution. Define the following events.
\begin{itemize}
    \item $E_1 = \{(M', c') \in \Sc_t\}$.
    \item $E_2 = \{(M, c) \in \mathcal{S}'_{b}\}$.
    \item $E_3 = \{ b \notin \im_{a}(M')\}$.
    \item $E_4 = \{ b \in \im_{a}(M)\}$.
\end{itemize}
First we note that
\[
\Pr[E_1 \; | E_3 \land E_4] = \Pr[E_1 \; | \; E_3],
\]
and as a result, 
\[
\Pr[E_1 \land E_3 \; | \; E_4]  = \Pr[E_1 \; | \; E_3 \land E_4] \cdot \Pr[E_3 \; | \; E_4] \leq \Pr[E_1 \; | \; E_3].
\]
We can then conclude
\begin{align*}
\Pr_{(M', c') = (M, c) \circ R}[E_1 \; | \; E_2 \land E_3] 
&=  \frac{\Pr[E_1 \land E_3 \land E_4]}{\Pr[E_2 \land E_3]} \\
&=  \frac{\Pr[E_1 \land E_3 ~|~ E_4]\Pr[E_4]}{\Pr[E_2 \land E_3]} \\
&\leq \frac{\Pr[E_1 \; | \; E_3] \Pr[E_4]}{\Pr[E_2 \land E_3]}\\
&= \frac{\Pr_{(M', c')}[E_1 \; |\; E_3]}{\Pr[E_2 \;| \; E_4] \Pr[E_3 \; | \; E_2]},
\end{align*}
where we use the fact that $E_2 \land E_4 = E_2$, and so $\Pr[E_2]/\Pr[E_4] = \Pr[E_2~|~E_4]$.

Notice that the numerator of the last fraction is at most $\epsilon$, while the denominator is at least $\alpha/2$, so this probability is at most $2\epsilon/\alpha$.  Thus, there exists $(M, c) \in \mathcal{S}'_{b}$ such that conditioned on this $(M, c)$, the probability $(M, c) \circ R \in S$ over $R \in \Res_{s, t+100, t}$ is at most $2\epsilon/\alpha$. 
\end{proof}

Fixing this $(M, c)$, the remaining test graph is now over $\Res_{s+t+100,s+ t}$ and is isomorphic $\AffShort(t+100, t)$. Moreover, by Lemma~\ref{lm: generic test}, we may focus solely on the flat given by $b' + \im(R')$ for $(R, b) \in \Res_{s, t+100, t}$. This allows us to define $\tilde{f}: \Ff_q^{t+100} \xrightarrow[]{} \Ff_q^{t}$ as done in Section~\ref{sec: relation to full flat tester}:
\[
\tilde{f}(\beta_1, \ldots, \beta_{t+100}) = \sum_{\alpha \in \Ff_q^s} f \circ A(\alpha, \beta_1, \ldots, \beta_{t+100}) \prod_{i=1}^{s/p} P(\alpha_{p(i-1)+1}, \ldots, \alpha_{pi}).
\]
By Lemma~\ref{lm: generic test} we now work with the standard $t$-flat test for $\tilde{f}$. The condition that $b \notin \im_a(M \cdot R)$ translates into the condition that the $t$-flat $U \subseteq \Ff_q^{t+100}$ does not contain the point $w \in \Ff_q^{t+100}$ equal the last $t+100$ coordinates of $a'$, where $a'$ is the unique point such that $Ma' = b$.

Translating Lemma~\ref{lm: zoom-in linear 1} gives
\begin{equation*} 
    \Pr_{B = z + A \in \AffGras(t+100, t)}[\deg(\tilde{f}|_{B}) \geq r \; | \; w \notin A] \leq \frac{2}{\alpha}\epsilon.
\end{equation*}
This is the same result shown at the start of the proof of \cite[Claim 3.4]{KM}, and the rest of the argument follows as therein. Let 
\[
\mathcal{B} = \{B = z + A \in \AffGras(t+100, t) \; | \; w \notin A,  \deg(\tilde{f}|_{A}) \geq r\}.
\]

\begin{lemma}\label{lem:cal_B_nonempty}
The set $\mathcal{B}$ is nonempty.
\end{lemma}
\begin{proof}
Since we chose $(M,c) \in \Sc'_b$, there is at least restriction $B$ such that $(M,c) \circ B$ rejects $f$ and hence there is a $t$-flat on which $\tilde{f}$ has degree greater than $r$. It follows that if we sample a random $t$-flat $z + A = B'' \subseteq B'$, then $\deg(f|_{B''}) \geq r$ with probability at least $1/q$. If $\mathcal{B}$ were empty, however, we would have $\deg(f|_{B''}) \geq r$ only if $w \in A$. In this case the probability that $w \in A$ is at most 
    \[
    \frac{q^t - 1}{q^{t+1}-1} < \frac{1}{q}.
    \]
    Therefore $\mathcal{B}$ is nonempty.
\end{proof}

We now proceed to the proof of Lemma~\ref{lm: lin part pseudo case}.

\begin{proof}[Proof of Lemma~\ref{lm: lin part pseudo case}]
   By Lemma~\ref{lem:cal_B_nonempty} $\mathcal{B}$ is non-empty, so there must be a $t+40$ flat $W$ such that

\begin{equation*}
    \mathcal{B}_{W} = \{B \in  \mathcal{B} 
    \; | \; B \subseteq W
    \}.
\end{equation*}
is nonempty. Let $\mu_W$ denote the uniform measure over $\AffGras(W, t)$. This graph is isomorphic to $\AffGras(t+40, t)$, but the ground space $\Ff_q^{t+40}$ is viewed as $W$. We first claim that since $\mu_W(\mathcal{B}_W) > 0$, it must be at least $q^{-100}$. If not, then $0 < \mu_W(\mathcal{B}_W) < q^{-100}$. However, $\mathcal{B}_W$ is $q^{-60}$ pseudo-random with respect to zoom-ins and zoom-ins on the linear part simply due to its size. Indeed, any zoom-in or zoom-in on the linear part already has measure $q^{-40}$, so $\mathcal{B}_W$ can only contain a $q^{-60}$ fraction. Furthermore, $\mathcal{B}_W$ is also $q^{-50}$ pseudo-random with respect to zoom-outs and zoom-outs on their linear part by a similar argument to Lemma~\ref{lm: zoom-out}. Finally, by a similar argument to Lemma~\ref{lm: expansion specific}, $1- \Phi_W(\mathcal{B}_W) \geq \frac{1}{q}$. Altogether, this then contradicts Theorem~\ref{th: pseudo aff gras}, so we conclude that $\mu_W(\mathcal{B}_W) \geq q^{-100}$. 

Thus, there exists $W$ such that $\mu_W(\mathcal{B}_W) \geq q^{-100}$. Fix such a $W$ and sample a uniform $t+99$-flat $Y = u + V \subseteq z + A$ conditioned on $w \notin V$, a uniform $t+60$-flat $A_2 \subseteq Y$, and consider $A_2 \cap W$. We may think of $W$ as being defined by a system of $60$ independent linear equations $\langle h_1, x \rangle = c_1, \ldots, \langle h_{60}, x \rangle = c_{60}$. That is, $W$ is the subspace of $\Ff_q^{t+100}$ that satisfies these $60$ equations. Likewise, $A_2$ is given by the restriction of $39$ linear equations,  $\langle h'_1, x \rangle = c'_1, \ldots, \langle h'_{39}, x \rangle = c'_{39}$. The probability that all $99$ linear equations are linearly independent is at least
\begin{equation*}
    \prod_{j=0}^{38} \frac{q^{99}-q^{60+j}}{q^{99}} \geq e^{-2 \sum_{j=1}^{\infty} q^{-j}} \geq e^{-4/q}.
\end{equation*}
When all $99$ linear equations are linearly independent, $A_2 \cap W$ is uniform over $\AffGras(W, t)$. Thus, 
\[
\Pr[A_2 \cap W \in \mathcal{B}_W] \geq e^{-4/q}q^{-100}.
\]
If $A_2 \cap W \in \mathcal{B}_W$, then $A_2 \in \mathcal{B}_Y^{\uparrow^{60}}$, where the upper shadow is taken with respect to $\AffGras(Y, t)$, so it follows that
\[
E_Y[\mu_Y(\mathcal{B}_Y^{\uparrow^{60}})] \geq e^{-4/q}q^{-100}.
\]
However, by Lemma~\ref{lm: uppershadow},

\[
\mu_Y(\mathcal{B}_Y^{\uparrow^{60}}) \leq q^{60}\mu_Y(\mathcal{B}_Y),
\]
so altogether we get that
\[
E_Y[\mu_Y(\mathcal{B}_Y)] \geq e^{-4/q}q^{-160}.
\]
To conclude, note that the left hand side is at most the probability that $f|_{B}$ has degree greater than $r$ over uniform $B = x + A \subseteq \Ff_q^{t+100}$  such that $w \notin A$. By assumption, this probability is at most $\frac{2}{\alpha}\epsilon$ so 
\[
\alpha \leq e^{-4/q}q^{160}2\epsilon.
\qedhere
\]
\end{proof}
We are now ready to prove Lemma~\ref{lm: lin part pseudo overall}.
\begin{proof}[Proof of Lemma~\ref{lm: lin part pseudo overall}]
    We show that the set $\mathcal{S}_t$ is $tq^{198}\eps$ pseudorandom with respect to zoom-ins on the linear part $C_{a,b,\lin}$ for any $a \in \Ff_{q}^{s+t}$. 
    
    If $a \neq a_{[s]}$, then we are done by Lemma~\ref{lm: lin part pseudo case}, so suppose that $a = a_{[s]}$, meaning $a$ is zero outside of its first $s$ coordinates. Clearly $a$ must have at least one nonzero coordinate, as otherwise $a = 0$ and $C_{a, b, \lin}$ is either all of $\T_{n, s+t}$ (if $b = 0)$, or empty (if $b \neq 0$). The former case cannot happen by the assumption that $\mu(\mathcal{S}_t) \leq q^{-M}$, while the latter case is trivially true. Without loss of generality suppose that $a_1 = \alpha \neq 0$. 
    
    For each $T \in \Sc_t \cap C_{a,b,\lin}$, there is an $e \in  \{0,\ldots, q-1\}^t$ satisfying $\sum_{i=1}^{t} e_i \leq t(q-1) - r$ such that $\langle f \circ T, H_e \rangle \neq 0$. Furthermore, as $r > 0$, there must be some ``special index'' $i$ such that $e_i < q-1$. Take the most common special index over all $T \in \Sc_t \cap C_{a,b,\lin}$ and without loss of generality suppose it is $1$. Thus, for at least $1/t$ of $T \in \Sc_t \cap C_{a,b,\lin}$, we have $\langle f \circ T, H_e \rangle \neq 0$ for some $e$ such that $e_1 < q-1$, and $\sum_{i=1}^{t} e_i \leq t(q-1) - r$. Let $\mathcal{A}$ denote the set of these transformations.

     Consider the map, $F_{\beta}: \Ff_q^{s+t} \xrightarrow[]{} \Ff_q^{s+t}$ that sends $x_{s+1}$ to $x_{s+1} + \beta x_1$ and keeps all other coordinates unchanged. That is,
    \[
    F_{\beta}(x_1, \ldots, x_{s+t}) = (x_1, \ldots, x_s, x_{s+1} + \beta x_1, x_{s+2}, \ldots, x_{s+t}).
    \]
    For any $T \in \mathcal{A}$, we claim that
    \[
    \Pr_{\beta \in \Ff_q}[\langle f \circ T \circ F_\beta, H_e \rangle \neq 0] \geq \frac{2}{q},
    \]
    where $e$ is the exponent vector such that $e_1 < q-1$, $\sum_{i=1}^{t} e_i \leq t(q-1) - r$, and $\langle f \circ T, H_e \rangle \neq 0$.

    Indeed, $\langle f \circ T \circ F_\beta, H_e \rangle$ is a linear combination of the coefficients of $f \circ T \circ F_\beta$ which is a polynomial in $\beta$. In fact, by Lemma~\ref{lm: monomial inner product}, it is a linear combination of coefficients of monomials where the degree of $x_{s+1}$ is $q-1-e_1 > 0$. In every monomial of  $f \circ T \circ F_\beta$, the degree of $\beta$ and the degree of $x_{s+1}$ add to at most $q-1$ however, so it follows that the coefficients who contribute to $\langle f \circ T \circ F_\beta, H_e \rangle$ have degree in $\beta$ at most $q-2$. Therefore $\langle f \circ T \circ F_\beta, H_e \rangle$ is a polynomial in $\beta$ of degree at most $q-2$. Finally this polynomial is nonzero because it evaluates to a nonzero value at $\beta = 0$. The inequality then follows from the Schwartz-Zippel lemma.

    In particular, this means for each $T \in \mathcal{A}$ there is a nonzero $\beta$ such that $\langle f \circ T \circ F_\beta, H_e \rangle \neq 0$, and we choose $\beta^{\star}$ to be the most common such $\beta$ (we remark that the fact that $\beta^{\star}\neq 0$ is the reason we needed probability of $2/q$ rather than $1/q$). Thus, for at least $2/q$ of the transformations in $\mathcal{A}$ it holds that $T \circ F_{\beta^{\star}} \in \Sc_t$. 
    Letting $a' = F_{\beta}^{-1}(a)$, we get that for at least $2/q$ fraction of the $T\in\mathcal{A}$ it holds that 
    $T\circ F_{\beta^{\star}} \in \mathcal{C}_{a', b, \lin}\cap \Sc_t$, and as $F_{\beta^{\star}}$ is a bijection it follows that
    \[
        \frac{1}{t}\mu(\Sc_t\cap \mathcal{C}_{a,b,\lin})
        \leq\mu(\mathcal{A})
        \leq \frac{q}{2} \mu(\Sc_t\cap \mathcal{C}_{a',b,\lin}).
    \]
    Finally, note that $a'\neq a'_{[s]}$ as its $(s+1)$ coordinate is non-zero by design. Applying Lemma~\ref{lm: lin part pseudo case} we get that  
    $\mu(\Sc_t\cap \mathcal{C}_{a',b,\lin})\leq 
    q^{161}\eps \mu(\mathcal{C}_{a',b,\lin})$, and as $\mu(\mathcal{C}_{a',b,\lin})=\mu(\mathcal{C}_{a,b,\lin})$, combining with the above we get that
    \[
        \frac{\mu(\Sc_t\cap \mathcal{C}_{a,b,\lin})}{\mu(\mathcal{C}_{a,b,\lin})}
        \leq tq^{162}\eps.
        \qedhere
    \]
\end{proof}

\section{Correcting the error and iterating} \label{sec: correcting the error and iterating}
The results from the previous subsections, along with the assumption that $\mu(\Sc_t) \leq q^{-M}$, establish that $\Sc_t$ satisfies the following properties:
\begin{enumerate}
    \item $\mu(\Sc_t) \leq q^{-M}$
    \item $1 - \Phi(\Sc_t) \geq \frac{1}{q}$.
    \item $\Sc_t$ is $q^{1-M}$-pseudorandom with respect to zoom-outs and zoom-outs on the linear part.
    \item $\Sc_t$ is $tq^{162-M}$-pseudorandom with respect to zoom-ins on the linear part.
\end{enumerate}
Since we take $t \geq 10$, it follows from Theorem~\ref{th: pseudorandom} with $\xi = tq^{162-M}$ and $\ell = s+t$ that there exists a pair $a, b$ such that
\begin{equation} \label{eqn: special point}
    \mu_{a,b}(\Sc_t) =  \frac{\mu(\Sc_t \cap \C_{a,b})}{\mu(\C_{a,b})} \geq 1 - \frac{1}{(q-1)^2} -\frac{1}{q^6}- 2000tq^{200}q^{-M/4},
\end{equation}
for a large enough $M$.

How shall we go about using this information? In words,~\eqref{eqn: special point} tells us that $\Sc_t$ is dense on transformations sending $a$ to $b$, and this suggests that the point $b$ is an erroneous point for $f$ which we should fix and, as a result, improve the acceptance probability of the test. Furthermore, it
stands to reason that this change should 
affect all tests that come from transformations in $\mathcal{C}_{a,b}$.

Upon inspection however, these tests can only be affected in the case that $a \in \supp(H)$; otherwise, in the test we perform, the value $f(T(a)) = f(b)$ is multiplied by $0$, and hence changing the value of $f$ at $b$ does not affect the test at all. Thus, for the above strategy to work, we must prove that~\eqref{eqn: special point} holds for a pair $a,b$ wherein $a$ is in ${\sf supp}(H)$. 

\begin{lemma} \label{lm: error located}
    There exists $a \in \supp(H)$ and $b \in \Ff_q^n$ such that $\mu_{a,b}(\Sc_{t}) \geq 1 - \frac{1}{(q-1)^2} -\frac{1}{q^6}- 2000tq^{200}q^{-M/4}$.
\end{lemma}

Once we have shown this lemma, we will be able to show that correction at $b$ can reduce the rejection probability of tests in $\C_{a,b}$; however, in order to show that the tester is optimal, we need to fix a larger fraction of tests with a single correction. Specifically, we need to be able to fix tests such that $T(a) = b$ for any $a \in \supp(H)$. In order for such an argument to work, we will have to show that $\Sc_t$ is dense in $\bigcup_{a \in \supp(H)}\C_{a,b}$.

\begin{lemma} \label{lm: support error}
    If there exists an $a \in \supp(H)$ such that $\mu_{a,b}(\Sc_t) \geq 1 -  \frac{1}{(q-1)^2} - \frac{1}{q^6}-2000tq^{200}q^{-M/4}$, then $\Sc_t$ has density at least $\geq 1 - \frac{1}{(q-1)^2} -\frac{1}{q^6}- \frac{2}{q^{100}(q-1)} - 2000tq^{200}q^{-M/4} - q^{-M/2}$ in $\cup_{a \in \supp(H)} \C_{a,b}$.
\end{lemma}
In words, $\cup_{a \in \supp(H)} \C_{a,b}$ is the set of $(s+t)$-affine transformations $T$ such that $Ta = b$ for some $a \in \supp(H)$. In the next two subsections we prove Lemmas~\ref{lm: error located} and \ref{lm: support error}. After showing these lemmas, the way to perform corrections will become obvious and we quickly conclude the proof of Theorem~\ref{thm:main}. Both Lemmas are proved in a similar manner, and we first present a general lemma that will be used in both proofs. Let $g: \Ff_q^{\ell + 100}$ be an arbitrary polynomial, let $d$ be a degree parameter, let $\nu$ denote uniform measure over $\ell$-flats in $\Ff_q^{\ell + 100}$, and let $b \in \Ff_q^{\ell + 100}$ be an arbitrary point. Define the following two sets:
\[
\mathcal{A} = \{U  \; | \;  \dim(U) = \ell, \deg(f|_U) > d, b \in U\},
\]
\[
\mathcal{B} = \{U \; | \;  \dim(U) = \ell, \deg(f|_U) > d, b \notin U\}.
\]
Keep $\epsilon$ and $M$ (the large absolute constant) the same as we have defined, so that $q^{M/2}O(\epsilon) < q^{-M/2}$ is small. We will use following result, which is an extension of the results in Section 3.2 of \cite{KM}:
 \begin{lemma} \label{lm: localization}
    Keep $g,d, b,$ and $\mathcal{B}$ as defined above and suppose $\ell \geq \max(\lceil \frac{d+1}{q-q/p} \rceil, 4)$. If $\nu(\mathcal{B}) \leq q^{M/2}O(\epsilon)$ then $\mathcal{B} = \emptyset$. Moreover there is a value $\gamma$ such that after changing $g(b)$ to $\gamma$, $\nu(\mathcal{A}) = 0$.
    
    As a consequence of these two points, $\deg(g) \leq d$ after changing the value of $g(b)$ to $\gamma$.
\end{lemma}
\begin{proof}
The proof is deferred to Section~\ref{sec: self-correct}. 
\end{proof}
\subsection{Proof of Lemma~\ref{lm: error located}}
To establish Lemma~\ref{lm: error located} we show that $\Sc_t$ is pseudorandom with respect to zoom-ins outside of the support, hence the zoom-in found in~\eqref{eqn: special point} must be in 
the support of $H$. Specifically, we show:
\begin{lemma} \label{lm: pseudo not in support}
    For any $a \notin \supp(H)$ and any $b \in \Ff_q^n$, $\mu_{a,b}(\Sc_t) < 1 - \frac{1}{(q-1)^2} -\frac{1}{q^6}- 2000tq^{200}q^{-M/4}$.
\end{lemma} 
We begin the proof of Lemma~\ref{lm: pseudo not in support}, and we assume for the sake of contradiction that the lemma is false. That is, suppose there is $a \notin \supp(H)$ such that $\mu_{a,b}(\Sc_t) \geq 1 - \frac{1}{(q-1)^2} - \frac{1}{q^6}-2000tq^{200}q^{-M/4}$. Since $a \notin \supp(H)$, there are $p$ consecutive coordinate $(a_{pi + 1}, \ldots, a_{pi + p-1}) \not\in \supp(P)$. It will be convenient to swap the order of the coordinates so that these are the last $p$ coordinates. Thus, the polynomial $H_e(x)$ for $e \in  \{0,\ldots, q-1\}^t$ is now,
\[
H_e(x_1, \ldots, x_{s+t}) = \left(\prod_{i=1}^{s/p-1}P(x_{p(i-1)+1}, \ldots, x_p)\right)x_{s-p+1}^{e_1}\cdots x_{s-p+t}^{e_t}P(x_{s+t-p+1}, \ldots, x_{s+t}).
\]
The test with affine transformation $T$ checks if $\langle f \circ T, H_e \rangle = 0$ for all valid $e$ with $H_e$ as defined above. Reordering the variables in this way changes the support, so that our assumption $a \notin \supp(H)$ becomes, without loss of generality, $(a_{s+t-p+1}, \ldots, a_{s+t}) \notin \supp(P)$, which will make our notation later a bit simpler (this is the only reason for reordering the variables). We will once again sample a larger affine transformation and choose a restriction to obtain a $T \in \T_{n, s+t}$, however the restrictions this time will be slightly different. We will require the restrictions to be full rank and we will fix the first $s+t-p$ columns/coordinates of $T$ as opposed to the first $s$ as before. Let $\Res_{s+t-p, p+102, p}$ consist of affine transformations $(R,b)$ of the form:

\begin{equation} \label{eqn: Res'}
 R =  \begin{bmatrix}
   I_{s+t-p} & 0 \\
   0 & R'
   \end{bmatrix}, b = \begin{bmatrix} 0 \\ b' \\ \end{bmatrix},
\end{equation}
where $R' \in \Ff_q^{(p+102) \times p}$ is full rank and $b' \in \Ff_q^{p+102}$. For $B =(R,b) \in \Res_{s+t-p, p+102, p}$ of this form, let $\fl(B) = b' + \im(R')$. Also, for an affine transformation $T$, define $\im_a(T) = \{Tx \; : \; x_{[s+t-p]} = a_{[s+t-p]} \}$, where recall $a_{[s+t-p]}$ is $a$ with all coordinates outside of the first $s+t-p$ set to zero. Sample a $T \in \T_{n, s+t}$ as follows:
\begin{enumerate}
    \item Choose a full rank $A \in \T_{n, s+t+102}$ such that $b \in \im_a(A)$.
    \item Choose $B \in \Res_{s+t-p, p+102, p}$.
    \item Output $T = A \circ B$.
\end{enumerate}
After $A$ is chosen, the first $s+t-p$ columns/coordinates of $A$'s linear part/affine shift are fixed, while the remaining parts are composed with some random restriction. Thus, once $A$ is fixed there is a unique $x^{\star}$ such that $Ax^{\star} = b$ and we can only have $Aa = b$ if $Ba = x^{\star}$. In particular, this only happens if $x^{\star} \in \im(B)$ where by design this $x^{\star}$ satisfies $x^{\star}_{[s+t-p]} = a_{[s+t-p]}$.  Setting $z^{\star} \in \Ff_q^{p+102}$ to be the last $p+102$ coordinates of $x^{\star}$, it follows that $Ta = b$ only if $z^{\star} \in \fl(B)$. Let $\mu_A$ denote measure in $\Res_{s+t-p, p+102, p}$. Define
\[
\mathcal{R}_A = \{B \; |\; A \circ B \text{ rejects}, z^{\star} \notin \fl(B) \}.
\]

If, after choosing $A$ and constructing $z^{\star}$ as above, we then condition on $z^{\star} \notin \fl(B)$ then $A \circ B$ is a uniformly random over full-rank transformations in $\T_{n, s+t}$ conditioned on its image not containing $b$. Therefore, $\E_{A}[\mu_A(\mathcal{R}_A)] \leq O(\epsilon)$, and with probability at least $1/2$ we have $\mu_A(\mathcal{R}_A) \leq O(\epsilon)$. 

On the other hand, if after choosing $A$ we choose $B$ uniformly such that  $A \circ B \in \C_{a,b}$, the distribution over $T$ is uniform over full rank $T \in \C_{a,b}$. By Remark~\ref{rmk: non full rank}, the fraction of affine transformations in $\C_{a,b}$ that are not full rank is at most $\frac{1}{q^{100}(q-1)}$ and by our assumption the set of rejecting transformations is dense in $\C_{a,b}$. Thus a simple averaging argument shows that the fraction of full rank $T \in \C_{a,b}$ that reject is also large, and in particular, strictly greater than $1/2$. Therefore with probability strictly greater than $1/2$ over $A$, there is at least one $B$ such that $A \circ B$ is full rank and rejects.

It follows that there exists a full rank $A \in \T_{n, s+t+102}$ such that the following two hold:
\begin{itemize}
    \item $\mu_A(\mathcal{R}_A) \leq O(\epsilon)$,
    \item There exists a $B^{\star}$ such that $T = A \circ B^{\star} \in \C_{a,b}$ is full rank and rejects.
\end{itemize}
Fix this $A$ and keep $x^{\star}$, $z^{\star}$, and $\mathcal{R}_A$ as defined. We now show how this leads to a contradiction. The idea is to apply Lemma~\ref{lm: localization} and argue that there is a point at which we can make a correction and cause $A \circ B$ to be accepted for all possible $B$, and in particular for $B^{\star}$. Since we assumed that there was high rejection probability on a zoom-in \textit{outside} the support, we will show that the correction is made at a point not looked at by the test $A \circ B^{\star}$. These two facts together form a contradiction because changing $f$ at a point that is not in the set of points $(A \circ B^{\star})(\alpha)$ for $\alpha \in \supp(H)$ cannot change the result of the test $A \circ B^{\star}$. In order to apply Lemma~\ref{lm: localization} however, we need some statement similarly to $\mu_A(\mathcal{R}_A) \leq O(\epsilon)$, but for a set of flats instead of a set of affine transformations. To this end, we will try to argue that the set of $\fl(B)$ for $B \in \mathcal{R}_A$ is small as well. 

We proceed with the formal argument. The first step is to define an auxiliary polynomial similar to that of Lemma~\ref{lm: generic test}. Suppose $T$ rejects because $\langle f \circ T, H_e \rangle \neq 0$ for a fixed valid $e$. Using this $e$, define $\tilde{f}: \Ff_q^{p+102} \xrightarrow[]{} \Ff_q$ by
\[
\tilde{f}(\beta) = \sum_{\alpha \in \Ff_q^{s-p+t}} f \circ A(\alpha_1, \ldots, \alpha_{s-p+t}, \beta) \left(\prod_{i=1}^{s/p-1}P(\alpha_{p(i-1)+1}, \ldots, \alpha_{pi})\right)\alpha_{s-p+1}^{e_1}\cdots \alpha_{s-p+t}^{e_t}.
\]
For any $B = (R,b) \in \Res_{s+t-p, p+102, p}$ written according to Equation~\eqref{eqn: Res'} it is easy to check that
\[
\langle f \circ A \circ B, H_e \rangle = \langle \tilde{f} \circ (R',b'), P \rangle,
\]
so after fixing $A$ we can view the test as being performed on $\tilde{f}$ with the transformation $(R',b') \in \T_{p+102, p}$. By Theorem~\ref{thm:perfect_accept}, if $\deg(\tilde{f}) < q(p-1)$, then $\langle \tilde{f} \circ (R',b'), P \rangle = 0$ for all $B' = (R', b') \in \T_{p+102, p}$. This leads to the following fact:
\[
B \text{ can only reject if } \deg(\tilde{f}|_{\fl(B)}) \geq q(p-1). 
\]
This fact is similar to Lemma~\ref{lm: generic test}, however, we only have one direction. Namely it is not true that $B$ always rejects if $ \deg(\tilde{f}|_{\fl(B)}) \geq q(p-1)$ because the test on $\tilde{f}$ only checks inner products with $P$ and not will all monomials of degree up to $q-p$.

Using this fact we can now relate $\mu_A(\mathcal{R}_A)$ to the measure of the set of $p$-flats $\fl(B)$ not containing $z^{\star}$ such that $\deg(\tilde{f}|_{\fl(B)}) \geq q(p-1)$. Define the set of $p$-flats,
\[
\mathcal{B}_A = \{\fl(B) \; | \; \deg(\tilde{f}_{\fl(B)}) \geq q(p-1), z^{\star} \notin \fl(B) \} \subseteq \AffGras(p+102, p).
\]
Equivalently, $\mathcal{B}_A$ is the set of all $p$-flats $U$ not containing $z^{\star}$ such that $\deg(\tilde{f})|_{U} \geq q(p-1)$.

To analyze the fractional size of this set, we can choose $(R',b')$ by choosing a $p$-flat, and then choosing a basis for the flat. More formally, with $A$ fixed,
 
\begin{enumerate}
    \item Choose $B \in \Res_{p+102, p}'$ such that $z^{\star} \notin \fl(B)$, and write $B$ according to Equation~\eqref{eqn: Res'}.
    \item Choose $T \in \T_{p,p}$, and replace $(R', b')$ with $(R', b') \circ T$. Let the resulting restriction be $B' \in \Res'_{p+102, p}$.
    \item Output $B'$.
\end{enumerate}
We claim that if initially $\deg(\tilde{f}|_{\fl(B)}) \geq q(p-1)$, then with probability at least $1/q$, the outputted $B'$ rejects.

\begin{lemma} \label{lm: random basis}
    Suppose $\deg(\tilde{f}|_{\fl(B)}) \geq q(p-1)$. Then with probability at least $1/q$ over $M \in \T_{p,p}$ in the second step, $B'$ rejects. That is, $\langle \tilde{f} \circ B', P \rangle \neq 0$.
\end{lemma}
\begin{proof}
    Write $B$ in the form of Equation~\eqref{eqn: Res'}. Then the assumption $\deg(\tilde{f}|_{\fl(B)}) \geq q(p-1)$ is equivalent to $\deg(\tilde{f} \circ (R', b')) \geq q(p-1)$.  

    Choose a full rank $T$ uniformly at random and consider
    \[
    \langle \Tilde{f} \circ (R', b') \circ T, P\rangle =  \langle \Tilde{f} \circ (R', b'), P \circ T^{-1} \rangle.
    \]
    Since $\deg(\tilde{f} \circ (R', b')) \geq q(p-1)$, Lemma~\ref{lm: invertible reject} implies that there is at least one choice of invertible $T$, and hence an $T^{-1}$, that makes the above nonzero. Therefore, we can view $\langle \Tilde{f} \circ (R', b'), P \circ T^{-1} \rangle$ as a nonzero polynomial in the entries of $T^{-1}$. Since the degree of $P$ is at most $q-p$, the inner product $\langle \Tilde{f} \circ (R', b'), P \circ T^{-1} \rangle$ is a nonzero polynomial in the entries of $T^{-1}$ of total degree at most $q-p$. By the Schwartz-Zippel Lemma, with probability at least $p/q$, over the entries of $T^{-1}$, $\langle \Tilde{f} \circ (R', b'), P \circ T^{-1} \rangle \neq 0$. Since $T^{-1}$ has to be invertible, we need to ignore the choices of entries that are non-invertible, but this is at most a $(p-1)/q$ fraction. Overall, we still get that with probability at least $1/q$ over $T \in \T_{p,p}$, $(R', b') \circ T$ rejects $\Tilde{f}$.
\end{proof}

Letting $\nu_A$ denote the uniform measure in $\AffGras(p+102, p)$, Lemma~\ref{lm: random basis} implies that $\nu_{\mathcal{A}}(\mathcal{B}_A) \leq q\mu_A(\mathcal{R}_A)$. We are now close to being able to apply Lemma~\ref{lm: localization}, but $\mathcal{B}_A$ is a set of $p$-flats, while Lemma~\ref{lm: localization} only works for sets of flats of dimension at least $4$. Thus in the $p=2,3$ cases, we cannot use this lemma. There is an easy fix however, which follows by looking at the upper shadow of $\mathcal{B}_A$
\begin{lemma} \label{lm: p+2}
    Let $\mathcal{B}'_A = \{U' \; | \; \dim(U') = p+2 , z^{\star}\notin U', \deg(\tilde{f}|_{U'})\} \geq q(p-1) \}$. Then,
    \[
    \nu(\mathcal{B}'_A) \leq q^2 \nu_A(\mathcal{B}_A),
    \]
    where $\nu$ is uniform measure in $\AffGras(p+102, p+2)$.
\end{lemma}
\begin{proof}
    Observe that $\mathcal{B}'_A \subseteq \mathcal{B}_A^{\uparrow^2}$. Indeed, for any $U' \in \mathcal{B}'_A$, there must be a $p$-flat $U \subseteq U'$ such that $\deg(\tilde{f}|_U) \geq q(p-1)$. Since $z^{\star}\notin U'$, it follows that $z^{\star}\notin U$ and thus $U \in \mathcal{B}_A$. The result then follows from applying the Lemma~\ref{lm: uppershadow} twice.
\end{proof}

We now wrap up the proof by applying Lemma~\ref{lm: pseudo not in support} and obtaining a contradiction.
\begin{proof}[Proof of Lemma~\ref{lm: pseudo not in support}]
    By Lemma~\ref{lm: p+2},  $\nu(\mathcal{B}'_A) \leq q^2 \nu_A(\mathcal{B}_A) \leq q^{3}O(\epsilon)$. We can now apply Lemma~\ref{lm: localization}, with $\ell = p+2$, degree parameter $q(p-1)-1$, and special point $z^{\star}$. The conditions of Lemma~\ref{lm: localization} are satisfied since $p + 2 \geq \max(\lceil \frac{q(p-1)}{q-q/p}\rceil, 4)$ and $q^3O(\epsilon) \leq q^{M/2}O(\epsilon)$. From Lemma~\ref{lm: localization} it follows that $\mathcal{B}'_A$ is empty and there is a value $\gamma$ such that after changing $\tilde{f}(z^{\star})$ to $\gamma$, we have $\deg(\tilde{f}) \leq p(q-1)$. In particular, it must be the case that $\langle f \circ A \circ B^{\star}, H_e \rangle = 0$, or equivalently, $\langle \tilde{f} \circ (R^{\star}, b^{\star}), P \rangle = 0$, where  $(R^{\star}, b^{\star})$ is the non-identity part of $B^{\star}$ when written according to~\eqref{eqn: Res'}.

Recall how the point $z^{\star}$ was defined: after $A$ is fixed, $x^{\star} \in \Ff_q^{s+t+102}$ is the unique point such that $Ax^{\star} = b$, and $z^{\star} \in \Ff_q^{p+102}$ is the last $p+102$ coordinates of $x^{\star}$. Since $B^{\star}$ satisfies $A \circ B^{\star} \in \C_{a,b}$, we have $(A \circ B^{\star})(a) = b$, and hence $B^{\star} a = x^{\star}$. Letting $a'$ be the last $p+102$ coordinates of $a$, it follows that $(R^{\star}, b^{\star})(a') = z^{\star}$. However, by assumption $a' \notin \supp(P)$, and since $(R^{\star}, b^{\star})$ is full rank, $(R^{\star}, b^{\star})(\alpha) \neq z^{\star}$ for any $\alpha \in \supp(P)$. We now see where the contradiction lies. After changing the value of $\tilde{f}(z^{\star})$, we suddenly have 

\[
\langle \tilde{f} \circ (R^{\star}, b^{\star}), P \rangle = \sum_{\alpha \in \supp(P)} \tilde{f}\left( (R^{\star}, b^{\star})(\alpha) \right) \cdot P(\alpha) = 0.
\]
However, since we only changed the value of $\Tilde{f}(z^{\star})$, no term in the above summation was changed, and this inner product should still be nonzero. Hence, the set $\Sc_t$ cannot be dense in any $\C_{a,b}$ where $a \notin \supp(H)$.
\end{proof}

\subsection{Proof of Lemma~\ref{lm: support error}}

We have established that there exists an $a^{\star} \in \supp(H)$ and $b \in \Ff_q^{n}$ such that
\[
\mu_{a^{\star},b}(\Sc_t) \geq 1 - \frac{1}{(q-1)^2} - \frac{1}{q^6}-2000tq^{200-M/4}.
\]
Using this, we will deduce Lemma~\ref{lm: support error}, which says that the set $\Sc_t$ is dense in $\cup_{v \in \supp{(H)}} \C_{v,b}$.

Let $\U_{s+t}$ denote the set of $(s+t)$-flats $U$ such that $\deg(f|_U) > d$ and let $\nu$ denote the uniform measure on the set of $(s+t)$-flats. We sample $T \in \C_{a^{\star},b}$ by first choosing an $(s+t)$-flat $U$ containing $b$, and then choosing $T$ whose image is $U$ conditioned on $Ta^{\star} = b$. The point of this procedure is that after choosing $U$, the resulting $T$ can only reject if $\deg(f|_U) > d$. Formally,

\begin{enumerate}
    \item Choose a random $(s+t)$-flat $U$ containing $b$, and a random basis, $U =\spa(u_1, \ldots, u_{s+t}) + u_0$. Let $T' = (M', u_0)$ where the $i$th column of $M'$ is $u_i$. By assumption $T'(x) = b$ for some $x \in \Ff_q^{s+t}$.
    \item Choose an arbitrary matrix $B \in \T_{s+t, s+t}$ such that $Ba^{\star} = x$ and output $T = (M' \cdot B, u_0)$.
\end{enumerate}

It is easy to check that this procedure samples uniformly from $\C_{a^{\star},b}$ and that as noted, $T$ can only reject if $\deg(f|_U) > d$, which leads to the following observation:
\begin{remark}
    \[
     1 - \frac{1}{(q-1)^2} -\frac{1}{q^6}- 2000tq^{200-M/4} \leq \mu_{a^{\star},b}(\Sc_t) \leq \nu_{b}(\U_{s+t}),
    \]
    where $\nu_b$ denotes density in the zoom-in $\D_b$  (on the affine Grassmann graph).
\end{remark}

Using this information about $\U_{s+t}$, we now try to apply Lemma~\ref{lm: localization}. Sample a full rank $T \in \cup_{v \in \supp(H)} \C_{v,b}$ as follows:
\begin{enumerate}
    \item Choose an $(s+t+100)$-flat $V$ uniformly at random containing $b$.
    \item Choose $U \subseteq V$ uniformly at random such that $b \in U$.
    \item Choose a basis representation $U = \spa(u_1, \ldots, u_{s+t}) + u_0$ and output $T = (M, u_0)$ where the $i$th column of $M$ is $u_i$, conditioned on $Tv = b$ for some $v \in \supp(H)$.
\end{enumerate}
After $V$ is chosen, recall the following two sets,
\[
\mathcal{A}_V = \{U \subseteq V \; | \;  \dim(U) = s+t, \deg(f|_U) > d, b \in U\},
\]
\[
\mathcal{B}_V = \{U \subseteq V \; | \;  \dim(U) = s+t, \deg(f|_U) > d, b \notin U\},
\]
and let $\nu_V$ denote measure over $s+t$-flats contained in $V$.

Then $\E_V[\nu_V(\mathcal{A}_V)] \geq 1 - \frac{1}{(q-1)^2} - \frac{1}{q^6}-2000tq^{200-M/4}$, so with probability at least $1 - \frac{1}{(q-1)^2} -\frac{1}{q^6}- 2000tq^{200-M/4}$, we have $\nu_V(\mathcal{A}_V) > 0$. On the other hand, $\E_V[\nu_V(\mathcal{B}_V)] = O(\epsilon)$, so with probability at least $1 - q^{-M/2}$, we have $\nu_V(\mathcal{B}_V) \leq q^{-M/2}O(\epsilon)$. Altogether, this implies that with probability at least $$1 -  \frac{1}{(q-1)^2} -\frac{1}{q^6}- 2000tq^{200-M/4} - q^{-M/2},$$ $V$ satisfies both $\nu_V(\mathcal{A}_V) > 0$ and  $\nu_V(\mathcal{B}_V) \leq q^{M/2}O(\epsilon)$.

Suppose such a $V$ is chosen and the above holds. We may now apply Lemma~\ref{lm: localization} to show that any $T$ that can be chosen in step $3$ must now reject. This will establish the desired result, that a random $T \in \cup_{v \in \supp(H)}\C_{v,b}$ rejects with probability close to $1$.

By Lemma~\ref{lm: localization} there exists a value $\gamma$ such that after changing $f(b)$ to $\gamma$, $\deg(f|_V) \leq d$. It must be the case that $\gamma \neq f(b)$ because $\nu_V(\mathcal{A}_V) > 0$. Let $f'$ be the the function after changing the value of $f$ at $b$. Then for any $T$ as described above, we must have $\langle f' \circ T, H_e \rangle = 0$ for every valid $e$. Since $T$ is full rank, 
there can only be one point mapped to $b$ and that point is some $v \in \supp(H)$, so we have
\[
\langle f' \circ T, H_e \rangle - \langle f \circ T, H_e \rangle = H_e(v)(f(b) - f'(b)) \neq 0.
\]
Thus, for every valid $e$ $\langle f \circ T, H_e \rangle \neq \langle f' \circ T, H_e \rangle$, and in particular, $T$ rejects $f$.

\begin{proof}[Proof of Lemma~\ref{lm: support error}]
    Sampling a full rank $T \in \cup_{v \in \supp(H)} \C_{v,b}$ via the procedure above, the previous argument shows that with probability at least $1 -  \frac{1}{(q-1)^2} -\frac{1}{q^6}- 2000tq^{200-M/4} - q^{-M/2}$ over $(s+t+100)$-flats $V$, the outputted $T$ in step 3 rejects. It follows that $\Sc_t$ has density at least $1 -  \frac{1}{(q-1)^2} - \frac{1}{q^6}-2000tq^{200-M/4} - q^{-M/2}$ in the set of full rank transformations $T$ such that $Tv = b$ for some $a \in \supp(H)$. Since non-full rank transformations constitute only a $\frac{1}{q^{100}(q-1)}$-fraction of transformations in $\cup_{v \in \supp(H)} \C_{v,b}$, by Remark~\ref{rmk: non full rank}, we subtract out another $\frac{2}{q^{100}(q-1)}$ to obtain the desired result.
\end{proof}
\subsection{Iterating the Argument} \label{sec: iterate}
Recall that $\rej_{s+t}(f)$ denotes the probability that a randomly chosen $T \in \T_{n, s+t}$ rejects $f$. By Lemma~\ref{lm: localization}, for at least $9/10$ of $T \in \cup_{v \in \supp(H)} \C_{v,b}$, there is a value $\gamma$ such that after changing the value of $f(b)$ to $\gamma$, $T$ accepts $f$. Thus,  there exists a $f'$ such that $f'$ is identical to $f$ at all points except $b$ and
    \[
    \Pr_{T}[T \text{ rejects } f \; | \; \exists v \in \supp(H), T(v) = b] \leq 1 - \frac{9}{10q}.
    \]

\begin{proposition}  \label{prop: improvement}
    If $0 < \rej_{s+t}(f) < q^{-M}$ and $t = O(p)$ there exists a point $z \in \Ff_q^{n}$ and a function $f'$ that is identical to $f$ at all points except $b$ such that
    \[
    \rej_{s+t}(f') \leq \rej_{s+t}(f) - \frac{|\supp(H)|}{q^n C(q)},
    \]
    where $C(q) = O(q)$.
\end{proposition}
\begin{proof} 
    By Lemma~\ref{lm: support error}, there exists $b$ such that 
    \[
    \Pr_{T}[T \text{ rejects } f \; | \; \exists v \in \supp(H), T(v) = b] \geq  1 -  \frac{1}{(q-1)^2}-\frac{1}{q^6} - \frac{2}{q^{100}(q-1)} - 2000tq^{200-M/4} - q^{-M/2}.
    \]
    By the above discussion, there exists a $f'$ such that $f'$ is identical to $f$ at all points except $b$ and
    \[
    \Pr_{T}[T \text{ rejects } f \; | \; \exists v \in \supp(H), T(v) = b] \leq 1 - \frac{9}{10q}.
    \]
    
    On the other hand, it is clear that for $T$ such that $T(v) \neq b$ for all $v \in \supp(H)$, the results of the tests on $f \circ T$ and $f' \circ T$ are the same because $f$ and $f'$ are identical on the points needed to evaluate $\langle f \circ T, H_e \rangle$ and $\langle f' \circ T, H_e \rangle$ for any $e$. Since $\frac{|\supp(H)|}{q^n}$ is the probability that $b \in \{T(v) \; | \; v \in \supp(H) \}$, a direct calculation yields
    \begin{align*}
    &\rej_{s+t}(f')\\ &\leq \rej_{s+t}(f) - \frac{|\supp(H)|}{q^n}\left( 1 -  \frac{1}{(q-1)^2}-\frac{1}{q^6} - \frac{2}{q^{100}(q-1)}- 2000tq^{200-M/4} - q^{-M/2} - \left(1 - \frac{9}{10q} \right)\right) \\
    & \leq \rej_{s+t}(f)- \Omega\left(\frac{|\supp(H)|}{q^n q}\right).
    \end{align*}
\end{proof}

Finally, we can conclude the proof of Theorem~\ref{thm:main}.
\begin{proof}[Proof of  Theorem~\ref{thm:main}]
   By Proposition~\ref{prop: improvement}, while $\rej_{s+t}(f) > 0$, we can change the value of $f$ at one point and reduce the rejection probability by $\Omega\left(\frac{|\supp(H)|}{q^n q}\right)$. When the rejection probability is $0$, the function must be degree at most $d$, therefore, 
\[
\delta_d(f) q^n \leq O\left(\frac{\rej_{s+t}(f)}{|\supp(H)|}q^nq\right)
\]
which implies that 
\[
\rej_{s+t}(f) \geq \Omega\left(\frac{|\supp(H)|}{q} \delta_d(f)\right).
\] 
\end{proof}

\section{Optimal Testing from other Local Characterizations}\label{sec:optimal_testing_from_other}
When showing that the sparse flat test is optimal, we relied minimally on the structure of $H_e$. Thus it is not hard to extend our methods and show optimal testing results for other polynomials that give local characterizations. We will reuse the variables $s$ and $r$ in this section, so they no longer refer to their previous definitions. Let $P: \Ff_q^{k} \xrightarrow[]{} \Ff_q$ be a polynomial of the form

 \[
    P(x_1, \ldots, x_k) = \prod_{i = 1}^{s} P_i(x_{m(i)}, \ldots, x_{m(i+1)-1})).
\]
Where $m(1) = 1$, $m(s) = k+1$, $m(i+1) - m(i) \leq t'$ for each $1 \leq i \leq s-1$ and some small constant $t'$. In words, $H$ is a $k$-variate polynomial that is the product of polynomials in few variables, where the variables of each of these polynomials is disjoint. Finally let $\mathcal{M} \subseteq \{\Ff_q^{t} \xrightarrow[]{} \Ff_q \}$ be an arbitrary nonempty set affine invariant set of polynomials and suppose $t = \poly(q)$. Define
\[
\mathcal{E}=  \{e \in  \{0,\ldots, q-1\}^t \; | \;  \prod_{i=1}^{t} x_i^{q-1-e_i} \notin \mathcal{M} \}.
\]
Notice that  $(q-1, \ldots, q-1) \notin \mathcal{E}$. It is well known that any affine invariant set of polynomials is given by the span of the monomials that appear in at least one polynomial of the family, c.f.\ \cite{KS}, and this fact is elaborated on in Appendix~\ref{sec:prob_1_accept}. In combination with Lemma~\ref{lm: inner product}, it follows that $g: \Ff_q^{t} \xrightarrow[]{} \Ff_q$ is in $\mathcal{M}$ if and only if $\langle g, \prod_{i=1}^{t} x_i^{e_i} \rangle = 0$ for every $(e_1, \ldots, e_t) \in \mathcal{E}$. In comparison with the description in Section~\ref{sec: lifted codes}, $\{  \prod_{i=1}^{t} x_i^{e_i} \; | \; e \in \mathcal{E} \}$ is an explicit basis for $\mathcal{M}^{\perp}$.

Using $H$ and monomials with exponent vectors in $\mathcal{E}$, we can define a test similar to the sparse flat tester. For $e \in \mathcal{E}$, define
\[
H_e(x_1, \ldots, x_{k+t}) = P(x_1, \ldots, x_k) \prod_{i=1}^{t}x_{k+i}^{e_i},
\]
Let $\mathcal{F}_n(H) = \{f: \Ff_q^n \xrightarrow[]{} \Ff_q \; | \; \forall T \in \T_{n, k}, \langle f \circ T, H_e \rangle = 0, \forall e \in \mathcal{E} \}$. It is clear that $\mathcal{F}_n(H)$ is affine invariant with the following natural tester:
\begin{enumerate}
    \item Choose $T\in \T_{n, k}$ uniformly at random.
    \item If $ \langle f \circ T, H_e \rangle \neq 0$ for any $e \in \mathcal{E}$, reject. Otherwise, accept.
\end{enumerate}
Let $\supp(H) = \cup_{e \in \mathcal{E}} \supp(H_e) = \supp(P) \times \Ff_q^{t}$. It is clear that the number of queries made by the tester is $Q = |\supp(H)|$, although this value will not be of interest for us. Instead, the goal for the remainder of this section is to show that this tester is optimal. As a consequence, this allows one to obtain lower query complexity optimal testers for general affine invariant codes by constructing local characterizations using $P$ with sparse support of the prescribed form above, similar to what is done for Generalized Reed-Muller Codes in \cite{RZS}.
\begin{theorem} \label{th: main general}
The above tester is optimal for $\mathcal{F}_n(H)$. That is, for any $f: \Ff_q^{n} \xrightarrow[]{} \Ff_q$, 
\[
\rej(f) \geq c(q)\min(1, Q\delta(f, \mathcal{F}_n(H))),
\]
where $\rej(f)$ is the probability that the tester rejects $f$, $c(q) = \frac{1}{\poly(q)}$, $Q = |\supp(H)|$ is the number of queries that the tester performs, and $\delta(f, \mathcal{F}_n(H))$ is the minimal relative hamming distance between $f$ and a member of $\mathcal{F}_n(H)$.
\end{theorem}

We will follow the same strategy, first locating a potential error, and then correcting the error and iterating. Let $\Sc$ denote the set of rejecting tests in $\T_{n, k+t}$ and assume that $\mu(\Sc) \leq q^{-M}$ for some $M$ such that $M > 10t'$ and $q^{M} \geq t$.

Before going into the proof, we introduce a lemma shown in \cite{KM}. Recall the definition of lifted affine-invariant codes from the introduction. By design, $\mathcal{F}_n(\mathcal{H})$ is a lifted affine invariant code. Letting $\mathcal{F} = \mathcal{F}_{k+t}(H)$, it is easy to see that 
\[
\mathcal{F}_{n}(H) = \Lift_n(\mathcal{F}).
\]
Since we can view $\mathcal{F}_n(H)$ as a lifted code, we can then apply the following lemma from \cite{KM}.
\begin{lemma} \label{lm: general affine invariant shadow lemma}
    Let $g: \Ff_q^{k'+1} \xrightarrow[]{} \Ff_q$ be a polynomial such that $g \notin \Lift_{k'+1}(\mathcal{F})$, and suppose that $k' \geq k$. Then,
    \[
    \Pr_{U}[g|_U \notin  \Lift_{k'}(\mathcal{F})] \geq \frac{1}{q},
    \]
    where the probability is over hyperplanes $U \subset \Ff_q^{k'+1}$. If this inequality is tight and the set of hyperplanes $U$ such that $g|_U  \notin  \Lift_{k'}(H)$ is of the form
    \[
    \{U \; | \; x^{\star} \in U\},
    \]
    for some point $x^{\star}$, then there is a function $g'$ equal to $g$ at all points except $x^{\star}$ such that $g' \in  \mathcal{F}_{k+i+1}(H)$. 
\end{lemma}
This lemma will play the role of Lemma~\ref{lm: uppershadow 2} in this section's analysis.

\paragraph{Locating a Potential Error}
In order to locate an error, we will once again show that $\Sc$ is dense on some zoom-in. As we already assume $\Sc$ has small measure, this step requires us to show that $\Sc$ is poorly expanding and pseudorandom with respect to zoom-outs, zoom-outs on the linear part, and zoom-ins on the linear part.

\begin{lemma} \label{lm: general pseudorandomness}
    The set $\Sc$ has the following properties:
    \begin{itemize}
        \item $\mu(\Sc) \leq q^{-M}$
        \item $\Phi(\Sc) \leq 1 - 1/q$
        \item $\Sc$ is $q \mu(\Sc)$-pseudorandom with respect to zoom-outs and zoom-outs on the linear part. 
        \item $\Sc$ is $tq^{162}\epsilon$-pseudorandom with respect to zoom-ins on the linear part.
    \end{itemize}
    where $\mu$ is measure in the set $\T_{n, k+t}$.
\end{lemma}

The first item is true by assumption. The second and third items can be shown by the exact same arguments as in Lemmas~\ref{lm: expansion general} and \ref{lm: zoom-out} respectively.

The proof for the fourth item also proceeds similar to the Reed-Muller case, but since this argument was more involved, we will review the steps.

Consider a zoom-in on the linear part $\C_{a,b,\lin}$. We can assume $a \neq 0$ as otherwise the set is either empty or $\T_{n, k}$. We will again have two cases, one where at least one of the last $t$ coordinates of $a$ is nonzero, and another when all are zero. We can prove the former case in exactly the same way as Lemma~\ref{lm: lin part pseudo case}. For the latter case, recall that by assumption $(q-1,\ldots, q-1) \notin \mathcal{E}$. Thus, the reduction from the former case to the latter case (with a factor of $tq$ loss) works exactly the same as in the Proof of Lemma~\ref{lm: lin part pseudo overall}.

Once again we can sample a transformation by first choosing a full rank $A \in \T_{n, k + t + 100}$, $B \in \Res_{k+ t + 100, k}$ and outputting $A \circ B = T$. Where now $ \Res_{k, t+100, t}$ is the set of affine transformations $(R,b)$ with 
\begin{equation} \label{eqn: Res 2}
 R =  \begin{bmatrix}
   I_{k} & 0 \\
   0 & R'
   \end{bmatrix}, b = \begin{bmatrix} 0 \\ b' \\ \end{bmatrix},
\end{equation}
with $R' \in \Ff_q^{(t+100) \times t}$ is full rank and $b' \in \Ff_q^{t+100}$. Call $(R',b')$ the non-trivial part of $B = (R,b)$ and let $\fl(B) = \im(R', b')$.

Following the same setup as in Section~\ref{sec: pseudorandomness zoom-in linear}, we can find $A$ such that the following two hold:
\begin{itemize}
    \item $b \in \im_{a_{[k]}}(A)$
    \item There exists $B \in \Res_{k, t+100, t}$ such that $A \circ B \in \Sc$
    \item $\Pr_{B}[A \circ B \in \Sc \; | \; b \notin \im_{a_{[k]}}(M \circ R)] \leq \frac{2}{\alpha}\eps$.
\end{itemize}
Fixing this $A$, define 
\[
\tilde{f}(\beta_1, \ldots, \beta_{t+100}) = \sum_{\alpha} f \circ A(\alpha, \beta_1, \ldots, \beta_{t+100}) P(\alpha).
\]
Then for any restriction $B$ with $B' = (R', b') \in \T_{t+100,t}$ as its non-trivial part, and any $t$-variate monomial $x^e$, we have
\[
\langle f \circ A \circ B, H_e \rangle = \langle \tilde{f} \circ B', x^e \rangle.
\]
It follows that $B = (R, b)$ rejects $f$ if and only if $\tilde{f}|_{b' + \im(R')} \in \mathcal{M}$.

By construction $Ax^{\star} = b$ for some $x^{\star} \in \Ff_q^{k+100}$ such that $x^{\star}$ is equal to $a$ in its first $k$ coordinates. Then $b \in \im_{a_{[k]}}(A \circ B)$ is equivalent to $z^{\star} \in \im(R')$ where $z^{\star}$ is the last $t+100$ coordinates of $x^{\star}$

Letting $\mathcal{R}_A = \{(R', b') \in \T_{t+100,t} \; | \; \tilde{f}|_{b' + \im(R')}, z^{\star} \in \im(R') \}$, we can translate the third item above to
\[
\Pr_{B}[A \circ B \in \Sc \; | \; b \notin \im_{a_{[k]}}(M \circ R)] = \mu_A(\mathcal{R}_A) \leq \frac{2}{\alpha}\eps,
\]
where $\mu_A$ is measure in $\Res_{k, t+100, t}$. Define
$$\mathcal{B}_A = \{b' + \im(R') \; | \; (R', b') \in \mathcal{R}_A \}.$$ By the same argument as in Lemma~\ref{lem:cal_B_nonempty}, $\mathcal{B}_A$ is nonempty. Finally, we can use the same proof as that of Lemma~\ref{lm: lin part pseudo case} (except referring to the first part of Lemma~\ref{lm: general affine invariant shadow lemma} instead of Lemma~\ref{lm: uppershadow} to bound the sizes of upper shadows) to show that $\alpha \leq 2e^{-4/q}\epsilon q^{160} \leq \epsilon q^{161}\epsilon$. This shows that $\Sc$ is $\epsilon q^{161}\epsilon$-pseudorandom with respect to $\C_{a,b,\lin}$ where one of $a$'s last $t$-coordinates is nonzero. If this is not the case, then since $(q-1,\ldots, q-1) \notin \mathcal{E}$, we can perform the same reduction as in Lemma~\ref{lm: lin part pseudo overall}, to show that $\Sc$ is $tq^{162}\epsilon$-pseudorandom with resepect to every zoom-in on the linear part.

Altogether, this establishes Lemma~\ref{lm: general pseudorandomness}, which allows us to apply Theorem~\ref{th: pseudorandom} with $\xi = tq^{162-M}$ and $\ell = s+t$. Since we take $t \geq 10$, Theorem~\ref{th: pseudorandom} implies that we have that $\Sc$ must have density at least $1 - \frac{1}{(q-1)^2} - \frac{1}{q^6}- 2000tq^{200 - M/4}$ on some zoom-in $\C_{a,b}$. Henceforth, fix this zoom-in and call it $\C_{a^{\star}, b}$.

\paragraph{The point $a^{\star}$ is in the support of $H$}
We next show that it must be the case that $a^{\star} \in \supp(H)$. We will need the following lemma which is the same as Lemma~\ref{lm: localization} but for the current setting. Just as in the setup of Lemma~\ref{lm: localization}, let $g: \Ff_q^{\ell + 100}$ be an arbitrary polynomial, let $\nu$ denote uniform measure over $\ell$-flats in $\Ff_q^{\ell + 100}$, and let $b \in \Ff_q^{\ell + 100}$ be an arbitrary point. Also let $\mathcal{G} \subseteq \{\Ff_q^{N} \xrightarrow[]{} \Ff_q \}$ be an arbitrary affine-invariant code.

Define the following two sets:
\[
\mathcal{A} = \{U  \; | \;  \dim(U) = \ell, g|_U \notin \Lift_{\ell}(\mathcal{G}), b \in U\},
\]
\[
\mathcal{B} = \{U \; | \;  \dim(U) = \ell, g|_U \notin \Lift_{\ell}(\mathcal{G}), b \notin U\}.
\]
Keep $\epsilon$ and $M$ (the large absolute constant) the same as we have defined, so that $q^{M/2}O(\epsilon) < q^{-M/2}$ is small. We will use the following result, which is an extension of the results in Section 3.2 of \cite{KM}:
\begin{lemma} \label{lm: localization lifted}
    Keep the notation defined above and suppose $\ell \geq \max(N, 4)$. If $\nu(\mathcal{B}) \leq q^{M/2}O(\epsilon)$ then $\mathcal{B} = \emptyset$. Moreover there is a value $\gamma$ such that after changing $g(b)$ to $\gamma$, $\nu(\mathcal{A}) = 0$
\end{lemma}
\begin{proof}
The proof of this lemma is the same as the proof of Lemma~\ref{lm: localization}, however we appeal to Lemma~\ref{lm: general affine invariant shadow lemma} instead of Lemma~\ref{lm: uppershadow 2}. 
\end{proof}

Following the same steps as in the proof of Lemma~\ref{lm: pseudo not in support}, this lemma almost immediately shows that it must be the case $a^{\star} \in \supp(H)$. Suppose for the sake of contradiction that $a^{\star} \notin \supp(H)$, and that it is the group of coordinates $(a^{\star}_{m(i)}, \ldots, a^{\star}_{m(i+1)-1})\notin \supp(P_i)$. By assumption $P_i$ is a $t''$-variate polynomial for $t'' \leq t'$. Let $$\mathcal{F}_{t''}(P_i) = \{g: \Ff_q^{t''} \xrightarrow[]{} \Ff_q \; | \; \langle g \circ T, P_i \rangle = 0, \forall T \in \T_{t'',t''} \}.$$

We can similarly find an auxiliary polynomial $\tilde{f}: \Ff_q^{t''+102} \xrightarrow[]{} \Ff_q$ and a special point $z^{\star} \in \Ff_q^{t''}$ such that the following hold:
\begin{itemize}
    \item $\mu(\mathcal{R}) \leq O(\epsilon)$, where $\mathcal{R} = \{B \in \T_{t''+102, t''} \; | \; \langle \tilde{f} \circ B, P_i \rangle \neq 0, z^{\star} \notin \im(R) \}$ and $\mu$ is measure in $\T_{t''+102, t''}$.
    \item There exists a full rank affine transformation $B^{\star} \in \T_{t'' + 102, t''}$ such that $\langle \tilde{f} \circ B^{\star}, P_i \rangle \neq 0$,
    \item $z^{\star} \notin \{B^{\star}(v) \; | \; v \in \supp(P_i) \}$.
\end{itemize}
Using the first fact, we can also bound the measure of the following set of $t''$-flats,
\[
\mathcal{B} = \{U \in \AffGras(t''+102, t'') \; | \; f|_U \in \mathcal{F}_{t''}(P_i), z^{\star} \in U\}.
\]

\begin{lemma} \label{lm: sparse to full general}
    Suppose $g$ is a $t''$-variate polynomial such that $g \notin \mathcal{F}_{t''}(P_i)$. Then with probability at least $\frac{1}{q^{t''}}$ over $T \in \T_{t'', t''}$ we have $\langle g \circ T, P_i \rangle \neq 0$.
\end{lemma}
\begin{proof}
 Choosing $T$ randomly, we can view $\langle g \circ T, P_i \rangle$ as a polynomial in the entries of $T$. Since $g \notin \mathcal{F}_{t''}(P_i)$ there must be some choice $T$ that makes this polynomials value nonzero, so $\langle g \circ T, P_i \rangle$ is a nonzero polynomial in the entries of $T$. As the individual degrees of $g$ must be at most $q-1$, it follows from Schwarz-Zippel that with probability at least $\frac{1}{q^{t''}}$, $\langle g \circ T, P_i \rangle \neq 0$.
\end{proof}
Let $\nu$ denote measure in $\AffGras(t''+102, t'')$. We can choose $B \in \T_{t''+102, t''}$ such that $z^{\star} \notin \im(B)$ by first choosing $U \in \AffGras(t''+102, t'')$ not containing $z^{\star}$ and then $B$ with image contained in $U$. Lemma~\ref{lm: sparse to full general} along with the assumption $\mu(\mathcal{R}) \leq O(\epsilon)$ implies that
\[
\frac{1}{q^{t''}} \nu(\mathcal{B}) \leq \mu(\mathcal{R}).
\]

Therefore, we have that $\nu(\mathcal{B}) \leq q^{t''}O(\epsilon)$ and by  Lemma~\ref{lm: localization lifted}, $\nu(\mathcal{B}^{\uparrow^4}) \leq q^{t''+4}O(\epsilon) \leq q^{M/2}O(\epsilon)$. We can then apply Lemma~\ref{lm: localization lifted} with $\mathcal{G} = \mathcal{F}_{t''}(P_i)$ and special point $z^{\star}$ to get that there is a value $\gamma$ such that after changing $\tilde{f}(z^{\star})$, we have,
\[
\tilde{f}|_{U'} \in \Lift_{t'' + 4}(\mathcal{F}_{t''}(P_i)), \text{\quad for all $t'' + 4$-flats $U'$ containing $z^{\star}$}.
\]
This also implies that $\tilde{f}|_U \in \mathcal{F}_{t''}(P_i)$ for all $t''$-flats $U \ni z^{\star}$, but note that changing the value of $\tilde{f}(z^{\star})$ does not affect the value of the test
\[
    \langle \tilde{f} \circ B^{\star}, P_i \rangle \neq 0,
\]
because of the third item above. Thus, we have a contradiction and it follows that  $a^{\star} \in \supp(H)$.

\paragraph{Dense on all zoom-ins inside the support of $H$}
After establishing that $\Sc$ has density at least $1 - \frac{1}{(q-1)^2} - 2000tq^{200-M/4}$ inside $\C_{a^\star, b}$ and $a^{\star} \in \supp(H)$, we can use the same arguments as in the proof of Lemma~\ref{lm: support error} to show that $\Sc$ must have density at least
\[
1 - \frac{1}{(q-1)^2} - \frac{2}{q^{100}(q-1)} - 2000q^{200+t-M/4} - q^{-M/2},
\]
inside $\cup_{v \in \supp(H)} \C_{v,b}$. We can choose $T \in \C_{a^{\star}, b}$ by first choosing a $k + t+100$-flat $V$ containing $b$, then a $k+t$-flat $U \subseteq V$ containing $b$, and finally outputting $T$ conditioned on $\im(T) = U$ and $Ta = b$. 

Using the same sampling and averaging argument, we get that with probability at least $1 - \frac{1}{(q-1)^2} - 2000q^{200+t-M/4} - q^{-M/2}$ over $k + t + 100$-flats $V$, the following hold:
\begin{itemize}
    \item $\nu_V(\mathcal{A}_V) > 0$, where $\mathcal{A}_V = \{U \subseteq V \; | \; \dim(U) = k, f|_{U} \notin \mathcal{F}_{k+t}(H), b \in U \}$,
    \item  $\nu_V(\mathcal{B}_V) \leq q^{M/2}O(\epsilon)$, where $\mathcal{B}_V = \{U \subseteq V \; | \; \dim(U) = k, f|_{U} \notin \mathcal{F}_{k+t}(H), b \notin U \}$,
\end{itemize}
Applying Lemma~\ref{lm: localization lifted} with $\mathcal{G} = \mathcal{F}_{k+t}(H)$ and $g = f|_V$, we get that there is a value $\gamma$ such that after changing the value of $f(b)$ to $\gamma$, $\nu_V(\mathcal{A}_V) = 0$. By the first item above, $\gamma \neq f(b)$. Therefore, prior to changing the value of $f(b)$, we must have had $\langle f|_V \circ B , H \rangle \neq 0$ for all full rank $B \in \T_{k+t+100, k+t}$ such that $B(v) = b$ for some $v \in \supp(H)$. After subtracting out the fraction of non-full rank transformations, we can conclude that $\Sc$ has density at least 
\[
1 - \frac{1}{(q-1)^2} - \frac{2}{q^{100}(q-1)}- 2000q^{200+t-M/4} - q^{-M/2},
\]
inside $\cup_{v \in \supp(H)} \C_{v,b}$.

\paragraph{Correcting the Error and Iterating}
Finally, we can correct the error and iterate as done in Section~\ref{sec: iterate}. For at least $9/10$ of $T \in \supp_{v \in \supp(H)} \C_{v,b}$, there is a value $\gamma$ such that after changing the value of $f(b)$ to $\gamma$, $T$ accepts $f$. Thus, there exists an $f'$ that is identical to $f$ at all points except $b$ and
\[
\Pr_{T}[T \text{ rejects } f \; | \; \exists v \in \supp(H), T(v) = b] \leq 1 - \frac{9}{10}q.
\]
By the same calculation as in Lemma~\ref{prop: improvement}, it follows that
\[
\rej(f') \leq \rej(f) - \frac{|\supp(H)|}{q^nC(q)},
\]
for some $C(q) = O(q)$, and we can conclude 
\[
\rej(f) \geq \Omega\left(\frac{|\supp(H)|}{q} \delta(f, \mathcal{F}_n(H)) \right).
\]

\bibliographystyle{plain}
\bibliography{references}

\appendix

\section{Proof of Theorem~\ref{th: pseudorandom}}\label{sec:pf_expansion}
We now prove Theorem~\ref{th: pseudorandom}. First we will need to show the following result, which is a slight variation of \cite[Theorem 2.4]{KM}.

\begin{thm} \label{th: pseudo aff gras}
    Let $\Ac  \subseteq \AffGras(n+\ell, \ell)$ satisfy 
    \begin{enumerate}
        \item $\mu(\Ac) \leq \xi$,
        \item $\Ac$ is $\xi$-pseudorandom with respect to hyperplanes, hyperplanes on its linear part, and points on its linear part.
        \item $1 - \Phi(\Ac) \geq \frac{1}{q} - \delta$.
    \end{enumerate}
    Then there exists a point $x \in \Ff_q^n$ such that $\mu(\Ac_x) \geq 1- q^2\left(4q^{-\ell} + 867\xi^{1/4} + \frac{\delta}{q-1} \right)$. 
    
    Measure, expansion, and pseudorandomness are all with respect to $ \AffGras(n+\ell, \ell)$, and $\mu(\Ac_x)$ denotes the fractional size of $\Ac$ in the subset of $\AffGras(n+\ell, \ell)$ that contains the point $x$.
\end{thm}

To show that an analogous result holds in the affine Bi-linear Scheme Graph, we retrace the steps in \cite{BKS} to establish a connection to the affine Grassmann Graph and then apply Theorem~\ref{th: pseudo aff gras}. To this end, define the following injective map $\varphi: \T_{n, \ell} \xrightarrow{} \V_{n + \ell, \ell}$. Let $\{e_1, \ldots, e_{\ell}\} \in \Ff_q^{\ell}$ be the canonical basis. For two column vectors $v$ and $w$ with lengths $\ell_1$ and $\ell_2$ respectively, we denote by $(v,w)^T$ the length $\ell_1 + \ell_2$ column vector obtained by vertically concatenating $v$ and $w$. Likewise, let $[v^T, w^T]$ be the length $\ell_1 + \ell_2$ column vector obtained by horizontally concatenating $v$ and $w$.

For an affine transformation $(M, c) \in \T_{n, \ell}$, where $M$ has columns $v_1, \ldots, v_{\ell}$ we set $\varphi((M,v)) = (0, c)^T + \spa((e_1, v_1)^T, \ldots, (e_{\ell}, v_{\ell}^T))$. It is clear that this map is injective. Moreover, the image of $\varphi$ is the set of $\ell$-flats $V \subseteq \Ff_q^{n+\ell}$ such that the projection of $V$ onto the first $\ell$ coordinates is full rank. 

The map $\varphi$ preserves edges, nearly preserves expansion, and maps the canonical non expanding sets in $\AffShort(n, \ell)$ to their counterparts in $\AffGras(n+\ell, \ell)$.

\begin{lemma} \label{lm: edges preserved}
If $T_1, T_2 \in \T_{n, \ell}$ are adjacent in $\AffShort(n, \ell)$, then $\varphi(T_1), \varphi(T_2)$ are adjacent in $\AffGras(n+\ell, \ell)$. Conversely if $V_1, V_2 \in V_{n, \ell}$ are adjacent in $\AffGras(n+\ell, \ell)$ and both in the image of $\varphi$, then $\varphi^{-1}(V_1), \varphi^{-1}(V_2)$ are adjacent in $\AffShort(n, \ell)$.
\end{lemma}
\begin{proof}
For the first direction, suppose $T_1$ and $T_2$ are given by $(M_1, c_1)$ and $(M_2, c_2)$ respectively. Let $M''_1$ be the $(n + \ell) \times n$ matrix obtained by vertically concatenating $I_{\ell}$ and $M$, and let $c''_1$ be the length $n + \ell$ vector obtained by vertically concatenating the length $\ell$ zero vector and $c_1$. Define $M''_2$ and $c''_2$ similarly. Then notice that $\varphi(T_1)$ is given by the column-image of affine transformation $T''_1 = (M''_1, c''_1)$, while $\varphi(T_2)$ is given by the column-image of affine transformation $T''_2 = (M''_2, c''_2)$. Thus, to show that edges are preserved it suffices to show that the column-images of $T''_1$ and $T''_2$ have intersection of dimension $\ell-1$. Indeed,  $T''_1 - T''_2 = (M''_1 - M''_2, c''_1 - c''_2)$, and since both $M''_1-M''_2$ and $c''_1 - c''_2$ are all zeros in the first $\ell$ rows, it follows that $\dim(\ker(T''_1 - T''_2)) =  \dim(\ker(T_1 - T_2)) = \ell - 1$. Therefore, $\dim(\im(T''_1 - T''_2)) = 1$ and $\dim(\varphi(T_1) \cap \varphi(T_2)) = \ell-1$.

In the converse direction, suppose that $U, V \in \V_{n+\ell, \ell}$ where $\dim(U \cap V) = \ell-1$. Write $U = u_0 + \spa(u_1, \ldots, u_{\ell})$ and  $V= v+ \spa(v_1, \ldots, v_{\ell})$ where $u_i = (e_i, u'_i)^T$, $v_i = (e_i, v'_i)^T$, $u_0 = (0, u'_0)^T$, and $v_0 = (0, v'_0)^T$. Let $M_1$ be the matrix with columns $u_i$ and $M_2$ be the matrix with columns $v_i$. Then by assumption the affine transformations $T_1 = (M_1, u_0)$ and $T_2 = (M_2, v_0)$ have images with intersection of dimension $\ell-1$. We claim that this implies $T_1$ and $T_2$ are in fact equal on an $\ell-1$ dimensional subspace. Indeed if $T_1(a) = T_2(b)$ for $a,b \in \Ff_q^{\ell}$ then examining the first $\ell$ coordinates implies that $a = b$. Therefore, $T_1, T_2$ are equal on the $\ell-1$-dimensional preimage of $\im(T_1) \cap \im(T_2)$. Finally, since $\varphi^{-1}(U)$ and $\varphi^{-1}(V)$ are given by the projections of $T_1$ and $T_2$ respectively onto the last $n$ coordinates, $\dim(\ker(\varphi^{-1}(U) -\varphi^{-1}(V))) =\ell-1$ and the conclusion follows.
\end{proof}

\begin{lemma} \label{lm: neighbors}
Let $V_1 \in \V_{n+\ell, \ell}$ be in the image of $\varphi$. Then at least $1-1/q$ fraction of its neighbors are also in the image of $\varphi$.
\end{lemma}
\begin{proof}
Recall that $V$ is in the image of $\varphi$ if the projection of $V$ onto its first $\ell$-coordinates is full rank. Fix $V_1 \in \V_{n+\ell, \ell}$ in the image of $\varphi$. Then a neighbor of $V_1$ can be sampled as follows: choose a uniformly random point $v_0 \in V$ and $\ell$ linearly independent directions, $v_1, \ldots, v_{\ell}$, so that $V_1 = v_0 + \spa(v_1, \ldots, v_{\ell})$, a uniformly random $v'_{\ell}$ outside of $\spa(v_1, \ldots, v_{\ell})$, and set $V_2 = v_0 + \spa(v_1, \ldots, v_{\ell-1}, v'_{\ell})$. Since $V_1$ is in the image of $\varphi$, the projection of $v_1, \ldots, v_{\ell}$ to the first $\ell$ coordinates is full rank. Likewise, $V_2$ is in the image of $\varphi$ if $v'_{\ell}$ projected onto its first $\ell$ coordinates is linearly independent to $v_1, \ldots, v_{\ell-1}$'s projections to the first $\ell$ coordinates. This happens with probability exactly $1 - 1/q$, completing the proof.
\end{proof}

\begin{lemma} \label{lm: expansion preserved}
If $S \subseteq \T_{n, \ell}$ satisfies $1-\Phi(S) = \eta$, then $\varphi(S)$ satisfies $1 - \Phi'(\varphi(S)) \geq \eta (1-1/q)$, where $\Phi$ and $\Phi'$ are expansion in $\AffShort(n, \ell)$ and $\AffGras(n +\ell, \ell)$ respectively.
\end{lemma}
\begin{proof}
This is a direct consequence of the previous two lemmas. By Lemma~\ref{lm: neighbors}, $1-1/q$ fraction of the neighbors of $\varphi(S)$ are also in the image of $\varphi$ and by the assumption and Lemma~\ref{lm: edges preserved} $\eta$ fraction of these neighbors are also in $\varphi(S)$. The result follows.
\end{proof}

\begin{lemma} \label{lm: nice sets preserved}
The map $\varphi$ is a bijection between the following pairs of sets:
\begin{itemize}
    \item $\C_{a,b}$ and $\D_{v} \cap \im(\varphi)$ for $v = [a, b]^T$.
    \item $\C_{a,b, \lin}$ and  $\D_{v, \lin} \cap \im(\varphi)$ for $v = [a, b]^T$.
    \item $\C_{a^T, b, \beta}$  and $\D_{W} \cap \im(\varphi)$ for $W$ given by $\langle [-b^T, a^T], x \rangle = \beta$.
    \item $\C_{a^T, b, \beta, \lin}$  and $\D_{W, \lin} \cap \im(\varphi)$ for $W$ given by $\langle [-b^T, a^T], x \rangle = \beta$.
\end{itemize}
\end{lemma}

\begin{proof}
    We show the zoom-in and zoom-out cases. The other two are analogous. 

zoom-ins: Take any arbitrary $T = (M, c) \in \C_{a,b}$ and let let $M' = [I_{\ell}, M]^T$ and $c' = [0, c]^T$. It is easy to check that $M'a + c' = [a, b]^T$. Since $T$ is arbitrary, this shows the first half of the bijection.

To complete the proof of this case, we must show that any flat in $\D_{v} \cap \im(\varphi)$ is mapped to. Take such a flat $U$. Since $U \in \im(\varphi)$, we can write $U = c' + \spa(u_1, \ldots, u_{\ell})$, where the first $\ell$ coordinates of $u_i$ is equal to $e_i$ and the first $\ell$ coordinates of $c'$ are all $0$. There is a unique linear combination of the $u_i$'s and $c'$ that sums to $v$, and looking at the first $\ell$ coordinates it must be the case that $v = c' + a_1 u_1 + \cdots + a_{\ell} u_{\ell}$ where $a_i$ is the $i$th coordinate of $a$. Let $M$ be the matrix whose $i$th column is given by the last $n$ entries of $u_i$ and let $c$ be the last $n$ entries of $c'$. It follows that $U = \varphi((M, c))$ and that $Ma = b$.

zoom-outs: Take an arbitrary $T = (M, c) \in \C_{a^T, b, \beta}$ and suppose it has columns $v_1, \ldots, v_{\ell}$. Then $\varphi(T) = c' + \spa(v'_0, \ldots, v'_{\ell})$ where $c' = [0, c]^T$ and $v'_i = [e_i, v_i]^T$ for $1 \leq 1 \leq \ell$. Let $h = [-b^T, a^T]$ be the length $\ell + n$ row vector whose first $\ell$ entries are $b^T$ and last $n$ entries are $a^T$. By construction $\langle h, v'_i \rangle = 0$ and $\langle h, c \rangle = \beta$. Therefore, $\varphi(T)$ is contained in the hyperplane given $\langle h, x \rangle = \beta$. Since $T$ was arbitrary, it follows that $\varphi(\C_{a^T, b, \beta}) \subseteq W$, where $W$ is the hyperplane given by $\langle h, x \rangle = \beta$.

For the other direction of the bijection, take any $z + V \in \D_{W, \lin} \cap \im(\varphi)$, let $(M, c)$ be its preimage under $\varphi$, and suppose $W$ is given by $\langle h, x \rangle = 0$ (note that $W$ must be a linear subspace since it contains $V$ which is a linear subspace). Take $b$ to be the negative of the first $\ell$ coordinates of $h$ and $a$ to be the last $n$ coordinates of $h$. It follows that $a^T M = b$.
\end{proof}
With these lemmas, we can obtain the desired characterization of poorly expanding setes in $\AffShort(n, \ell).$
\begin{proof}[Proof of Theorem~\ref{th: pseudorandom}]
    Suppose $\Sc$ satisfies the conditions of Theorem~\ref{th: pseudorandom} and let $\Ac = \varphi(\Sc) \subseteq \AffGras(n+\ell, \ell)$. Henceforth use $\Phi$ and $\mu$ to denote expansion and measure in $\AffGras(n+\ell, \ell)$. By Lemma~\ref{lm: expansion preserved} we have that $1-\Phi(\Ac) \geq \frac{1}{q}-\frac{1}{q^2}$. By Lemma~\ref{lm: nice sets preserved}, we have that $\Ac$ is also $\xi$-pseudorandom with respect to hyperplanes, hyperplanes on its linear part, and points on its linear part. Finally since $\varphi$ is an injection, it is clear that $\mu(\Ac) \leq \xi$. By Lemma~\ref{th: pseudo aff gras} it follows that there exists a point $v$ such that 
    \[
    \frac{|\Ac \cap \D_{v}|}{|\D_{v}|}\geq 1- q^2\left(4q^{-\ell} + 867\xi^{1/4} + \frac{1}{q^{100}(q-1)} \right).
    \]
    By Lemma~\ref{lm: nice sets preserved}, there is a zoom-in (in the affine bi-linear scheme graph) $\C_{a,b}$ such that $\varphi^{-1}(\Ac \cap \D_{v}) = \Sc \cap \C_{a,b}$, so $|\Sc \cap \C_{a,b}| = | \Ac \cap \D_v|$. Finally $\frac{|\C_{a,b}|}{|\D_v|}$ is the probability that a randomly chosen flat from $\D_v$ is not full rank when restricted to its first $\ell$ coordinates. This probability is at most $1-1/q$. Indeed a flat from $\D_v$ can be chosen by choosing $\ell$ linearly independent vectors $v_1, \ldots, v_{\ell}$ and taking the flat to be $v + \spa(v_1, \ldots, v_{\ell})$. Conditioned on the first $\ell-1$ vectors being chosen so that their restriction to the first $\ell$ coordinates is linearly independent, there is a $1-1/q$ chance that $v_{\ell}$ to satisfy this property as well. It follows that
    \[
     \frac{|\Sc \cap \C_{a,b}|}{|\C_{a,b}|}\geq \frac{q}{q-1} \left(1- q^2\left(4q^{-\ell} + 867\xi^{1/4} + \frac{1}{q^{100}(q-1)} \right)\right) \geq 1 - \frac{1}{(q-1)^2} - \frac{q^3}{q-1} \left(4q^{-\ell} + 867\xi^{1/4} \right).
    \]
\end{proof}

\section{Canonical Monomials characterizing $\RM[n,q,d]$}\label{sec:prob_1_accept}
In this section we include the proof of Theorem~\ref{thm:perfect_accept}, showing that any polynomial passing the test with probability $1$ is degree $d$. The proof is implicit in ~\cite{RZS}, but as our tester is presented differently we give a proof for completion.

Write $d+1 = \ell \cdot p (q-q/p) + r$, where $1 \leq r \leq p(q-q/p)$, and set $s = \ell \cdot p$. First, note that our tester detects monomials of the form:
\begin{equation} \label{eqn: canonical mon}
    \left( \prod_{i=1}^{s} x_i^{q-q/p} \right) \prod_{j=1}^{p+2} x_{s+j}^{e_j},
\end{equation}
for any $(e_1, \ldots, e_{t})$ such that $\sum_{j=1}^{t}e_j \geq r$. Indeed, for $(e_1, \ldots, e_{t})$ such that $\sum_{j=1}^{t}e_j \geq r$, let $e' = (q-1-e_1, \ldots, q-1-e_{t})$ and consider the expansion of $H_{e'}(x_1, \ldots, x_{s+t})$. Using Lucas's Theorem, it is not hard to see that the expansion of $H_{e'}(x_1, \ldots, x_{s+t})$ contains the monomial $\left( \prod_{i=1}^{s} x_i^{q/p-1} \right) \prod_{j=1}^{t} x_{s+j}^{q-1-e_j}$, and hence by Lemma~\ref{lm: inner product}, any canonical monomial in \eqref{eqn: canonical mon} is rejected.

Let $\mathcal{G}$ be the family of functions that pass the sparse $s+t$-flat test with probability $1$. In this section we will prove:

\begin{lemma} \label{lm: local characterization}
    If $f$ passes the sparse-$s+t$-flat test with probability $1$, then $f \in \RM[n,q,d]$. Equivalently, $\mathcal{G} = \RM[n,q,d]$
\end{lemma}
One direction of this lemma is clear. If $\deg(f) \leq d$ then $f$ passes with probability $1$. The other direction amounts to showing that $\mathcal{G} \subseteq \RM[n,q,d]$. To show this, we will actually show the contrapositive and prove that if $\deg(f) \geq d+1$, and $f \in \mathcal{G}$, then one of the canonical monomials in \eqref{eqn: canonical mon} must be in $\mathcal{G}$. This is a contradiction and establishes that  $\mathcal{G} \subseteq \RM[n,q,d]$.

Before proceeding to the proof, we will need the following two facts from~\cite{KS} about affine-invariant families of polynomials. For a family of polynomials $\mathcal{F}$, let $\supp({\mathcal{F}})$ denote the set of monomials that appear in at least one of these polynomials. The first fact says that these monomials are a basis of $\mathcal{F}$.

\begin{lemma} [Monomial Extraction Lemma \cite{KS}] \label{lm: monomial extraction}
    If $\mathcal{F}$ is an affine-invariant family of polynomials then $\mathcal{F} = \spa(\supp(\mathcal{F}))$.
\end{lemma}

The second fact says that the monomials appearing in an affine-invariant family are $p$-shadow closed. For two integers $a, b \in \{0, \ldots, q-1 \}$, let $a = \sum_{i=0}^{k-1}p^ia_i$ and $b = \sum_{i=0}^{k-1}p^ib_i$ be their base $p$ representations (recall $q = p^k$). Then we say $a$ is in the $p$-shadow of $b$ if $a_i \leq b_i$ for $i = 0, \ldots, k-1$, and denote this by $a \leq_p b$. Then for two exponent vectors $e = (e_1, \ldots, e_n)$ and $e' = (e'_1, \ldots, e'_n)$, we say $e \leq_p e'$ if $e_i \leq_p e'_i$ for every $i$. Affine-invariant families of polynomials have the following shadow closed property.

\begin{lemma} \label{lm: shadow closed}
    Let $\mathcal{F} = \spa(\supp(\mathcal{F}))$ be an affine invariant family of polynomials. If $e \leq_p e'$ and $e' \in \supp(\mathcal{F})$, then $e \in \supp(\mathcal{F})$ as well.
\end{lemma}

Finally, we will need the following which will allow us to go from one monomial in $\mathcal{F}$ to another with the same degree, but with the distribution of the individual degrees shifted.

\begin{lemma} \label{lm: monomial shift}
    Suppose $x^e \in \mathcal{F}$ and suppose $m \leq_p e_2$. Then the monomial $x^{e'}\in \mathcal{F}$, where $e' = (e_1 + m, e_2-m, e_3, \ldots, e_n)$.
\end{lemma}
\begin{proof}
    Let $T$ be the affine transformation $T(x) = (x_1, x_1+x_2, x_3, \ldots, x_n)$. Then, 
    \[
    x^e \circ T = x_1^{e_1}(x_1+x_2)^{e_2} \prod_{j=3}^n x_j^{e_j} = \left(\sum_{i=0}^{d_2} \binom{d_2}{i}x_1^{e_1+i}x_2^{e_2-i} \right)  \prod_{j=3}^n x_j^{e_j}.
    \]
    By Lucas's Theorem and the assumption $m \leq_p e_2$, $\binom{d_2}{m} \neq 0$ in $\Ff_q$, and so $x^e \circ T = x_1^{e_1}(x_1+x_2)^{e_2} \prod_{j=3}^n x_j^{e_j} $ contains the monomial $x^{e'}$. The result then follows from Lemma~\ref{lm: monomial extraction}.
\end{proof}

With these three lemmas, we are ready to proceed to the proof of Lemma~\ref{lm: local characterization}. Supposing for the sake of contradiction that $\deg(f) > d$, using the above lemmas, we will show that this implies one of the canonical monomials is in $\mathcal{F}$. This cannot happen however, as all monomials in~\eqref{eqn: canonical mon} are rejected by $T = I$, the identity.

\begin{proof} [Proof of  Lemma~\ref{lm: local characterization}.]
   Let $\mathcal{F}$ be the family of polynomials that pass the sparse $s+t$-flat test with probability $1$. Suppose for the sake of contradiction that $f \in \mathcal{F}$ and $f$ contains a monomial of degree $d' > d$. Let $x^e = \prod_{i=1}^nx_i^{e_i}$ be such monomial and let $\ell$ be the smallest index such that $e_1 + \cdots + e_{\ell} > d$. Then $(e_1, \ldots, e_{\ell}, 0, \ldots, 0) \leq_p e$, so $\prod_{i=1}^{\ell} x_i^{e_i} \in \mathcal{F}$ and
   \[
    d+1 \leq \sum_{i=1}^{\ell} e_i \leq d+q-1 = s(q-q/p) + r + q-1.
   \]
   
   We will show that this monomial will lead to one of the canonical monomials being in $\mathcal{F}$. First, we claim that there is a monomial $x^{e'}$ such that $e'_i \geq q-q/p$ for $i = 1, \ldots, s$. Define
   \[
   c(e) = \sum_{i=1}^{s} \max(0, (q-q/p)-e_i).
   \]
   If $c(e) = 0$, then $e$ is of the desired form. If $e$ is not of the desired form, then there is $i \leq s$ such that $e_i < q-q/p$. Otherwise one of the following must be true:
   \begin{itemize}
       \item There is $j > s$ such that $e_j > 0$, in which case we simply find some $p^{m} \leq_p e_j$ and apply Lemma~\ref{lm: monomial shift} to obtain the monomial $x_i^{e_i+p^m}x_j^{e_j - p^m}$ in place of $x_i^{e_i}x_j^{e_j}$,
       \item There is $j \leq s$ such that $e_j > q-q/p$. In this case we can find $p^m$ such that $e_j - p^m \geq q-q/p$, and apply Lemma~\ref{lm: monomial shift} to obtain the monomial $x_i^{e_i+p^m}x_j^{e_j - p^m}$ in place of $x_i^{e_i}x_j^{e_j}$.
   \end{itemize}
In either case, we can find another monomial $x^{e'}$ such that $x^{e'} \in \mathcal{F}$ and $c(e') < c(e)$. Iterating this process, we must eventually find the desired monomial with $c(e) = 0$. 

Now, abusing notation, let $x^e$ be this monomial. We have, $d+1 \leq \sum_{i=1}^n e_i \leq d+q-1 = s(q-q/p) + r + q-1$ and $e_i \geq q-q/p$ for $1 \leq i \leq s$. We can now perform essentially the same argument and move any degree above $q-q/p$ to $e_{s+1}, \ldots, e_{s+t}$. There is at most 
\[
d+q-1 - s(q-q/p) = r + q-1 \leq p(q-q/p) + q-1,
\]
degree leftover after subtracting so we can perform the above argument to shift degree until each of the degrees $e_{s+1}, \ldots, e_{s+t-1}$ is at least $q-q/p$. This will leave at most $q-1$ degree leftover, which we can simply shift to $e_{s+t}$.
\end{proof}

\begin{lemma} \label{lm: invertible reject}
    Suppose $g: \Ff_q^{p} \xrightarrow[]{} \Ff_q$ satisfies $\deg(g) \geq q(p-1)$. Then there exist full rank affine transformations $T_1, T_2 \in \T_{p,p}$ such that:
    \begin{enumerate}
        \item $g \circ T_1$ contains a monomial $\prod_{i=1}^p x_i^{e_i}$ with $q-q/p \leq e_i \leq q-1$ for all $i \in [p]$.
        \item $ \langle g \circ T_2, P \rangle \neq 0$.
    \end{enumerate}
\end{lemma}
\begin{proof}
    To show the first part of the lemma, we use a similar idea to the proof of Lemma~\ref{lm: monomial shift}. Suppose $g$ contains a monomial $\prod_{i=1}^{p}x_i^{e'_i}$ and suppose $m \leq_p e_2'$. Consider the full rank affine transformation $T_{\alpha}(x) = (x_1, \alpha x_1 + x_2, x_3, \ldots, x_p)$ for $\alpha \in \Ff_q$ and the coefficient of $x_1^{e_1'+m}x_2^{e_2'-m}\prod_{i=3}^{p}x_i^{e_i'}$ in $g \circ T_{\alpha}$. It is not hard to see that this coefficient is a nonzero polynomial in $\alpha$ and therefore there exists an $\alpha$ such that $g \circ T_{\alpha}$ contains the monomial $x_1^{e_1'+m}x_2^{e_2'-m}\prod_{i=3}^{p}x_i^{e_i'}$. As $T_{\alpha}$ is a full rank affine transformation, we can repeatedly apply such transformations to obtain a full rank $T_1 \in \T_{p,p}$ such that $g \circ T_1$ contains a monomial $\prod_{i=1}^p x_i^{e_i}$ with $q-q/p < e_i \leq q-1$, proving the first part of the lemma.

    Let $T_1$ be the transformation from part $1$ and let $g' = g \circ T_1$, so that $g'$ contains a monomial $\prod_{i=1}^p x_i^{e_i}$ with $q-q/p < e_i \leq q-1$. For $a \in \Ff_q^{p}$, let $T_{a}$ be the full rank affine transformation such that $T_a(x) = (x_1 + a_1, \ldots, x_p + a_p)$ and consider the inner product
    \[
    \langle g' \circ T_a, P \rangle,
    \]
    as a polynomial in $a$. Recall that $P$ contains the monomial $\prod_{i=1}^p x_i^{q/p-1}$.
    Since $g'$ contains the monomial $\prod_{i=1}^p x_i^{e_i}$ with $q-q/p < e_i \leq q-1$ there is a contribution of $\prod_{i=1}^{p} a_i^{e_i - (q-q/p)}$ with nonzero coefficient, and this cannot be cancelled out from any other monomial. Therefore, $ \langle g' \circ T_a, P \rangle$ is a nonzero polynomial in $a$ and there exists an $a \in \Ff_q^{p}$ such that $ \langle g' \circ T_a, P \rangle \neq 0$. This establishes the second part of the lemma.
\end{proof}

\section{Proof of Lemma~\ref{lm: localization}} \label{sec: self-correct}
In this section we provide the proof of Lemma~\ref{lm: localization}. This is essentially the same as ~\cite[Proposition 3.5]{KM} with a slight modification.  Recall that $g: \Ff_q^{\ell+100}$ is an arbitrary polynomial, $d$ is an arbitrary degree parameter, and $b \in \Ff_q^{\ell+100}$ is an arbitrary point. We work in $\AffGras(\ell+100, \ell)$. Use $\nu$ to denote the uniform measure over all $\ell$-flats, and $\nu_b$ to denote uniform measure over the zoom-in on $b$. When referring to zoom-ins, zoom-ins on the linear part etc.\ we are referring to the versions in the affine Grassmann graph. We have the following two sets of $\ell$-flats,

\[
\mathcal{A} = \{U \; | \;  \dim(U) = \ell, \deg(g|_U) > d, b \in U\},
\]
\[
\mathcal{B} = \{U \; | \;  \dim(U) = \ell, \deg(g|_U) > d, b \notin U\},
\]
and $\nu(\mathcal{B}) \leq q^{M/2} O(\epsilon)$. We will show the following weaker form of Lemma~\ref{lm: localization}, and then explain why this gives the full Lemma~\ref{lm: localization}.
\begin{lemma} \label{lm: localization KM}
    If $\nu(\mathcal{B}) \leq q^{M/2}O(\epsilon)$ then $\mathcal{B} = \emptyset$. Moreover there is a value $\gamma$ such that after changing $f(b)$ to $\gamma$, $\nu_b(\mathcal{A}) \leq 1- \frac{1}{q}$.
\end{lemma}
Before proving this lemma, we explain why it implies Lemma~\ref{lm: localization}. Assuming that Lemma~\ref{lm: localization KM} holds, suppose that $g'$ is the resulting function after changing the value of $g(b)$ to $\gamma$, and consider the set $\U$ of $\ell$-flats $U$ such that $\deg(g'|_U) > d$. Assume for the sake of contradiction that $\U$ is nonempty. We record some facts about $\U$,
\begin{enumerate}
    \item $\U \subseteq \D_b = \{U \; | \; \dim(U) = \ell, b \in U \}$.
    \item $\nu_V(\U) \leq q^{-100}$.
    \item $\Phi(\U) \leq 1 - \frac{1}{q}$.
    \item $\nu_V(\U)$ is $q^{-99}$ pseudorandom with respect to zoom-outs and zoom-outs on the linear part.
    \item $\nu_V(\U)$ is $q^{-100}$ pseudorandom with respect to zoom-ins on the linear part.
\end{enumerate}
The first item follows from assuming Lemma~ \ref{lm: localization KM}.
The second item follows from the first item and from bounding the fraction of $\ell$-flats that contain $b$. The third and fourth items follow from the same arguments as Lemmas~\ref{lm: expansion specific} and \ref{lm: zoom-out}. Finally, the fifth item follows again from the first item and the fact that over any zoom-in on the linear part, a random $\ell$-flat from this set contains the point $b$ with probability at most $q^{-100}$.

Applying Theorem~\ref{th: pseudo aff gras} with $\xi = q^{-100}$, we get that $\U$ must have density at least $1 - q^2(867q^{-25} + q^{-\ell}) > 1 - 1/q$ inside some zoom-in, where we use the fact that $\ell \geq 4$. There can only be one zoom-in on which $\U$ has nonzero density, and that is the zoom-in on $b$. However, this contradicts the assumption that $\nu_b(\mathcal{A}) \leq 1- \frac{1}{q}$ after the value of $g(b)$ is changed. In short, we have shown that, under this setting, if a correction causes the rejection probability within a zoom-in to drop below $1-1/q$, then the rejection probability must actually be $0$.

We now give the proof of Lemma~\ref{lm: localization KM}. This proof requires the following result used in both \cite{BBKSS, KM}, which was referred to as the shadow lemma in the introduction.
\begin{lemma} \label{lm: uppershadow 2}
Let $d \in \mathbb{N}$ be a degree parameter and $k \geq \ceil{\frac{d+1}{q-q/p}}$. Then for any $f: \Ff_q ^{k+1} \xrightarrow[]{} \Ff_q$:
\begin{enumerate}
    \item If $\deg(f) > d$, then $\epsilon_{k,d}(f) \geq \frac{1}{q}$.
    \item If $d < \deg(f) < (k+1)(q-1)$ then $\epsilon_{k,d}(f) > \frac{1}{q}$.
\end{enumerate}
Where $\epsilon_{k,d}(f)$ and $\epsilon_{k+1,d}(f)$ are the fraction of $k$-flats and $k+1$-flats respectively on which the restriction of $f$ has degree greater than $d$.
\end{lemma}
 In \cite{KM} this lemma is used to show that the set of rejecting flats is non-expanding, similar to Lemma~\ref{lm: expansion specific} in this paper. The lemma has another use though, which is stated in its second item. For our purposes, the second item says that if $f|_{V}$ has degree greater than $d$ and for greater than $1/q$ of the hyperplanes $U \subseteq V$ we have $\deg(f|_{U}) > d$ as well, then $f|_{V}$ cannot contain the maximum degree monomial. In this section we will use Lemma~\ref{lm: uppershadow 2} with the parameters $d = d$ and $k = \ell$. Note that $\ell \geq \lceil \frac{d+1}{q-q/p} \rceil$ by assumption.

Our first goal is to show $\mathcal{B} = \emptyset$. Start by taking an $(\ell+40)$-flat $W$ uniformly conditioned on $b \notin W$. Let
\[
\mathcal{B}_W = \{B \subseteq W \; | \; B \in \mathcal{B} \},
\]
and work in $\AffGras(W, \ell)$ --- the affine Grassmann graph over the $\ell$-flats contained in $W$. 

Let $\mu_W$ and $\Phi_W$ denote measure and expansion respectively in $\AffGras(W, \ell)$. We claim that either $\mu_W(\mathcal{B}_W) = 0$ or $\mu_W(\mathcal{B}_W) \geq q^{-100}$. Otherwise, $0 < \mu_W(\mathcal{B}_W) < q^{-100}$ and $\mathcal{B}_W$ is $q^{-60}$ pseudo-random with respect to zoom-ins and zoom-ins on the linear part. By a similar argument to Lemma~\ref{lm: zoom-out} we have that $\mathcal{B}_W$ is $q^{-98}$ pseudorandom with respect to zoom-outs and zoom-outs on the linear part and by a similar argument to Lemma~\ref{lm: expansion general} we have that $1- \Phi_W(\mathcal{B}_W) \geq 1/q$. This contradicts Theorem~\ref{th: pseudo aff gras}, however, so we conclude that either $\mu_W(\mathcal{B}_W) = 0$ or $\mu_W(\mathcal{B}_W) \geq q^{-100}$. 

\begin{lemma}
    $\mathcal{B} = \emptyset.$
\end{lemma}
\begin{proof}
    Otherwise we may find a $W$ such that $\mu_W(\mathcal{B}_W) \geq q^{-100}$. Fix such a $W$ and sample a uniform $\ell+99$-flat $Y$ conditioned on $b \notin Y$, and a uniform $\ell+60$-flat $A_2 \subseteq Y$. Now consider $A_2 \cap W$. We may think of $W$ as being defined by a system of $60$ independent linear equations $\langle h_1, x \rangle = c_1, \ldots, \langle h_{60}, x \rangle = c_{60}$. That is, $W$ is the subspace of $\Ff_q^{\ell+100}$ that satisfies these $60$ equations. Likewise, $A_2$ is given by the restriction of $39$ linear equations,  $\langle h'_1, x \rangle = c'_1, \ldots, \langle h'_{39}, x \rangle = c'_{39}$. The probability that all $99$ linear equations are linearly independent is at least
\begin{equation*}
    \prod_{j=0}^{38} \frac{q^{99}-q^{60+j}}{q^{99}} \geq e^{-2 \sum_{j=1}^{\infty} q^{-j}} \geq e^{-4/q},
\end{equation*}
and $A_2 \cap W$ is uniform over $\AffGras(W, \ell)$. Thus, 
\[
\Pr[A_2 \cap W \in \mathcal{B}_W] \geq e^{-4/q}q^{-100}.
\]
If $A_2 \cap W \in \mathcal{B}_W$, then $A_2 \in \mathcal{B}_Y^{\uparrow^{60}}$, where the upper shadow is taken with respect to $\AffGras(Y, \ell)$, so it follows that
\[
\E_Y[\mu_Y(\mathcal{B}_Y^{\uparrow^{60}})] \geq e^{-4/q}q^{-100}.
\]

On the other hand, by Lemma~\ref{lm: uppershadow 2}, 
\[
\E_Y[\mu_Y(\mathcal{B}_Y^{\uparrow^{60}})]\leq q^{60}\E_Y[\mu_Y(\mathcal{B}_Y)] \leq 2q^{60} \mu(\mathcal{B}_V) \leq 2q^{60+M/2} O(\epsilon).
\]
By assumption $\epsilon < q^{-M}$ so altogether we get that,
\[
q^{-M/2} \geq \frac{1}{C(q)q^{-160}},
\]
for some constant $C(q)$. For large enough $M$ this is a contradiction.  
\end{proof}

We now complete the proof of Lemma~\ref{lm: localization KM}, which in turn proves Lemma~\ref{lm: localization}.
\begin{proof}[Proof of Lemma~\ref{lm: localization}.]
    Fix an $\ell$-flat $U$ that contains $b$. We will show that there exists a value $\gamma$ such that changing the value of $g(b)$ to $\gamma$ results in a degree $r$ function on $U$. Establishing this fact and choosing the most common $\gamma$ over all $\ell$-flats proves Lemma~\ref{lm: localization KM}, which in turn establishes Lemma~\ref{lm: localization} by our previous discussion.
    
    Suppose $\deg(g|_{U}) >d$, as otherwise we are done by setting $\gamma = g(b)$. Pick an $\ell+1$-flat $U' \supset U$ and let $g' = g|_{U'}$, and $M(x) = 1_{x \neq b}$, over $U'$. Note that $M(x)$ contains the degree $(\ell+1)(q-1)$-monomial, so there is some value $\gamma$ such that $\gamma M(x) + g'(x)$ has degree strictly less than $(\ell+1)(q-1)$. Showing that for this $\gamma$, $\deg(\gamma M(x) + g'(x)) \leq d$ completes the proof.
    
    Since we have already shown that, $\deg(g|_{U'}) \leq d$ for any $\ell$-flat that does not contain the point $b$, this implies that $\gamma M(x) + g'(x)$ also has degree at most $d$ on any $\ell$-flat not containing $b$. This is because when restricted to such flats, $\gamma M(x) + g'(x)$ is equal to the restriction of $g$ plus some constant polynomial. Since the degree of $g$'s restriction is at most $d$, so is the degree of $\gamma M(x) + g'(x)$. It follows that $\gamma M(x) + g'(x)$ can only have degree greater than $d$ when restricted to $\ell$-flats that contain $b$. This is at most a $1/q$-fraction of $\ell$-flats contained in $U'$ however, which, along with the fact that $\deg(\gamma M(x) + g'(x)) < (\ell+1)(q-1)$, implies that $\deg(\gamma M(x) + g'(x)) \leq d$ by Lemma~\ref{lm: uppershadow 2}.
\end{proof}
\end{document}